%% file: main_arxiv.tex
\documentclass[12pt]{article}
\usepackage{amsmath, amssymb}
\usepackage{graphicx}
\usepackage{enumerate}
\usepackage{natbib}
\usepackage{overpic}
\usepackage{subfiles}
\usepackage{hyperref}
\usepackage{cleveref}

\newcommand{\blind}{1}

\usepackage{enumitem}
\usepackage{mathrsfs, booktabs, amsthm, amsfonts, bbold}
\usepackage{accents}

\usepackage{xcolor}
\protected\def\[#1\]{\begin{equation}\begin{aligned}#1\end{aligned}\end{equation}}
\protected\def\(#1\){\begin{equation*}\begin{aligned}#1\end{aligned}\end{equation*}}

\newtheorem{theorem}{Theorem}

\theoremstyle{remark}
\newtheorem{remark}{Remark}

\theoremstyle{definition}

\newtheorem{example}{Example}

\usepackage{subcaption}
\usepackage{float}

\usepackage{comment}

\theoremstyle{definition}

\begin{document}

\def\spacingset#1{\renewcommand{\baselinestretch}
{#1}\small\normalsize} 
\spacingset{1}

\if1\blind
{
  \title{\bf Bayesian Distance-to-Set Models: \\ from Latent Variable to Latent Projection}
  \author{
	Leo Duan \thanks{\href{email:li.duan@ufl.edu}{li.duan@ufl.edu}}\\
	Department of Statistics, University of Florida\\
  AND\\
Yuexi Wang \thanks{\href{email:yxwang99@illinois.edu
}{yxwang99@illinois.edu}}
\\
Department of Statistics, University of Illinois Urbana-Champaign\\
AND\\
  Jason Xu\thanks{\href{email:jqxu@g.ucla.edu}{jqxu@g.ucla.edu}}\\
  Department of Biostatistics, University of California, Los Angeles
    }
    \date{}
  \maketitle
  \vspace{-0.5in}
} \fi

\if0\blind
{
  \bigskip
  \bigskip
  \bigskip
  \begin{center}
    {\LARGE\bf }
\end{center}
  \medskip
} \fi

\bigskip

\begin{abstract}
    Statistical models often assume that data are generated near a structured, smooth,
    or low-dimensional set. A common approach is to use Bayesian latent variable models,
    in which each observation is associated with a latent coordinate on the set, and the
    observed data are modeled as noisy deviations from these coordinates. The deviation
    is typically characterized by a location-scale distribution, such as Gaussian.
    Despite their intuitive appeal and popularity, latent variable models often present
    practical challenges in posterior computation. In particular, Markov chain Monte Carlo
    samplers may suffer from slow mixing, especially when the sample size is large and there
    is no closed form for integrating out the latent coordinates. In this article, we propose
    an alternative approach that replaces the deviation-from-coordinate with a distance-to-set.
    Specifically, the distance-to-set is defined as the distance between a data point and its
    projection onto the set, where the projection can be rapidly computed by optimization and
    replaces the latent coordinate in the likelihood. This change substantially reduces the
    dimensionality of the parameter $\times$ latent variable space, leading to efficient
    posterior computation. We establish several important statistical properties for the
    distance-to-set models, such as the independence between the normal-cone noise and fixed-effect
    parameters, posterior consistency, and an Occam's razor effect that automatically penalizes
    overfitting. We demonstrate the effectiveness of our approach through simulation studies, applications to
    multi-environment study and Bayesian transfer learning.
\\
The implementation source code for the models in this article can be found at \\
\url{https://github.com/leoduan/distance-to-set-models}.
\end{abstract}
\vfill

\addtolength{\textheight}{-0.8in}
\addtolength{\topmargin}{0.5in}

\newpage

\spacingset{1.8}

\subfile{main_text}

\end{document}

%% file: main_text.tex
\section{Introduction}

Latent variable models are a central tool in statistics and machine learning, with canonical examples including finite mixture models, random-effects models, network models and factor models \citep{fraley2002model,hoff2002latent,dunson2006bayesian}. These models introduce latent variables to capture certain structures in the data generating mechanism. Formally, let
$\mathcal Y$ denote the observation space and
$\mathcal Z$ denote the latent variable space. We denote by $y=(y_1,\ldots,y_n)\in \mathcal Y^n$ the observed data, by $\theta$ the parameter of interest, and by $z = (z_1,\dots,z_n) \in \mathcal Z^n$ the latent variables, where both $y_i$ and $z_i$ may be multivariate.
The marginal likelihood for $y$, under a latent variable model, is typically given by
\[\label{eq:canonical_augmented_likelihood}
\mathcal L(y\mid \theta) = \int \mathcal L(y\mid \theta, z) \pi_{\mathcal L}(z\mid \theta) \text{d}z.
\]
where $\pi_{\mathcal L}$ is the latent variable distribution, and the corresponding augmented likelihood (or complete-data likelihood) is $\mathcal L(y, z\mid \theta) := \mathcal L(y\mid \theta, z) \pi_{\mathcal L}(z\mid \theta)$. Note that we do not consider $\pi_{\mathcal L}(z\mid \theta)$ as a prior distribution for $z$, but as part of the likelihood that describes the unobserved data $z$.

The framework in \eqref{eq:canonical_augmented_likelihood} admits a two-layer generative process interpretation: latent variables are first drawn from $\pi_{\mathcal L}(z\mid \theta)$, and observations are subsequently generated from $\mathcal L(y\mid \theta, z)$. This viewpoint is often useful for model specification and interpretation.
In particular, when both $z_i$ and $y_i$ are continuous, which is also the main focus of this work, one may think of $\mathcal Z$ as a structured, smooth or low-dimensional set, and each $z_i$ as a coordinate on that set. The corresponding $y_i$ is then modeled as a stochastic deviation from $z_i$, often through a location-scale model for $\mathcal L(y\mid \theta, z)$, such as $y_i \stackrel{\text{indep}}{\sim} \text{N}( g(z_i), I\sigma)$ with $g$ a deterministic function and $\sigma>0$. Such a formulation appears in a wide range of applications, including
hierarchical latent variable models \citep{blei2014build},  Gaussian process latent variable models \citep{lawrence2005probabilistic,titsias2010bayesian} and locally linear  models \citep{park2015bayesian}, probabilistic principal surface models \citep{chang2001principal}, latent Gaussian models \citep{rue2009approximate}, functional manifold models \citep{chen2012nonlinear,lila2016smooth}, and probabilistic  factorization models \citep{mnih2008probabilistic}, to name a few.

For statistical inference, quantifying the uncertainty of the parameter $\theta$ is a key objective. In the Bayesian setting, this is typically done by sampling from the posterior distribution $\Pi(\theta\mid y)$. In most cases, however, the integral in \eqref{eq:canonical_augmented_likelihood} does not have a closed form, and one needs to resort to data augmented Markov chain Monte Carlo (MCMC) that targets the joint posterior $\Pi(\theta, z \mid y)$ \citep{tanner1987calculation,van2001art}. A chief challenge is that the joint posterior becomes supported in $\Theta \times \mathcal Z^n$, whose dimension grows with the sample size $n$ even when $\Theta$ itself is low-dimensional.
The high dimensionality introduced by the augmented space often leads to slow mixing of MCMC algorithms and affects the quality of the posterior estimate.

Several remedies have been proposed to address these computational issues. First, for specific classes of latent variable models, one may design customized MCMC algorithms with substantially improved mixing performance \citep{papaspiliopoulos2007general,duan2018scaling,rajaratnam2019uncertainty,zens2024ultimate}. Second, when the primary interest lies in  $\pi(z\mid y)$ and the latent variable distribution $\Pi_{\mathcal L}$ is Gaussian, one may approximate the posterior $\Pi(\theta\mid y)$ for instance by Laplace approximation, leading to tractable approximate marginal $\Pi(z\mid y)$ via integrated nested Laplace approximation (INLA) \citep{rue2009approximate}. Third, when inference focuses on $\Pi(z\mid y, \hat\theta)$ at some value of $\hat\theta$, amortized variational inference \citep{kingma2013auto} provides a scalable alternative by employing a flexible parameterization (such as a neural network) $\tilde\Pi(z\mid y, \hat\theta, \gamma)$ and optimizing $(\hat\theta, \hat\gamma)$ to minimize a statistical divergence between the latent and variational distributions.

While effective in many applications, these approaches are often either model-specific, requiring bespoke approaches, or rely on approximations that may lead to spurious inference. This motivates an alternative formulation that preserves the geometric structure of latent variable models while yielding a marginal $\mathcal L(y\mid \theta)$ that is tractable (or close-to-tractable).  The proposed approach also facilitates direct posterior inference for $\theta$ via $\Pi(\theta \mid y)$, while still supporting interpretable low-dimensional representations of the data through $\Pi(z_i\mid y_i, \theta)$.

In this work, we show that such an ideal scenario is in fact achievable with broad generality, by modifying how we view the deviation of $y_i$ from the latent space $\mathcal Z$. Instead of modeling the deviation between $y_i$ and a certain coordinate $z_i$, for instance by some location-scale model, we model the distance between $y_i$ and $\mathcal Z$ directly, as shown in \Cref{fig:generative_view}.

\begin{figure}[H]
  \centering
  \begin{subfigure}[t]{0.48\textwidth}
      \centering
      \includegraphics[width=\textwidth]{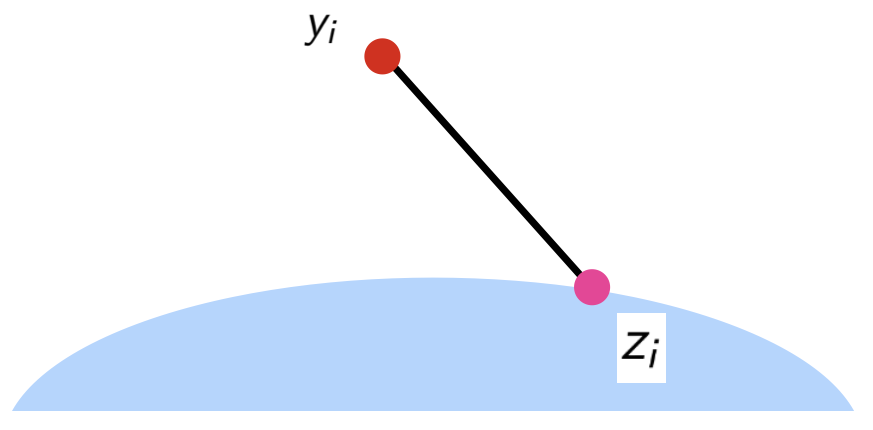}
      \caption{Canonical latent variable model: $z_i$ is drawn on $\mathcal Z_\theta$, then an unconstrained deviation $\epsilon_i$ is added to form $y_i=z_i+\epsilon_i$.}
  \end{subfigure}
  \hfill
  \begin{subfigure}[t]{0.48\textwidth}
      \centering
      \includegraphics[width=\textwidth]{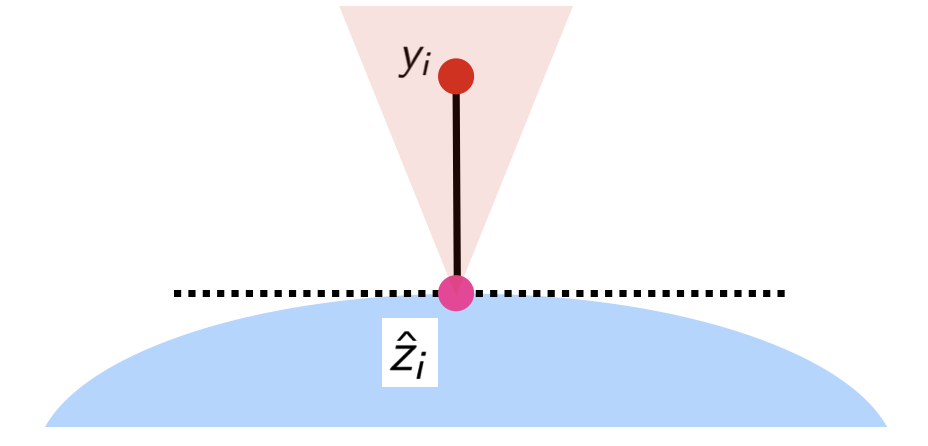}
      \caption{Distance-to-set model: $\hat z_i$ is drawn on $\mathcal P_{\mathcal Z_\theta}(\mathcal Y)$, then a deviation $\epsilon_i$ is drawn in a constrained space (red shadowed area, normal cone if $\mathcal Z_\theta$ is convex), so that $ \hat z_i + \epsilon_i$ is projectable to $\hat z_i$.}
  \end{subfigure}
  \caption{Comparison of the generative views for the canonical latent variable model and the distance-to-set model.  The former considers $z_i$ drawn from $\mathcal Z_\theta$, and $\epsilon_i$ often unconstrained, while the latter assumes $\hat z_i$ drawn from the subset $\mathcal P_{\mathcal Z_\theta}(\mathcal Y)\subseteq\mathcal Z_\theta$, and $\epsilon_i$ is drawn in a geometry-induced displacement space (more details in \Cref{sec:generative_view}).} \label{fig:generative_view}
\end{figure}

The distance from a point to a set is the shortest distance between $y_i$ and any point in $\mathcal Z$. This treatment eliminates the need to specify a latent distribution $\pi_L(z\mid\theta)$, thereby allowing inference to proceed directly through the marginal likelihood. At the same time, the \textit{projection} of $y_i$ onto the set $\mathcal{Z}$, denoted by $\hat{z}_i$, arises naturally as the minimizer in the distance calculation. This quantity admits a clear interpretation as a low-dimensional representation of $y_i$ on $\mathcal{Z}$, can be computed efficiently for fixed $\theta$ in many scenarios of statistical interest, and inherits a posterior distribution through the uncertainty in $\theta$. In this sense, the proposed framework retains the interpretability of latent variable models while avoiding their reliance on explicit latent coordinates.

\section{Method}

\subsection{Distance-to-set model}
Suppose each data point $y_i$ lies in the ambient space $\mathcal Y \subset \mathbb{R}^d$, and is generated near a latent space $\mathcal Z_\theta$. We consider the following likelihood $\mathcal L$, which characterizes how $y_i$ deviates from $\mathcal Z_\theta$:
\[\label{eq:distance_to_set_model}
  \text{dist}(y_i, \mathcal Z_\theta) = \min_{z_i\in \mathcal Z_\theta} \|y_i - z_i\|, \qquad
  \mathcal L(y_i ; \theta, \sigma) =   \phi [ \text{dist}(y_i, \mathcal Z_\theta); \sigma]
\]
where $\mathcal Z_\theta$ is a closed set so that the minimum is attainable. Here $\phi(\text{dist}; \sigma)$ is a kernel function that is monotonically non-increasing in the distance, and $\sigma>0$ is a scalar that controls the rate at which the likelihood decays as $y_i$ moves away from the set $\mathcal Z_\theta$. In this paper, we focus on the following Gaussian distance kernel:
\[\label{eq:gaussian_kernel}
\phi(\text{dist}; \sigma) = m_{\theta,\sigma}^{-1} \exp(-\text{dist}^2/\sigma)
\]
where $m_{\theta,\sigma}$ is the normalizing constant. We refer to $\phi$ as a distance {\em kernel}, since the above likelihood is still interpreted as a probability density/mass function for $y_i$, rather than for the distance itself. The distance \eqref{eq:distance_to_set_model} may be generated via any norm $\|\cdot\|$; for simplicity we primarily focus on the Euclidean norm $\|\cdot\|_2$ with $(y_1,\ldots,y_n)$ i.i.d. from \eqref{eq:distance_to_set_model}. One can generalize to the non-iid setting by using a dependent distance kernel or a non-separable distance between $(y_1,\ldots,y_n)$ and $ \times_{i=1}^n \mathcal  Z^{(i)}_\theta$. We discuss these extensions in \Cref{sec:theory}.

To illustrate these model concepts, we now turn to a simple example involving a translated $\ell_1$-ball.
\begin{example}[\bf Multivariate mixed-effects model with sparse random effect]\label{ex:multi_mix_effects}
Consider the multivariate model
$$y_i = \mu + \gamma_i + \epsilon_i, \qquad y_i\in \mathbb{R}^d, \, i=1,\ldots,n $$
where $\mu\in \mathbb{R}^d$ is the parameter of interest and $\gamma_i\in \mathbb{R}^d$ is a random effect for the $i$-th data point, which we expect to be sparse. Typically, $\mu$ may have additional structure such as covariate dependence; our focus in this exposition is on inducing sparsity in $\gamma_i$.

We define the following latent space:
\[\label{eq:sparse_latent_space}
\mathcal Z_\theta = \bigg \{ \mu + \gamma_i:  \sum_{j=1}^d |\gamma_{ij}| \le r \bigg \},
\]
where the subscript $\theta=(\mu, r)$ denotes the parameter that determines the latent space, and $r>0$ denotes the radius around this set. Effectively, changing the random effect $\gamma_i$ is akin to changing coordinate $z_i$ on the set $\mathcal Z_\theta$.

In this case, the Euclidean projection of  $y_i$ onto $\mathcal Z_\theta$ is available in closed form through soft-thresholding. Denote $\hat z_i= \mu+\hat\gamma_i$: we have
\[
& \text{dist}(y_i, \mathcal Z_\theta) = \|y_i - (\mu + \hat\gamma_i)\|_2,\\
&  \hat\gamma_{ij} = \operatorname{sign}(y_{ij}-\mu_j)\,\max\bigl(|y_{ij}-\mu_j|-\lambda_i,\,0\bigr),
\]
where $\lambda_i\ge 0$ is chosen such that $\sum_{j=1}^d |\hat\gamma_{ij} |  = r$ when $y_i\not\in \mathcal Z_\theta$, and $\lambda_i=0$ when $y_i\in \mathcal Z_\theta$. The resulting joint likelihood is $\mathcal L(y; \theta, \sigma) \propto \prod_{i=1}^n\sigma^{-d/2}\exp( - \text{dist}^2(y_i, \mathcal Z_\theta)/\sigma)$, where $\mathcal Z_\theta$ does not depend on $y_i$. Specifically, the squared distances add up in the exponents $\sum_{i=1}^n (y_i-\hat z_i)^2$, with $\hat z_i = \mu +  \hat\gamma_i$ and $ \hat\gamma_i$ is a sparse vector function of $y_i$ and $\theta$.

Under a conventional latent-variable model, each observation-specific latent variable $\gamma_i$ would need to be sampled  jointly with $\mu$ and $r$. Here, by contrast, the sparse representation $\hat\gamma_i$ is obtained directly by projection, which is available in closed form. The likelihood therefore depends on $(\mu,r,\sigma)$ through the projected distances, without requiring high-dimensional latent-variable augmentation.

\hfill $\blacksquare$
\end{example}

Generalizing from the last example, we highlight three useful features of this distance-to-set formulation. First, it bypasses the need to sample latent coordinates. Instead, for each observation it computes a projection:
\[
  \hat z_i \in \mathcal Z_\theta: \|\hat z_i-y_i\|=\text{dist}(y_i, \mathcal Z_\theta),
\]
which we denote by $\hat z_i \in \mathcal P_{\mathcal Z_\theta}(y_i)$.
 Even when the projection is not unique, any element $\hat z_i \in \mathcal P_{\mathcal Z_\theta}(y_i)$ would suffice for evaluating the likelihood. Thus, \eqref{eq:distance_to_set_model} is already a marginal likelihood of $y_i$ given $(\theta, \sigma)$, whereas $\hat z_i$ is only an auxiliary variable arising from the distance calculation.

Second, despite their auxiliary nature, $\hat z_i$ can be interpreted as an embedding for $y_i$ --- a model-based smoothing of $y_i$ in the space of $\mathcal Z_\theta$. Therefore, we refer to $\hat z_i$ as a {\em latent projection}. Once uncertainty in $\theta$ is incorporated, $\hat z_i$ and $\mathcal P_{\mathcal Z_\theta}(y_i)$ also become stochastic.

Third, the projection only depends on the specification of $\mathcal Z_\theta$ and the norm that defines the distance. This gives flexibility in choosing the distance kernel $\phi$, while leaving projection  $\mathcal P_{\mathcal Z_\theta}(y_i)$ unchanged. For example, one could use the set \eqref{eq:sparse_latent_space} for inducing sparsity in the latent projection, while adopting a richer hierarchical model for other parts of the data-generating mechanism.

\subsection{Generative view for the latent projection and its deviation}\label{sec:generative_view}
We now complete a Bayesian perspective by introducing a generative view for the latent projection and its deviation. The purpose of this construction is to clarify the probabilistic meaning of the projection and the constrained residual structure induced by the distance-to-set formulation.

Suppose ($\theta, \sigma$) are generated from a prior distribution $\Pi_0(\theta, \sigma)$, and $\theta$ determines the space $\mathcal Z_\theta$. Let $\mathcal P_{\mathcal Z_\theta}(\mathcal Y)$ denote the image of the observation space under the projection onto $\mathcal Z_\theta$. Consider  the hierarchical data generating process
\[\label{eq:two_layer_generative_process}
 (\hat z_i\mid \theta,\sigma) \sim q(\cdot; \theta,\sigma)  \quad
   (\epsilon_i \mid \hat z_i,\sigma) \sim h(\cdot; \sigma, \hat z_i), \quad
   y_i = \hat z_i + \epsilon_i.
\]
where $q$ is a proper distribution supported on $P_{\mathcal Z_\theta}(\mathcal Y)$, and $h$ is a proper distribution supported on the constrained space $\{v: \hat z_i \in \mathcal P_{\mathcal Z_\theta}(\hat z_i+v)\}$. That is, $\hat z_i+\epsilon_i$ must project  back to $\hat z_i$.

This construction in \eqref{eq:two_layer_generative_process} defines a joint model of $(y_i, \hat z_i)$ and is more general than the distance-to-set model \eqref{eq:distance_to_set_model}, since it  not only specifies the distribution of $y_i$, but also assigns a proper distribution for $\hat z_i \in \mathcal P_{\mathcal Z_\theta}(y_i)$. On the other hand, the distance-to-set model is recovered as the special case when
	$q(\hat z_i; \theta,\sigma) h(v_i;\sigma,\hat z_i)\propto \phi (\|\epsilon_i\|; \sigma) J(\hat z_i, \epsilon_i)$, where $J$ is the Jacobian factor for the transformation $ (\hat z_i,\epsilon_i) \mapsto y_i$.

To better understand the constrained spaces of $\hat z_i$ and $v_i$, we now focus on the special case where $\mathcal Z_\theta$ is convex. Under Euclidean distance, projection is unique, so $\mathcal P_{\mathcal Z_\theta}(y_i)$ is always a singleton. If $\mathcal Y \cap \mathcal Z_\theta = \varnothing$, then the projection image $\mathcal P_{\mathcal Z_\theta}(\mathcal Y)$ lies on the boundary of $(\mathcal Z_\theta)$, and $\hat z_i$ is a random boundary point. Moreover, the deviation $\epsilon_i$ lies in the normal cone of $\mathcal Z_\theta$ at $\hat z_i$:
\[\label{eq:normal_cone}
  \mathcal N_{\mathcal Z_\theta}(\hat z_i) = \left\{ v\in \mathbb{R}^d: v^\top (\zeta -\hat z_i) \le 0\; \forall \zeta\in \mathcal Z_\theta \right\},
\]
In geometric terms, $\epsilon_i$ points {\em outward} from $\mathcal Z_\theta$ at $\hat z_i$.

For an intuitive example, consider the case in which $\mathcal Z_\theta = \{z: \|z -c \|_2\le r\}$ is a Euclidean ball with center $c\in\mathbb{R}^d$ and radius $r>0$. Its boundary is $\{z:\|z-c\|_2=r\}$, and the normal cone at $\hat z_i\in \text{bd}(\mathcal Z_\theta)$ is $\{ \alpha (\hat z_i-c): \alpha\ge 0\}$, namely the collection of all normal vectors orthogonal to the tangent plane at $\hat z_i$ \citep{rockafellar1997convex}. When the boundary is smooth like in the above example,
the tangent plane at each $\hat z_i$ is unique, and the normal cone corresponds to normal vectors. For non-smooth boundary points (such as the vertices of a polyhedron), the normal cone defined in \eqref{eq:normal_cone} provides a useful generalization.

\addtocounter{example}{-1}
\begin{example}[\bf Multivariate mixed-effects model with sparse random effect continued]
In the previous example, for a given $\hat z_i = \mu +\gamma_i$, when $\sum_j |\gamma_{ij}|=r$,   the normal cone at $\hat z_i$ is
\(
\mathcal N_{\mathcal Z_\theta}(\hat z_i)
=
\Bigl\{\alpha s:\alpha\ge 0,\
s_j=
\,\mathrm{sign}(\gamma_{ij}), \text{ if } \gamma_{ij}\neq 0,\;
s_j\in[-1,1], \text{ if } \gamma_{ij}=0\;
\Bigr\}.
\)

 Essentially the distance-to-set model induces   a dependent and constrained structure on the deviation vector $\epsilon_i = y_i-\hat z_i$. Here, the normal-cone restriction reveals which patterns of residual variation are compatible with a given sparse projection $\gamma_i$: (i) the zero random effect at index $k$ should not have  deviation magnitude larger than that of a non-zero random effect $|\epsilon_{i,k}|\le |\epsilon_{i,j}|$ at index $j$, and (ii) non-zero random effects share a common deviation magnitude.
Constraint (i) can also intuitively be understood as an identifiability condition: the non-zero random effect can only be identified when the difference between the observed and systematic component $\mu_j$ is large enough.

\hfill $\blacksquare$
\end{example}

\subsection{Independence between set-shift parameter and  deviation}\label{sec:set_parameterization}
In a canonical statistical model $y_i = \mu_i + \epsilon_i$, the error $\epsilon_i$ is often assumed to be independent of the location parameter $\mu_i$ (e.g. fixed effects in regression). One may wonder if an analogous property holds for the distance-to-set model, where the deviation is constrained by the geometry of the set.
Toward understanding this, let the residuals belong to some set $\Gamma_r$ where $r>0$ is a scale parameter on the radius $\text{rad}(\Gamma_r)=r\; \text{rad}(\Gamma_1)$, where $\text{rad}(C)= \inf_{c}\sup_{z\in C} \|z-c\|$. Consider an additive model structure defined $ \mathcal Z_{\theta}^n = \{ [f(x_i,\beta) + \gamma_i]_{i=1}^n:  \gamma\in \Gamma_r\} : = \mathcal M_{\beta,r}$, and data
\(
  y_i = f(x_i,\beta) + \hat \gamma_i + \epsilon_i, \qquad i=1,\ldots,n,
\)
where $f(x_i,\beta)$ is a \textit{set-shift} parameter that translates the set, akin to the location parameter. Here, the deviation $\epsilon_i = y_i - f(x_i,\beta)-\hat \gamma_i$ plays the role of the error, and is constrained to $V_{\mathcal Z_\theta}({\hat z_i})= \{  v\in \mathbb{R}^d:   \hat z_i\in \mathcal P_{\mathcal Z_\theta}(\hat z_i+v) \}$, which we will call the displacement space.

As $\hat z_i=f(x_i,\beta)+\hat \gamma_i$, our question is made precise by asking whether the dependence of $\mathcal Z_\theta$ on $\beta$ induces dependency between $\beta$ and the deviation $\epsilon_i$. For convex sets under Euclidean distance, the answer is no.
By shift invariance of the normal cone,
$$\mathcal{N}_{\mathcal Z_{\beta,r}}(\hat z_i)
=\mathcal{N}_{\mathcal Z_{\beta,r} - f(x_i,\beta)}[\hat z_i-f(x_i,\beta)]=\mathcal{N}_{\Gamma_r}(\gamma_i);
$$
that is, the admissible deviation space depends only on the set geometry encoded by $\Gamma_r$, not on the shift implied by $f(x_i,\beta)$. Thus the deviation remains independent of the set-shift parameter, even when $f$ is nonlinear.
The following theorem generalizes this to a broader class of projections and to non-convex sets.

\begin{theorem}
Let $\mathcal P_{\mathcal Z_\theta}:\mathbb R^d \mapsto \mathcal Z_\theta$ be a (possibly set-valued) projection rule onto a nonempty set $\mathcal Z_\theta\subset\mathbb R^d$.
For $z\in \mathcal Z_\theta$, define the displacement set
$
V_{\mathcal Z_\theta}(z) := \{\, v\in\mathbb R^d \;:\; z\in \mathcal P_{\mathcal Z_\theta}(z+v) \,\}.
$
For any $b\in\mathbb R^d$,
$
V_{\mathcal Z_\theta +b}(z+b) = V_{\mathcal Z_\theta}(z)\quad \text{for all } z\in \mathcal Z_\theta
$
if and only if the projection rule is translation--equivariant:
$$
\mathcal P_{\mathcal Z_\theta +b}(y+b) = \mathcal P_{\mathcal Z_\theta}(y) + b \qquad \text{for all } y,b\in\mathbb R^d .
$$
\end{theorem}

\begin{remark}
The above theorem is stated for projection rules, but the same invariance property holds more generally for functions $\mathcal T:\mathbb R^d \to \mathcal Z_\theta$ that are translation-equivariant and single-valued.
\end{remark}

\subsection{Normalizing constant and implicit set size regularization}
The distance-to-set likelihood assigns higher density to observations that lie closer to $\mathcal Z_\theta$. Since given a set of observations, distance decreases as the set expands, a larger set will generally provide a better fit
\(
\mathcal{Z}_{\theta_1} \subseteq \mathcal{Z}_{\theta_2} \implies \operatorname{dist}(y_i, \mathcal{Z}_{\theta_1}) \geq \operatorname{dist}(y_i, \mathcal{Z}_{\theta_2}).
\)
As a consequence, larger sets will always be favored in the posterior if not for the
contribution $m^{-1}_{\theta,\sigma}$. Thus, controlling the set size is essential.

Fortunately, there is an implicit regularization effect of the set size in the posterior \eqref{eq:joint_posterior}.
Intuitively, as the set $\mathcal Z_\theta$ expands in \eqref{eq:two_layer_generative_process}, the projection image $\mathcal P_{\mathcal Z_\theta}(\mathcal Y)$ also tends to expand, making the induced distribution of projection points $q(\cdot; \theta,\sigma)$ more diffuse. As a result, the marginal density $\sum_{z_i\in \mathcal P_{\mathcal Z_\theta}(y_i)}h( \epsilon_i; \sigma, \hat z_i)q(\hat z_i; \theta,\sigma)$ would be relatively smaller under more diffuse $q(\hat z_i; \theta,\sigma)$, and in turn $Z_\theta$ is weighed unfavorably in the posterior, reminiscent of the Bayesian version of Occam's razor \citep{jefferys1992ockham}. On the other hand, since the distribution $q$ in the generative view of \eqref{eq:two_layer_generative_process} is often intractable, this motivates a direct analysis of the same regularization effect through the normalizing constant.

To this end, consider the alternative two-layer generative process: (i) first draw a distance magnitude $t_i>0$ from a distribution, which defines a level set $S_{\theta,t_i} := \{y:\mathrm{dist}(y,\mathcal Z_\theta)=t_i\}$, and then  (ii) draw a point $y_i$ from $S_{\theta,t_i}$ uniformly. The two densities involved are:
\(
\tilde q(t_i; \sigma) = \frac{\exp(-t_i^2/\sigma)\,A(S_{\theta,t_i})}{m_{\theta,\sigma}}, \qquad
\tilde h(y_i; \sigma, t_i) = \frac{1(y_i\in S_{\theta,t_i})}{A(S_{\theta,t_i})}.
\)
The first is a half-Gaussian exponentially tilted by the area of the level set denoted as $A(S_{\theta,t_i})$, and enables us to compute the normalizing constant:
\[\label{eq:normalizing_constant}
m_{\theta,\sigma}= \int_{t>0} {\exp(-t^2/\sigma)\,A(S_{\theta,t})} dt.
\]
This expression makes the regularization mechanism transparent.
The level set $S_{\theta,t}$ is the boundary of the closed $t$-neighborhood $\mathcal Z^t_{\theta}=\{\tilde z = z + v, z\in \mathcal Z_{\theta}, \|v\|\le t \}$, and   as $\mathcal Z_{\theta}$ expands, we see that $\mathcal Z^t_{\theta}$ expands too for each $t>0$. For convex sets, the boundary area also increases monotonically in turn, and we have $\mathcal Z_{\theta_1}\subseteq \mathcal Z_{\theta_2} \implies A(S_{\theta_1,t})\le A(S_{\theta_2,t})\implies m^{-1}_{\theta,\sigma}\ge m^{-1}_{\theta_2,\sigma}$. That is, the normalizing constant $m_{\theta,\sigma}$ would increase and its inverse acts as a stronger penalty on larger sets.
For non-convex sets, this monotonicity may not hold as separate components can emerge as
$t$ increases. Nonetheless, since the volume $V(\mathcal Z_{\theta,t})$ is nondecreasing and locally Lipschitz for $t>0$, it is differentiable for almost every $t > 0$ so that  $A(S_{\theta,t})=\partial V(\mathcal Z_{\theta,t})/\partial t$ gives us a way to quantify the constant \eqref{eq:normalizing_constant}.

We now return our attention to the case when the distance is defined using Euclidean norm, for which the volume can be computed explicitly.

\begin{theorem}
  For a convex set $\mathcal Z_{\theta}$ and distance defined using Euclidean norm, the normalizing constant admits the form:
  \(
  m_{\theta,\sigma} = \sum_{k=0}^{d-1} \pi^{(d-k)/2} \,V_k(\mathcal Z_{\theta})\,
\,\sigma^{(d-k)/2}.
  \)
  where $V_k(\mathcal Z_{\theta})>0$ is the $k$-th dimensional volume for $\mathcal Z_{\theta}$. For bounded $\mathcal Z_{\theta}$ with non-empty interior, $V_k(\mathcal Z_{\theta})\in (0,\infty)$ for all $k=0,1,\ldots,d-1$.
\end{theorem}

Finally, we remark that $V_k(\mathcal Z_{\theta})$ can be computed explicitly or numerically for simple convex sets such as vector-norm balls, while more complicated sets can be handled as well owing to the shift-invariance and scaling laws of volume. For example, recall the additive model considered in the previous section $\mathcal Z_{\theta}=\{z\in\mathbb R^d: f(x_i,\beta)+ \gamma_i:  \gamma_i\in \Gamma_r\}$. We know for such a set:
\[\label{eq:normalizing_constant_scaling}
  m_{\theta,\sigma} = \sum_{k=0}^{d-1} \pi^{(d-k)/2} \,V_k( \Gamma_1)\,
\, r^k\sigma^{(d-k)/2}
= (\pi \sigma)^{d/2} \sum_{k=0}^{d-1}  \,V_k( \Gamma_1)\,
\,  (\frac{r}{\sqrt{\pi\sigma}})^k
.
  \]
where we use the scaling law of volumes $V_k(r \Gamma_1)= V_k(\Gamma_1) r^k$.
This example shows explicitly that $m^{-1}_{\theta,\sigma}$ does penalize large sets by penalizing large $r$. In particular, when $r> \sqrt{\pi\sigma}$, the terms with larger $k$ can dominate in the polynomial, leading to  accelerated decay in $m^{-1}_{\theta,\sigma}$.

\subsection{Prior specification strategy}\label{sec:prior_specification}
For data generated i.i.d. from \eqref{eq:distance_to_set_model}, the posterior can be written as:
\[\label{eq:joint_posterior}
  \Pi(\theta, \sigma \mid y) = \frac{ \Pi_0(\theta, \sigma) m_{\theta,\sigma}^{-n}  \exp( - \sum_{i=1}^n \text{dist}^2 (y_i, \mathcal Z_{\theta})/\sigma)}{\int \Pi_0(\theta, \sigma) m_{\theta,\sigma}^{-n} \exp( - \sum_{i=1}^n \text{dist}^2 (y_i, \mathcal Z_{\theta})/\sigma) d\theta d\sigma}.
\]
A notable feature is the explicit appearance of the normalizing constant $m_{\theta,\sigma}$, which provides implicit regularization over the set size as discussed in the previous section.

In practice, one could use a proper prior for $\theta$ and $\sigma$ that reflects the prior knowledge about the set geometry and the noise level.
In this paper, we use a normal prior $\text{N}(0,10^2)$ for unconstrained parameters.  For positive parameters, we posit an inverse-Gaussian prior $\text{InvGaussian}(a,b)$ \(\Pi_0(s; a,b)\propto s^{-3/2}\exp\{- \frac{b}{2a^2}s - \frac{b}{2s} \}\)
which has prior mean $a$ and variance $a^3/b$.  This prior form is convenient because its two exponential terms balance two competing effects: shrinkage toward smaller values of $s$ and protection of $s$ from becoming singular.
In this paper,  we assign $\text{InvGaussian}(1, 1)$ for $\sigma$ throughout. For models with a radius parameter $r$, \eqref{eq:normalizing_constant_scaling}  suggests that $r$ scales with $\sqrt{\sigma}$, hence we assign the prior $r\sim \text{InvGaussian}(0.1 \sqrt{\sigma}, 0.1 \sqrt{\sigma})$.

\begin{remark}
In the Supplementary Material, we discuss how to generalize these distances to broader classes of divergences via majorization.
\end{remark}

\section{Computation}
These models are amenable to Markov chain Monte Carlo (MCMC) for posterior computation. When $m_{\theta,\sigma}$ is tractable, and $\theta$ is low-dimensional, the Metropolis-Hastings algorithm is a good baseline due to its simplicity. Starting with the current parameter $(\theta, \sigma)$ and $\text{dist}(y_i, \mathcal Z_\theta)$, each Metropolis-Hastings iteration proceeds as follows:

\begin{itemize}
  \item Propose a new parameter value $(\theta^\prime, \sigma^\prime)$ from a proposal distribution $\mathcal K( \theta^\prime, \sigma^\prime \mid \theta, \sigma)$.
  \item Run optimization subroutine to find projections $\hat z_i^\prime$ given $(\theta^\prime, y_i)$,
  and compute $\text{dist}(y_i, \mathcal Z_{\theta^\prime})$ for $i=1, \ldots,n$.
  \item Accept the proposal $(\theta',\sigma')$ with probability
  \[
 \min\left(1, \frac{\Pi(\theta^\prime, \sigma^\prime \mid y)
    \mathcal K( \theta, \sigma \mid \theta^\prime, \sigma^\prime)
    }{\Pi(\theta, \sigma \mid y)
    \mathcal K( \theta^\prime, \sigma^\prime \mid \theta, \sigma)
    }\right),
  \]
  otherwise keep the current parameter $(\theta, \sigma)$.
\end{itemize}

When $\theta$  is moderate- or high-dimensional, random-walk Metropolis-Hastings often mixes poorly. In such settings, gradient-based MCMC algorithms such as  the Metropolis-Hastings algorithm with a Barker proposal (MH-Barker) \citep{livingstone2022barker} and the No-U-Turn Sampler (NUTS) \citep{hoffman2014no} are often preferred. The Barker proposal offers lower computational cost per iteration and mixes well in moderate dimensions, whereas NUTS incurs a slightly higher per-iteration cost but typically achieves superior performance in higher-dimensional settings. In our implementation, we use the \texttt{numpyro} package \citep{bingham2019pyro} for most of the posterior computations.

We also remark that our distance-based formulation is well positioned to benefit from tools from convex analysis such as duality \citep{rockafellar1997convex}. The projection onto a set is equivalently viewed as the proximal operator of the indicator function of that set \citep{polson2015proximal}, enabling tools such as Moreau's decomposition for well-known proximal operators \citep{parikh2014proximal}. Due to length considerations and the technical detail of optimization concepts, we detail in the Supplemental Material how duality can significantly simplify three complicating facets that may arise in using distance-to-set models.  First, repeated projection onto $\mathcal Z_\theta$  may be expensive when a closed-form expression is unavailable. Second, when gradient-based algorithms are used, efficient evaluation of the gradient is essential. Third, we discuss how to sample from the model's predictive distribution.

\section{Posterior consistency on set parameter recovery}\label{sec:theory}
In this section, we conduct a theoretical analysis establishing the sufficient conditions for posterior consistency, in particular for recovery of the set parameter $\theta$. There is a large literature on posterior consistency \citep{schwartz1965bayes,barron1999consistency,walker2007rates,ghosal2000convergence}, including cases where the data are generated in an independent, but not identically distributed setting \citep{ghosh2003bayesian,ghosal2007convergence}. In our setting, however, a primary challenge is that the likelihood induced by the distance-to-set model may not be factorizable over the data index $i$. As a result, standard consistency conditions  in the canonical setting do not apply directly, and the role of the normalizing constant must be handled explicitly.

To facilitate analysis, we focus on the following setup.
Let $y^{(n)} \in \mathbb{R}^{n d}$ denote the stacked data vector, and let $\mathcal{M}^{(n)}_{\theta} = \times_{i=1}^n \mathcal Z_\theta^i \subset \mathbb{R}^{n d}$ denote the product model set indexed by $\theta \in \Theta$, where $\Theta$ is compact and the component sets  $\mathcal Z_\theta^i$ may potentially vary with the covariates. We consider a general distance between $y^{(n)}$ and $\mathcal{M}^{(n)}_{\theta}$,
\(
  \mathrm{dist}(y^{(n)}, \mathcal{Z}^{(n)}_{\theta}) = \min_{z\in \mathcal{M}^{(n)}_{\theta}} \|y^{(n)}-z\|,
\)
which includes $\mathrm{dist}^2(y^{(n)}, \mathcal{Z}^{(n)}_{\theta})=\sum_{i=1}^n \mathrm{dist}^2(y_i, \mathcal Z_\theta^i)$ as a special case. For simplicity, we treat $\sigma$ as a known constant, as is common in standard Bayesian asymptotic theory literature.
We also assume that the prior does not depend on $n$, so the posterior contains the normalizing factor $1/m^{n}_{\theta,\sigma}$.

Define the empirical and population risks:
\(
R_n(\theta) = \frac{1}{n} \, \mathrm{dist}^2(y^{(n)}, \mathcal{Z}^{(n)}_{\theta}),\qquad
R(\theta) = \limsup_{n \to \infty} \mathbb{E} \, R_n(\theta)
\)
where the expectation is taken under the data-generating mechanism $\mathbb{P}$. Assume that $R(\theta)$ is uniquely minimized at $\theta_0$. Our goal is to show that the posterior $\Pi(\theta \mid y^{(n)})$ concentrates near $\theta_0$ under regularity conditions as follows.
  \begin{itemize}
    \item[(A1)] There exists a metric $\rho$ on $\Theta$ and a constant $\kappa > 0$ such that
    $
        R(\theta) - R(\theta_0) \geq \kappa\, \rho^2(\theta, \theta_0)
    $
    $\forall \theta \in \Theta$.
    \item[(A2)]
    $
        \sup_{\theta \in \Theta} \left| R_n(\theta) - \mathbb{E} R_n(\theta) \right| \to 0
    $
    in probability as $n \to \infty$.
    \item[(A3)] For every $\varepsilon > 0$,
    $
        \Pi_0\left\{ \theta : R(\theta) - R(\theta_0) < \varepsilon \right\} > 0
   .$
   \item[(A4)] There exists  a  sequence $\varepsilon_n \downarrow 0$ and  constants $M_1>M > 0$,
$\Theta_{n}(M,M_1)=\{\theta:\ M\varepsilon_n \le \rho(\theta,\theta_0)\le M_1\varepsilon_n\}$, such that for $n$ large enough
for any $\theta\in \Theta_{n}(M,M_1)$,
 $
 \Delta C_n(\theta)=(\log m^{(n)}_{\theta,\sigma}-\log m^{(n)}_{\theta_0,\sigma})\ge -c_0\,n\epsilon_n^2$ for some constant $c_0>0$.
  \item[(A5)] For the same $\varepsilon_n$ in (A4), for all $n$ large enough, for any $\theta: R(\theta)-R(\theta_0)< \sigma\varepsilon_n^2$,
    $
  \big|\Delta C_n(\theta)\big|
    \le c_1\,n\varepsilon_n^2
    $ for some constant $c_1>0$.
  \item[(A6)] For the same $\varepsilon_n$ and $M_1$ in (A4), for all $n$ large enough, the prior mass
   $
\Pi_0\{\theta:\rho(\theta,\theta_0)> M_1\varepsilon_n\}
\ \le\ \exp\{-b\,n\,\varepsilon_n^2\}$ for some $b>2 c_1+ 6$.
\end{itemize}
With the above ingredients, we are ready to state the consistency result.
\begin{theorem}
\label{thm:consistency}
Under (A1)--(A3), for the likelihood given by \eqref{eq:distance_to_set_model} and \eqref{eq:gaussian_kernel},
\[
    \Pi\left( \theta : \rho(\theta, \theta_0) > M \epsilon_n \mid y^{(n)} \right) \underset{n\to \infty}{\stackrel{\mathbb P}{\longrightarrow}} 0
\]
for any sequence $\varepsilon_n \downarrow 0$ that obeys (A4)-(A6).
\end{theorem}

\begin{remark}
  Condition (A1) describes the separation; (A2) is the uniform law of large numbers; (A3) is commonly known as prior support; (A4-A5) ensure that the log-difference in normalizing constant gets dominated by the risk difference, in the shell $\Theta_n(M,M_1)$ and a small neighborhood of $\theta_0$ defined by $R(\theta)-R(\theta_0)$; (A6) ensures the prior assigns an exponentially small mass to the outside of the ball of radius $M_1\varepsilon_n$.
\end{remark}

We want to highlight the breadth of the above theorem. Since we define the distance between $y^{(n)}$ and $\mathcal{Z}^{(n)}_{\theta}$ in a general way, it may or may not be decomposable over the data index $i$. Because of this,
the underlying data-generating process can be dependent. The same argument also extends to the divergence-to-set models discussed in the Supplement, provided (A1)--(A6) are satisfied.
The generality is due to the use of a likelihood ratio bound that enables us to construct an exponentially consistent test. We provide the details in the Supplementary Materials.

\section{Empirical study using multi-environment study data}
In this section, we use a real dataset to benchmark the computational performance of our distance-to-set model. Additional simulated experiments and their details appear in the Supplementary Materials.

In many applications, researchers are interested in whether an effect is broadly generalizable across environments. Examples include policy evaluations across different regions \citep{andrews2022transfer}, educational interventions across different schools, and large-scale experiments across products or user segments.
In such problems, full invariance is often too restrictive, while environment-specific modeling without any constraint can obscure the shared structure.
 The distance-to-set framework provides a natural middle ground: it formulates generalizability as a structured neighborhood around a shared global parameter. In this sense, it is closely related to a mixed-effects formulation, where one decomposes an effect into a common component plus environment-specific, random departures. The key difference is interpretive: rather than treating those departures as random effects drawn from a distribution and summarizing this heterogeneity through a variance component, we model them as deviations from the exact-generalizability baseline, so that the radius parameter
$r$ measures how much separation from that benchmark is needed to sufficiently explain the data. Thus, $r$ serves as a direct and interpretable summary of the degree of generalizability.

We illustrate this idea using the Student/Teacher Achievement Ratio\footnote{\url{https://dataverse.harvard.edu/dataset.xhtml?persistentId=hdl:1902.1/10766}} (STAR) dataset \citep{achilles2008tennessee}. Project STAR was a longitudinal class-size experiment conducted in Tennessee, in which students were randomly assigned to small classes, regular classes, and regular classes with a full-time aide. The data contain student-level covariates, teacher characteristics, school characteristics, and test scores measured over multiple grades. This setting is useful for studying generalizability because the intervention is randomized and the outcome is comparable across schools. We are interested in the scientific question of whether the small-class effect is approximately stable across school contexts.

\subsection{Model specification}
We focus on the math scores and specify the model as
\begin{equation}\label{eq:star_model}
y_{jg} = \alpha_g + \eta^\top x_j + \phi^\top w_{jg} + (\beta+\delta_{s_{jg}}) D_{jg} + \varepsilon_{jg},
\end{equation}
where $y_{jg}$ is the math score for student $j$ in grade $g$, $\alpha_g$ is the grade-specific intercept, $x_j$ contains student covariates such as gender, ethnicity and birth cohort, $w_{jg}$ contains time-varying controls such as free-lunch status, teacher experience, teacher degree, teacher ethnicity, and school urbanicity.  The indicator $D_{jg}^{S}$  denotes assignment to a small class in grade $g$, and  $s_{jg} \in \{1, \ldots, 79\}$ denotes the school ID attached to student $j$ in grade $g$.

The parameter of primary interest is $\beta$, which captures the overall small-class effect across schools. The school-specific term $\delta_{s_{jg}}$ captures deviation from this global effect in school $s$.  When all $\delta_s$ are zero, the model reduces to a pooled specification in which the treatment effect is identical across schools. To allow for occasional deviation while preserving a dominant shared effect, we constrain the school-specific deviations to lie in an $\ell_1$-ball
\(
\mathcal M^{(n)}_{\theta=(\alpha_g, \eta, \phi, \beta, r)}=\Big\{ \big[\mu_{jg}+\delta_{S_{jg}} D_{jg} \big]_{ij}: \sum_{s=1}^{n_S}|\delta_s|\leq r \Big\}.
\)
where $\mu_{jg}=\alpha_g + \eta^\top x_j + \phi^\top w_{jg}+\beta D_{jg}$.

Here, projection onto the relevant set is given by a weighted soft-thresholding  \[
 \label{eq:weighted_soft_thresholding}
\hat\delta_s
=
\operatorname{sign}(\bar R_{s1})
\Big(|\bar R_{s1}|-\frac{\lambda}{n_{s1}}\Big)_+, \quad\bar R_{s1}:=\frac{\sum_{j,g:\,S_{jg}=s,\ D_{jg}=1} (y_{jg}-\mu_{jg})}{n_{s1}}, \quad n_{s1}=\sum_{j: S_{jg}=s} D_{jg}
\]
for some $\lambda>0$, chosen so that $\sum_{s=1}^{n_S}|\delta_s|= r$. In implementation,  it is convenient to reparameterize from $(\alpha_g, \eta, \phi, \beta, r)$ to $(\alpha_g, \eta, \phi, \beta, \lambda)$ using the change-of-variable trick, which avoids solving repeatedly for $\lambda$ given $r$ during posterior computation.

For comparison, we consider the four alternative models: (1) a pooled model with $\delta_s=0$ for all $s$ (Pooled); (2) a hierarchical model with $\delta_s \overset{iid}{\sim} \text{N}(0, \tau^2)$ (Standard); (3) a sparse hierarchical model with $\delta_s$ following a standard horseshoe prior \citep{carvalho2009handling} (Horseshoe); and (4) a sparse hierarchical model with $\delta_s$ following a continuous spike-and-slab prior \citep{george1993variable} (Spike \& Slab). We refer to our proposed method as the distance-to-set (DTS) model. For a fair comparison,  we use the same priors for $\alpha_g$, $\eta$, $\phi$, $\beta$ and $\sigma$, as suggested in \Cref{sec:prior_specification}, across all models whenever possible.

\subsection{Generalizability simulation study}

We first conduct a simulation study under a similar setting designed to mimic the STAR dataset, which allows us to vary the degree of school-specific heterogeneity in a controlled way. We consider a simplified base-level model without the grade effect and school-level covariates as $\mu_{j}=\alpha_j + \eta^T x_j$, where $x_j\in \mathbb{R}^2$ is generated from standard Normal. We set $\beta=2, \sigma=1, \eta =(1, -0.5)^T$. Each student has 50\% chance of being assigned to a small class.

To study generalizability, we vary the sparsity level of the true school-specific deviations $\delta_s$ from $10\%$ to $90\%$. \Cref{fig:star_simu_sparsity}  shows that the posterior behavior of the DTS model tracks this structure closely: as the true deviations become sparser, the posterior radius $r$ decreases and the number of near-zero projected effects $\hat\delta_s$ increases. This supports the interpretation of $r$ as a summary of transferability across schools.

To compare with alternative methods, here we provide a close look at the results when true $\delta_s$ are 80\% sparse. We summarize the posterior distribution of $\beta$ and the trace plot of the sparse $\delta_s$ values in \Cref{fig:star_simu_beta_post}. Unsurprisingly, the pooled regression model gives a biased estimate of $\beta$ because it cannot accommodate school-level departures. The hierarchical model gives a wider posterior range for $\beta$ than our method, likely due to the fact that it does not leverage the sparsity structure of $\delta_s$. The horseshoe, spike-and-slab and DTS models give similar posterior distribution for $\beta$. In terms of the sparsity pattern, we find that the spike-and-slab model gives a slightly higher sparsity, followed by the DTS model, while the horseshoe model has much lower sparsity.
\begin{figure}[!ht]
  \centering
  \begin{subfigure}[t]{0.31\textwidth}
      \includegraphics[width=\textwidth]{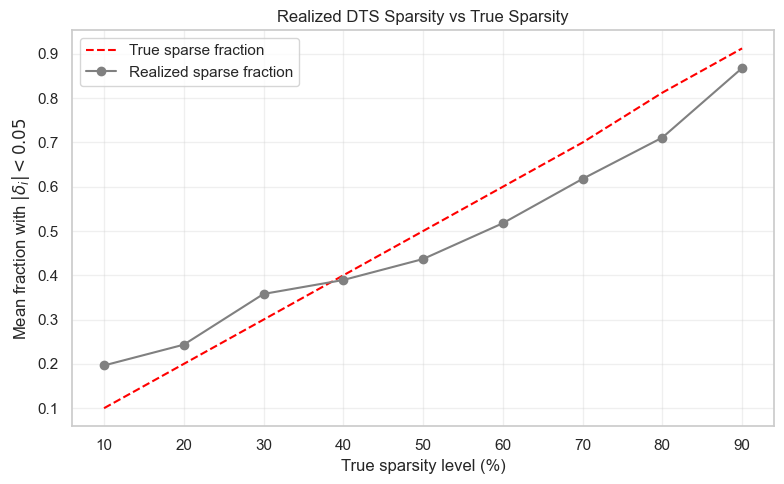}
      \caption{Realized DTS sparsity}\label{fig:star_simu_sparsity}
  \end{subfigure}
  \begin{subfigure}[t]{0.31\textwidth}
    \includegraphics[width=\textwidth]{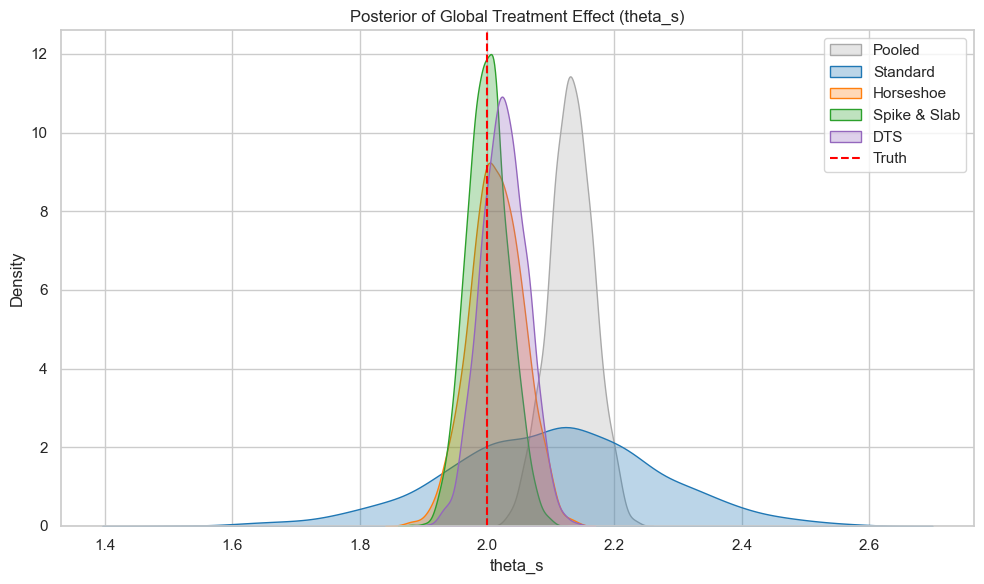}
    \caption{Posterior of $\beta$}
  \end{subfigure}
  \begin{subfigure}[t]{0.31\textwidth}
       \includegraphics[width=\textwidth,height=3cm]{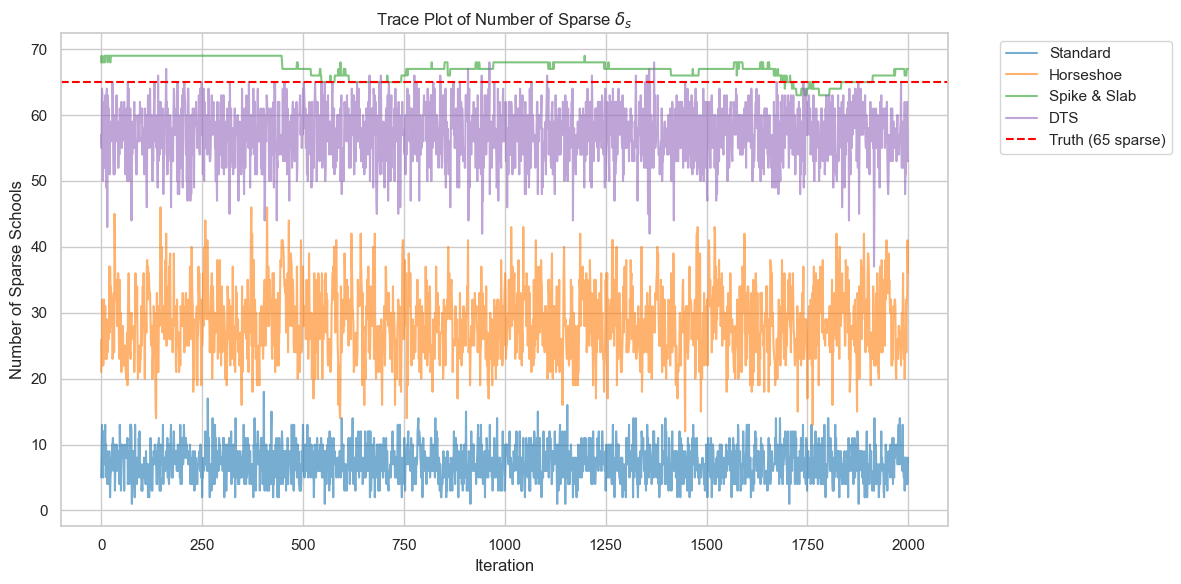}
       \caption{Trace of sparse $\#\{\delta_s: |\delta_s|<0.05\}$}
  \end{subfigure}
  \caption{Posterior of the global small-class effect $\beta$ and trace plot of sparse $\delta_s$'s. The red dashed lines indicate the true value of $\beta$ and the true sparsity level.}\label{fig:star_simu_beta_post}
\end{figure}

\Cref{fig:star_simu_deltas} provides a more detailed view of the school-specific deviations by showing posterior intervals for four schools with true $\delta_s=0$ and four schools with the largest nonzero $|\delta_s|$. The DTS model and spike-and-slab prior both exhibit stronger shrinkage toward zero for null schools while still recovering the larger deviations reasonably well. Overall, the simulation suggests that DTS provides an interpretable and computationally efficient way to characterize approximate generalizability.

\begin{figure}[H]
  \centering
  \begin{subfigure}[t]{\textwidth}
    \includegraphics[width=\textwidth,height=3cm]{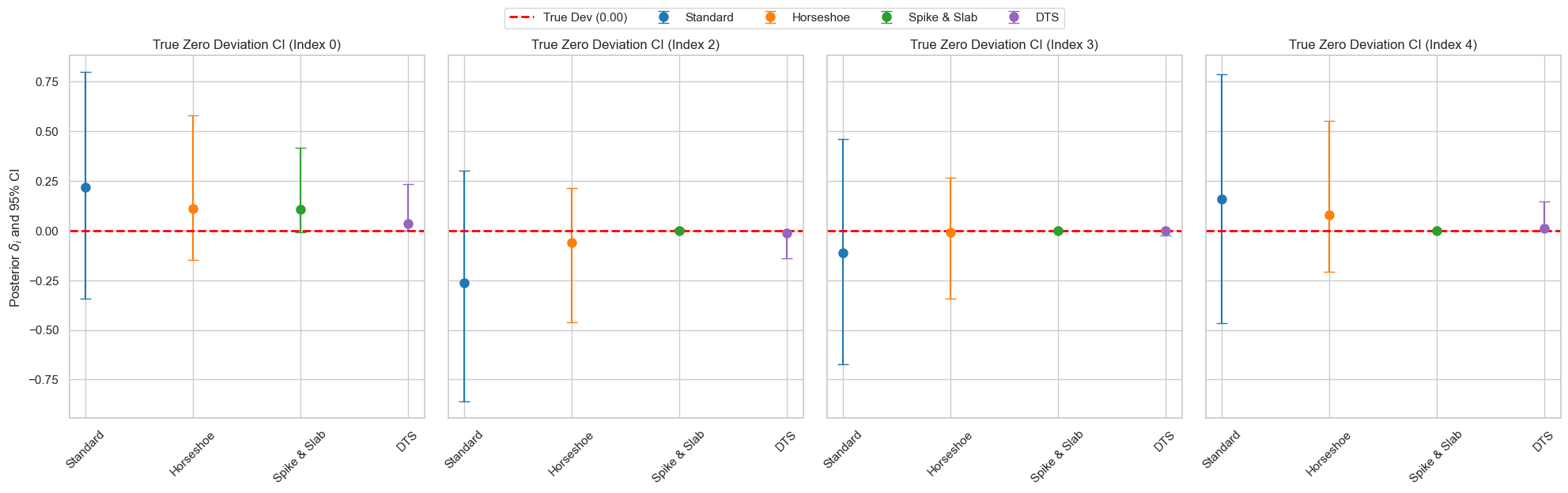}
    \caption{schools with $\delta_s=0$}
  \end{subfigure}
  \begin{subfigure}[t]{\textwidth}
    \includegraphics[width=\textwidth,height=3cm]{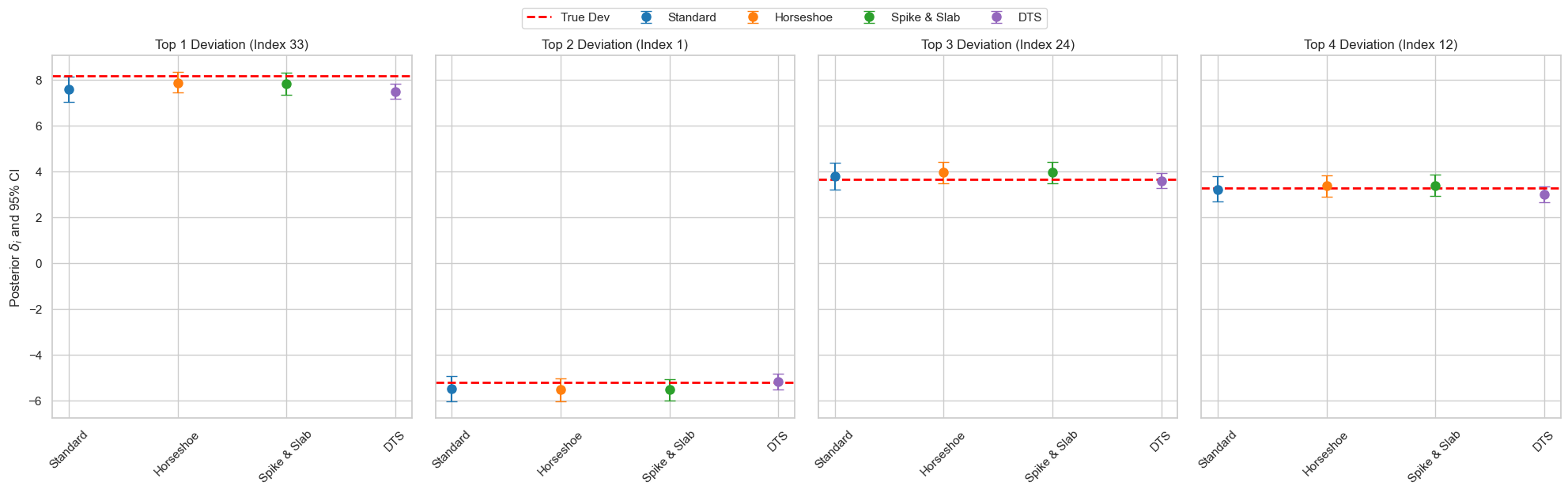}
    \caption{schools with largest $|\delta_s|$}
  \end{subfigure}
  \caption{The posterior range of $\delta_s$ for 4 schools with true $\delta_s=0$ and 4 schools with largest $|\delta_s|$ in the simulation study. The red dashed line indicates the true value of $\delta_s$ for each school.}\label{fig:star_simu_deltas}
\end{figure}

\subsection{Application to the student--teacher achievement ratio (STAR) dataset}

Now we apply the model in \eqref{eq:star_model} to the STAR dataset to assess the generalizability of the small-class effect across schools.
We summarize the posterior distribution of $\beta$, and 4 $\delta_i$'s with the largest posterior mean in absolute value in Figure \ref{fig:star_param_post}.  Similar to the simulation study, horseshoe and spike-and-slab models give similar posterior for $\beta$ as our method, whose posterior mean is lower than pooled regression and hierarchical models.

To summarize generalizability, we examine the number of schools with $|\delta_s|<0.05$ across posterior draws. This is not intended as a formal selection rule, but rather as an interpretable summary of how often the treatment effect effectively does not change at the school level. From Figure \ref{fig:star_param_post}, we observe that our method and the spike-and-slab model yield a similar sparsity pattern, with around $40\%$ to $50\%$ schools being identified as having $\delta_s$ close to zero across iterations, while the hierarchical model and the horseshoe model imply much lower sparsity.

To evaluate how well the global parameters generalize to new schools while accounting for potential deviation, we randomly hold out 20\% of the schools to examine the out-of-sample prediction performance of the global parameters estimated from each model. We repeat this process for 20 random splits of the schools and summarize the prediction error in Figure \ref{fig:star_param_post}(d). We report the Root Mean Squared Error (RMSE) for each method across the 20 splits. All models yield similar average predictive performance, but the DTS model exhibits fewer large-error outliers across splits, suggesting more stable generalization.

\begin{figure}[H]
  \centering
  \begin{subfigure}[t]{1\textwidth}
    \includegraphics[width=\textwidth,height=3cm]{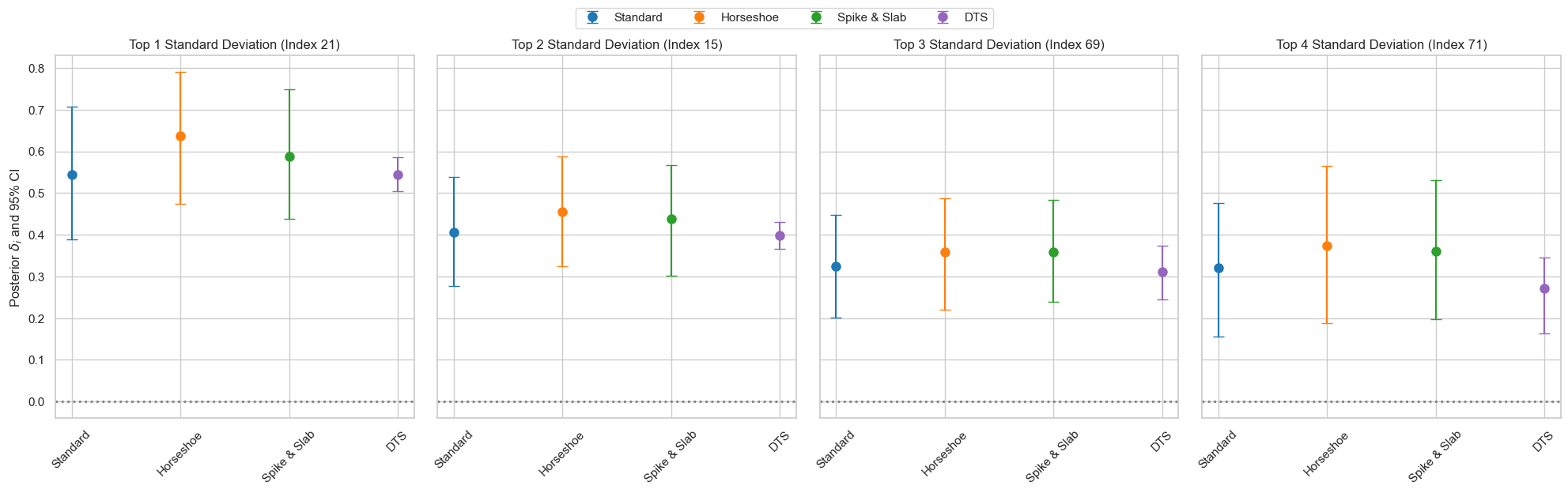}
    \caption{schools with largest $|\delta_s|$}
  \end{subfigure}\\
  \begin{subfigure}[t]{0.3\textwidth}
    \includegraphics[width=\textwidth,height=3.5cm]{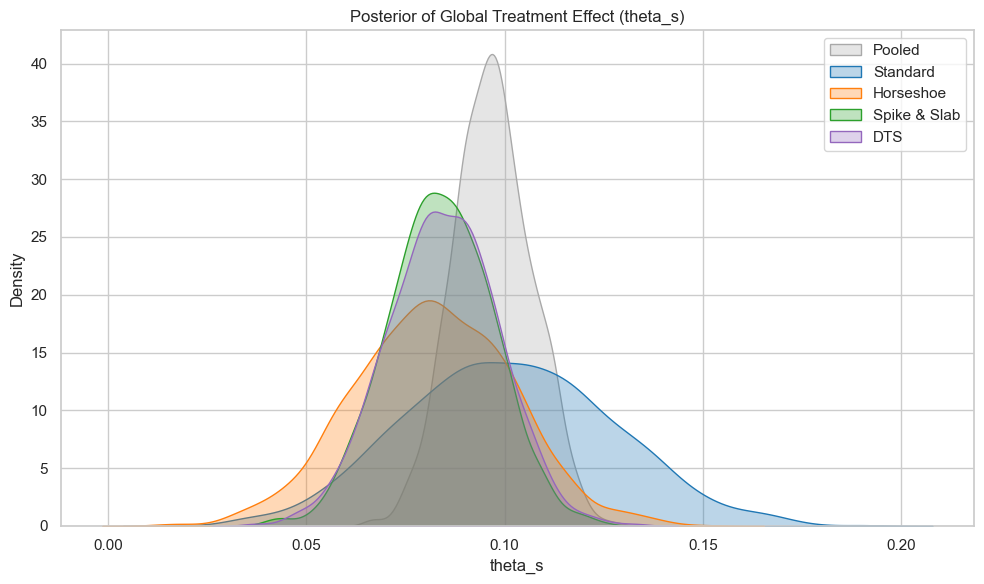}
    \caption{Posterior range of $\beta$}
  \end{subfigure}
  \begin{subfigure}[t]{0.3\textwidth}
    \includegraphics[width=\textwidth,height=3.5cm]{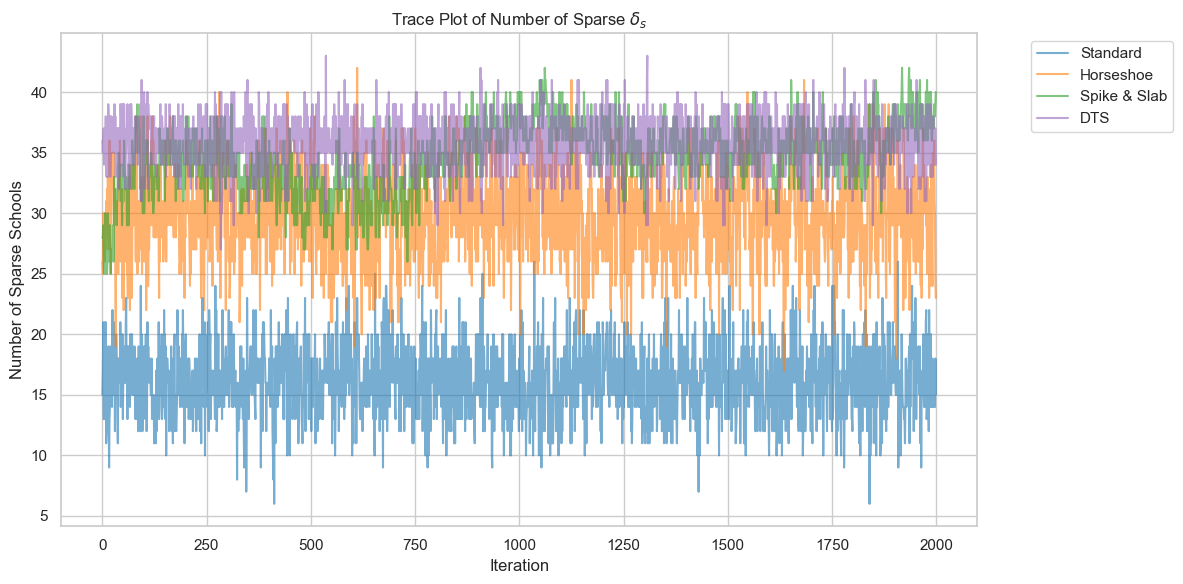}
    \caption{The trace of the number of sparse $\delta_s$'s across iterations in the real data analysis.}
    \end{subfigure}
    \begin{subfigure}[t]{0.3\textwidth}
      \includegraphics[width=\textwidth,height=3.5cm]{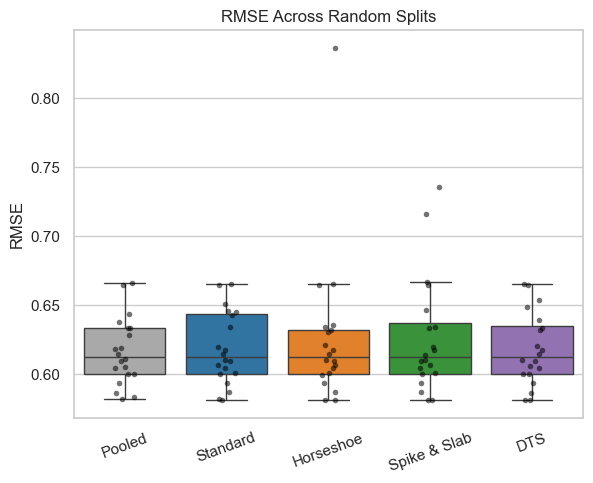}
      \caption{Prediction error across different train/test splits in the real data analysis.}
      \end{subfigure}
    \caption{Results from the STAR dataset analysis.}\label{fig:star_param_post}
\end{figure}

Finally, we report the computation time for all methods in the simulation study and the real data analysis in \Cref{tab:star_comp_time}. For all methods, we implement the MCMC using the NUTS sampler \citep{hoffman2014no}  and run the same number of iterations.  We find that sampling under our model is much faster than horseshoe and spike-and-slab models, and also outpaces the standard hierarchical model, likely due to the fact that we leverage the sparsity structure of $\delta_s$ and avoid the need to sample them separately. In this application, we find the distance-to-set formulation provides a useful combination of interpretability, competitive predictive performance, and reduced computational cost.

\begin{table}[!ht]
\caption{Computation Time Comparison (in seconds)}
\centering
\begin{tabular}{cccccc}
  \toprule
& pooled & hierarchical & horseshoe & spike\&slab & DTS \\
\midrule
simulation & 5.2  & 12.0 & 39.9 & 143.7 & 10.7  \\
real data &  52.5 &164.1  &189.0 & 1044.8 &  107.1 \\
\bottomrule
\end{tabular}\label{tab:star_comp_time}
\end{table}

\section{Application in Transfer learning}
In this section, we consider a transfer learning problem, where we seek to use information from a source dataset $S$ to improve inference in a related but smaller target dataset $T$ \citep{pan2009survey}. We denote the size of $S$ and $T$ as $n_S$ and $n_T$, respectively, and $n_T<n_S$.

A common Bayesian approach is to first fit a model to the source dataset $S$, then extract a {\em point} estimate (such as the posterior mean) and  then use it to construct a subjective prior for the target dataset $T$. While this procedure is valid, it is suboptimal, since it discards posterior uncertainty from $S$ that can be large unless $n_S$ is very large.
Another alternative is to fit a joint model for $S$ and $T$. However, a naive joint formulation can induce undesirable reverse transfer: information from the noisy target sample can distort inference on the source parameter. To illustrate,
consider the following model
\(
  & y_S = X_S \beta_S + \epsilon_S, \quad \epsilon_S \sim \text{N}(0, \sigma_S^2 I_{n_S}),\quad \beta_S \sim \text{N}(0, \Sigma_S),\\
& y_T = X_T \beta_T + \epsilon_T, \quad \epsilon_T \sim \text{N}(0, \sigma_T^2 I_{n_T}),\quad \beta_T \sim \text{N}(\beta_S, \lambda^{-1} I_{p}),
\)
where the $\lambda>0$ controls the amount of information transfer from $\beta_S$ to $\beta_T$. Although this formulation appears to encourage borrowing strength from $S$ to $T$, it also allows feedback in the opposite direction.  In particular, as $\lambda\to \infty$, $\beta_S \stackrel{D}{\to} \text{N}(m_\infty,V_\infty)$, where
\[ V_\infty^{-1}
=
\Sigma_S^{-1}
+\frac{1}{\sigma_S^2}X_S^\top X_S
+\frac{1}{\sigma_T^2}X_T^\top X_T,
 \quad
\text{ and } \quad
m_\infty
=
V_\infty\left(
\frac{1}{\sigma_S^2}X_S^\top y_S
+\frac{1}{\sigma_T^2}X_T^\top y_T
\right).
\]
Here, the source estimate is severely impacted by the target data, which is not desirable.

\subsection{Bayesian transfer learning via distance-to-ellipsoid model}

Instead, we can model
the target model as a distance-to-set model, as the following:
\[\label{eq:transfer_learning_distance_to_set}
  & y_S = X_S \beta_S + \epsilon_S, \quad \epsilon_S \sim \text{N}(0, \sigma_S^2 I_{n_S}),\\
  & \beta_S \sim \text{N}(0, \Sigma_S),\\
  & \mathcal L (y_T \mid \beta_S, r) \propto \exp(- \omega{  \min_{ \beta_T: \|\beta_T-\beta_S\|_2\le r}\|y_T - X_T\beta_T\|_2^2})1(y_T\not\in \mathcal Z_{\theta}).
\]
In this case, the set $\mathcal Z_{\theta}=
X_T\beta_S+\{X_Tu:\|u\|_2\le r\}$ is an ellipsoid centered at the source-based prediction. It is not hard to see that the projection simply corresponds to a ridge estimator:
\(
\hat \beta_T=
\bigl(X_T^\top X_T+\lambda_r I_p\bigr)^{-1}
\bigl(X_T^\top y_T+\lambda_r \beta_S\bigr)
=
\beta_S+\bigl(X_T^\top X_T+\lambda_r I_p\bigr)^{-1}X_T^\top (y_T-X_T\beta_S),
\)
where $\lambda_r>0$ is chosen such that $\|\hat\beta_T-\beta_S\|_2=r$. Since we assume $y_T\not\in \mathcal Z_{\theta}$, $\lambda_r>0$ is always the case. With a proper prior $\pi^r_0$ on $r$, when $\lambda=\lambda_r>0$, the posterior of $\beta_S$ is given by
\(
& \pi(\beta_S\mid y_S,y_T,\lambda) \propto \pi^r_0( \|\hat\beta_T-\beta_S\|_2 ) \left|\frac{d (\beta_S,r)}{d(\beta_S,\lambda)}\right| \bigg\vert_{r=\|\hat\beta_T-\beta_S\|_2} \\
& \cdot
\exp\!\left(
-\frac{1}{2\sigma_S^2}\|y_S-X_S\beta_S\|_2^2
-\frac12 \beta_S^\top \Sigma_S^{-1}\beta_S
-\omega_r\, \|(I-H_\lambda)(y_T-X_T\beta_S)\|_2^2
\right),
\)
where $H_\lambda:=X_T(X_T^\top X_T+\lambda I_p)^{-1}X_T^\top$, and we use change of variable from $r$ to $\lambda$. The second line corresponds to a kernel of $\text{N}(m_\lambda,V_\lambda)$:
\(
& V_\lambda^{-1}
=
\Sigma_S^{-1}
+\frac{1}{\sigma_S^2}X_S^\top X_S
+2\omega X_T^\top (I-H_\lambda)^2 X_T,\\
& m_\lambda
=
V_\lambda
\left(
\frac{1}{\sigma_S^2}X_S^\top y_S
+
2\omega X_T^\top (I-H_\lambda)^2 y_T
\right).
\)
As $\lambda\to \infty$, $\hat\beta_T\to \beta_S$. In the meantime, $H_\lambda\to 0$, but the $ V_\lambda^{-1}\to \Sigma_S^{-1}+ (\sigma_S^2)^{-1} X_S^\top X_S +2\omega X_T^\top X_T$. Therefore, with $\omega< (2\sigma_S^2)^{-1}$, the reverse impact from the target data is controlled.

While one could set $\omega$ to be  decreasing in $\lambda$, this is not necessary because the distance-to-set model automatically tunes $\omega$ in the following sense: as $\lambda$ increases, the distance
$
  \|(I-H_\lambda)(y_T-X_T\beta_S)\|_2
$
increases so that the distance kernel in \eqref{eq:transfer_learning_distance_to_set} favors smaller $\omega$. In this application, we assign a simple shrinkage prior $\omega\sim \text{Exp}(1)$, and half-Cauchy prior $r\sim \text{C}^+(0.01)$.

\begin{remark}
  Using the envelope-mixture equivalence for Gaussian location mixture \citep{polson2016mixtures}, there exists $(r,\tilde\lambda): r=\|\hat\beta_T(\tilde\lambda, y_T, \beta_S)-\beta_S\|_2$, such that \\ $\max_{\beta_T: \|\beta_T-\beta_S\|_2\le r}\exp(- \omega{  \|y_T - X_T\beta_T\|_2^2}) \propto \int \exp(- {  \omega \|y_T - X_T b \|_2^2 - \omega \tilde\lambda\|b-\beta_S\|_2^2}) \text{d}b,$ for given $y_T,\beta_S$. Despite some connection, $\tilde \lambda$ here is an implicit function of $y_T,\beta_S$, whereas $\lambda$ in the joint model is independent of $y_T,\beta_S$, hence the two models differ fundamentally.
\end{remark}

\subsection{Bayesian transfer learning simulations}
We evaluate the distance-to-set transfer learning model in  a simulation study, and compare it to a naive joint hierarchical model. We set $n_S = 200$, $n_T = 30$, and $p = 5$. The true source coefficient is $\beta^0_S = (1.5, 1.2, 1.8, 0, 0)^\top$. To examine a range of transferability conditions, we generate target coefficients according to $\beta^0_T = \beta^0_S + \alpha\,\varepsilon$ with $\varepsilon \sim \mathrm{N}(0, I_p)$ and use $\alpha=0.05, 2.0$ and $8.0$, corresponding to small, moderate and large deviation from the source.
Source and target data are generated as $y_S = X_S\beta^0_S + \varepsilon_S$ and $y_T = X_T\beta^0_T + \varepsilon_T$, with $\varepsilon_S \sim \mathrm{N}(0, I_{n_S})$ and $\varepsilon_T \sim \mathrm{N}(0, I_{n_T})$, and design matrices with i.i.d.\ standard normal entries.

Figure~\ref{fig:sim_betaT_comparison}(a) shows the marginal posterior distributions of $\hat\beta_T$ under the distance-to-set model for each scenario. When the deviation is small ($\alpha = 0.05$), the target is effectively exchangeable with the source, and the posterior concentrates tightly near the source OLS estimate, demonstrating effective borrowing of strength. Under moderate deviation ($\alpha = 2.0$), the model exhibits bimodal posteriors for several coordinates, capturing genuine uncertainty about whether source information is transferable for each component of $\beta_T$. Under large deviation ($\alpha = 8.0$), the source is of limited relevance and the posteriors spread out, gravitating toward the target-only OLS estimate rather than forcing transfer.

\begin{figure}[H]
  \begin{subfigure}[t]{0.48\textwidth}
  \includegraphics[width=\textwidth]{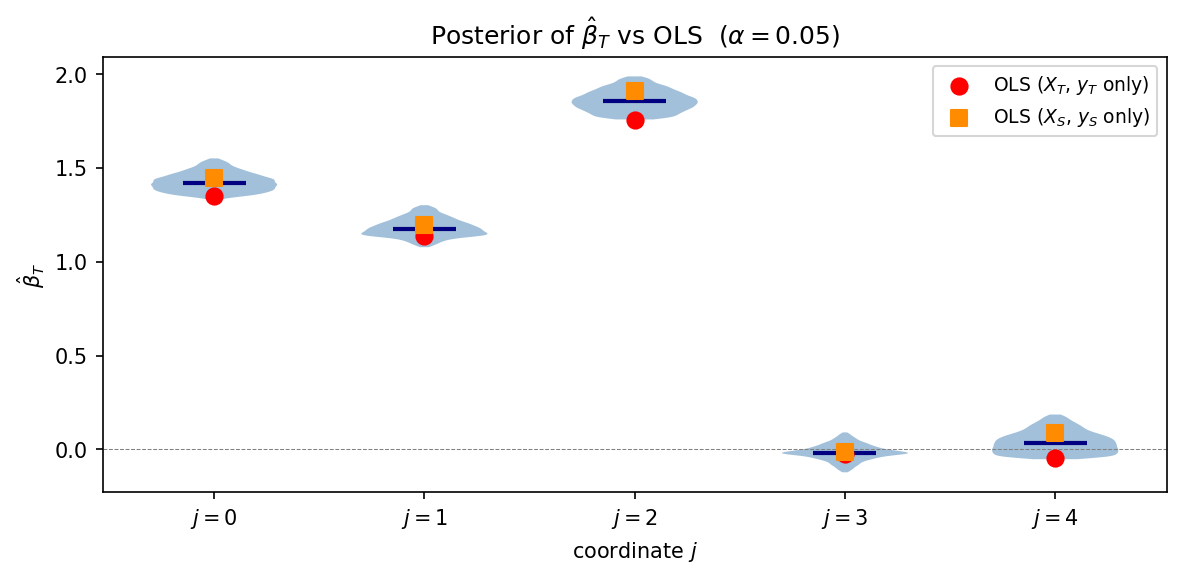}\\
  \includegraphics[width=\textwidth]{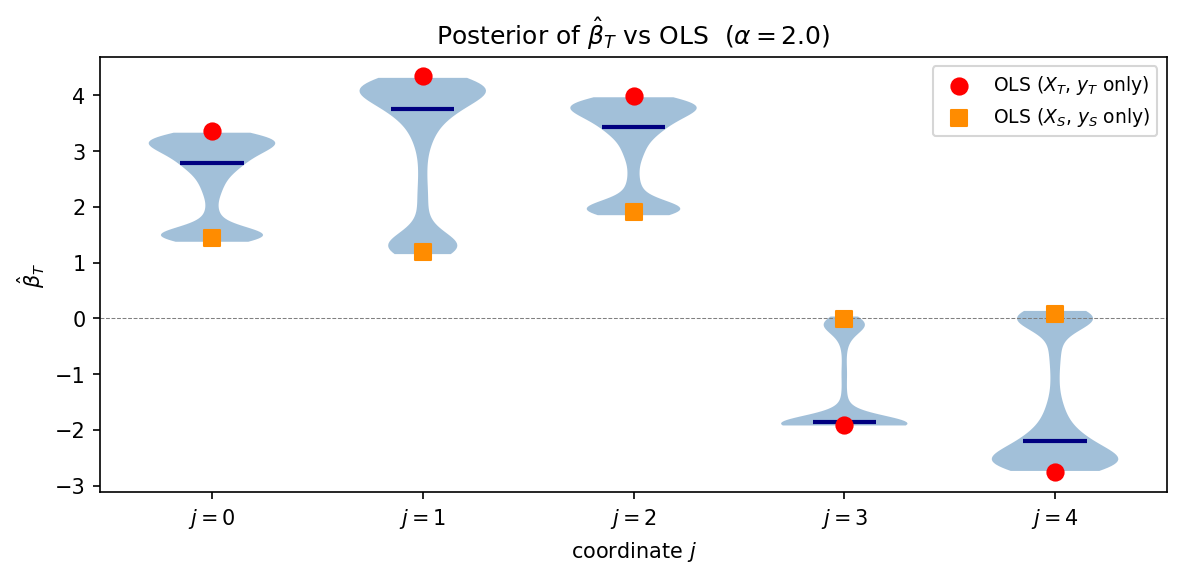}\\
  \includegraphics[width=\textwidth]{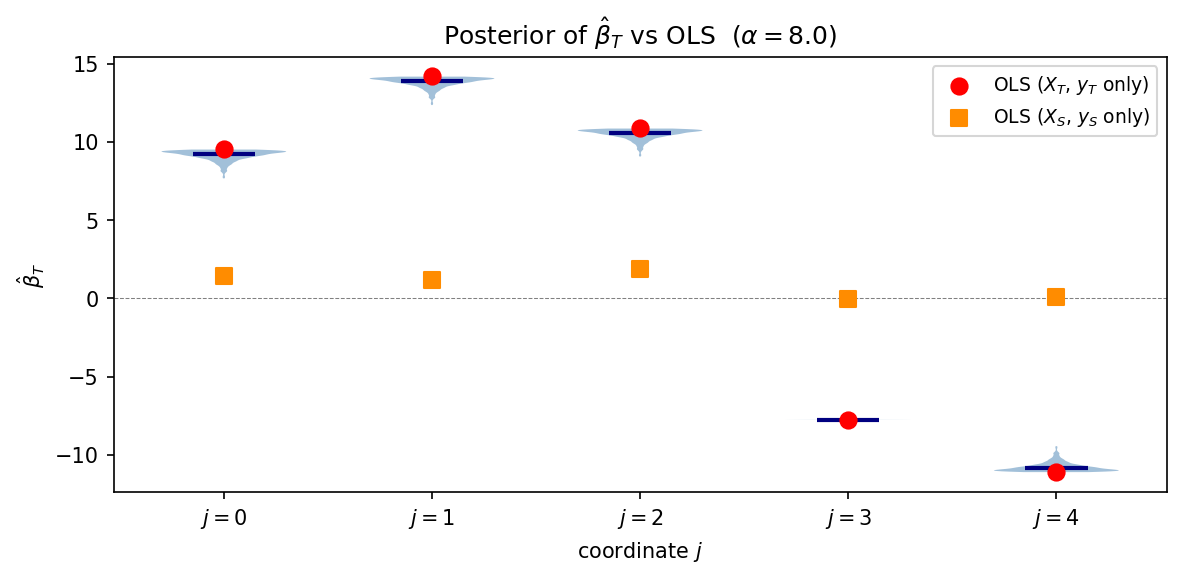}
  \caption{Posterior of $\hat\beta_T$ under the distance-to-set model.}
\end{subfigure}
\hfill
\begin{subfigure}[t]{0.48\textwidth}
  \includegraphics[width=\textwidth]{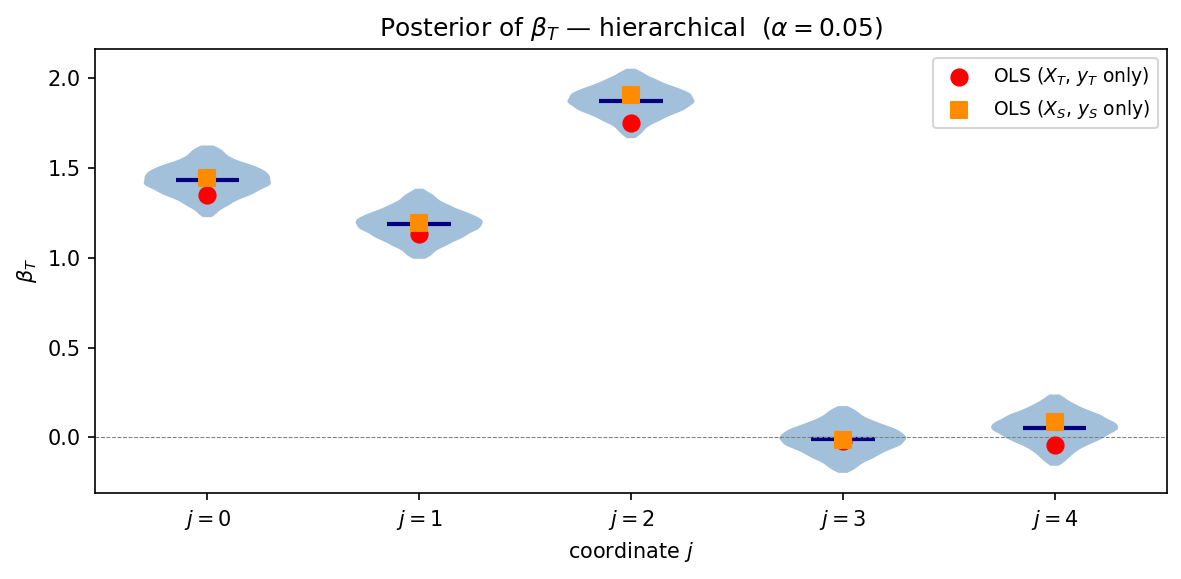}\\
  \includegraphics[width=\textwidth]{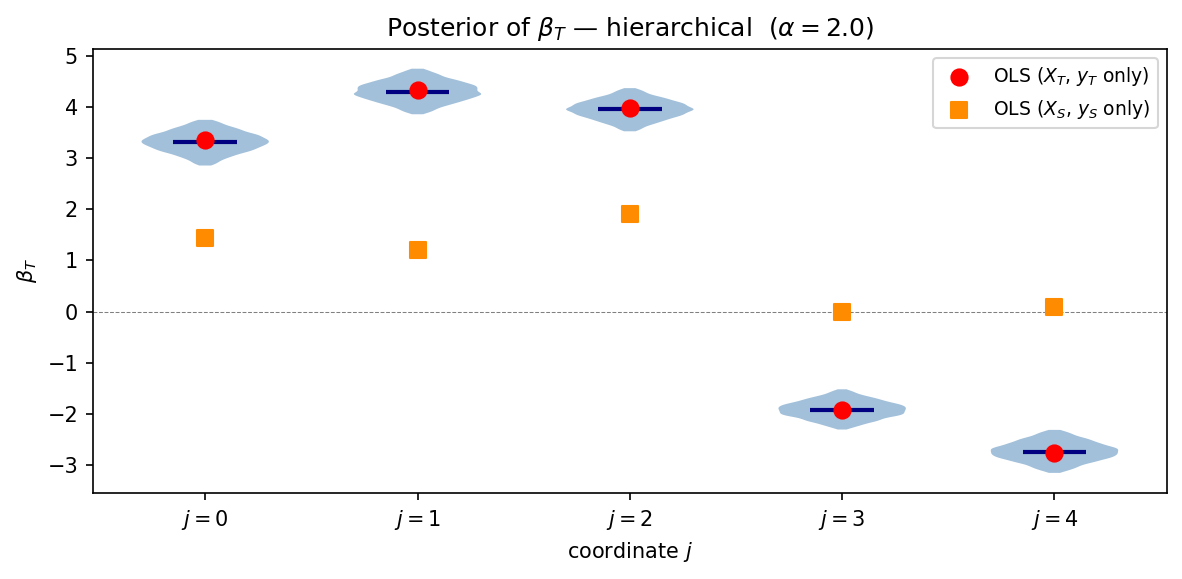}\\
  \includegraphics[width=\textwidth]{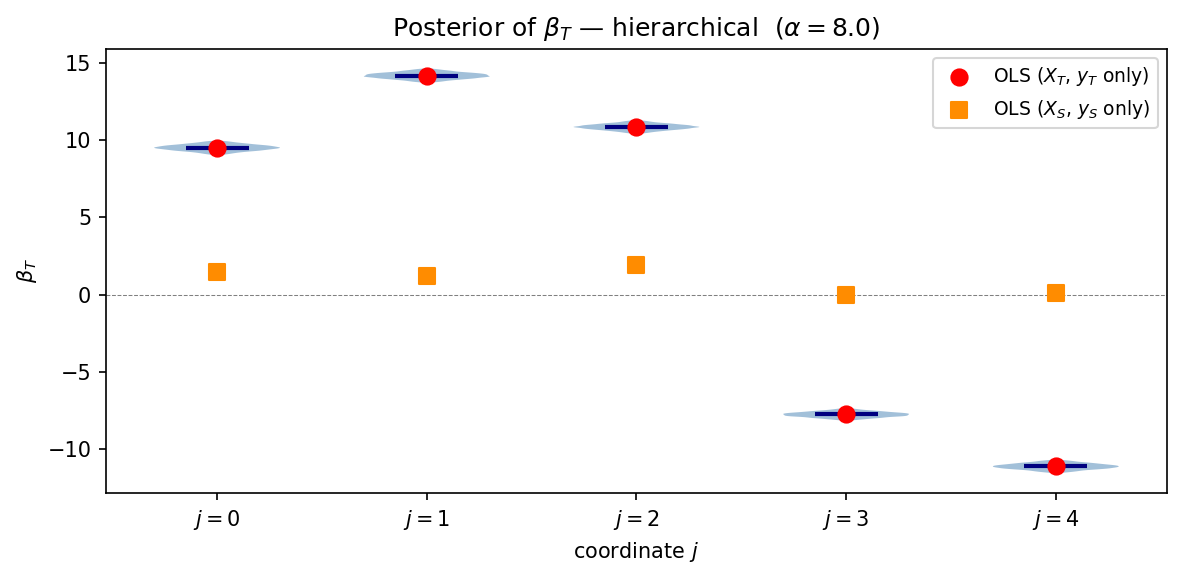}
  \caption{Posterior of $\beta_T$ under the naive joint model.}
\end{subfigure}
\caption{Comparison of the posterior of target regression coefficient under the distance-to-set model and the naive joint model for small ($\alpha = 0.05$, top), moderate ($\alpha = 2.0$, middle), and large ($\alpha = 8.0$, bottom) deviation. Violin plots display the marginal posterior for each coordinate; red circles denote the target-only OLS estimate and orange squares the source-only OLS estimate.\label{fig:sim_betaT_comparison}}
\end{figure}

\begin{figure}[H]
  \begin{subfigure}[t]{0.48\textwidth}
    \includegraphics[width=\textwidth]{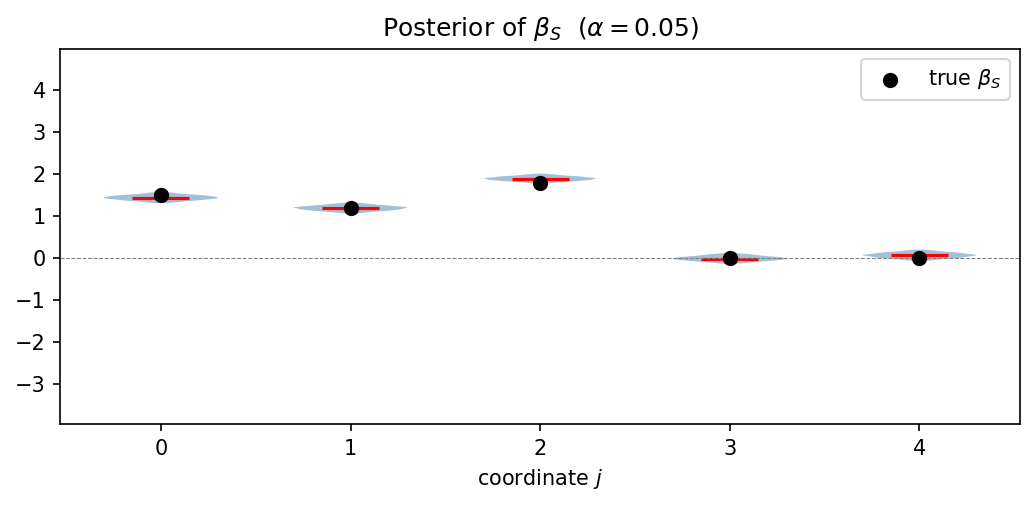}\\
    \includegraphics[width=\textwidth]{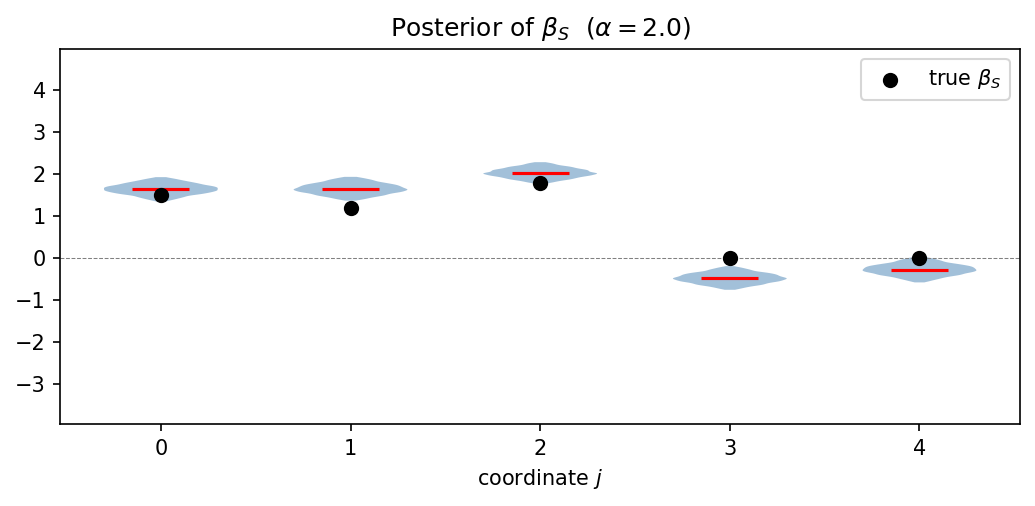}\\
    \includegraphics[width=\textwidth]{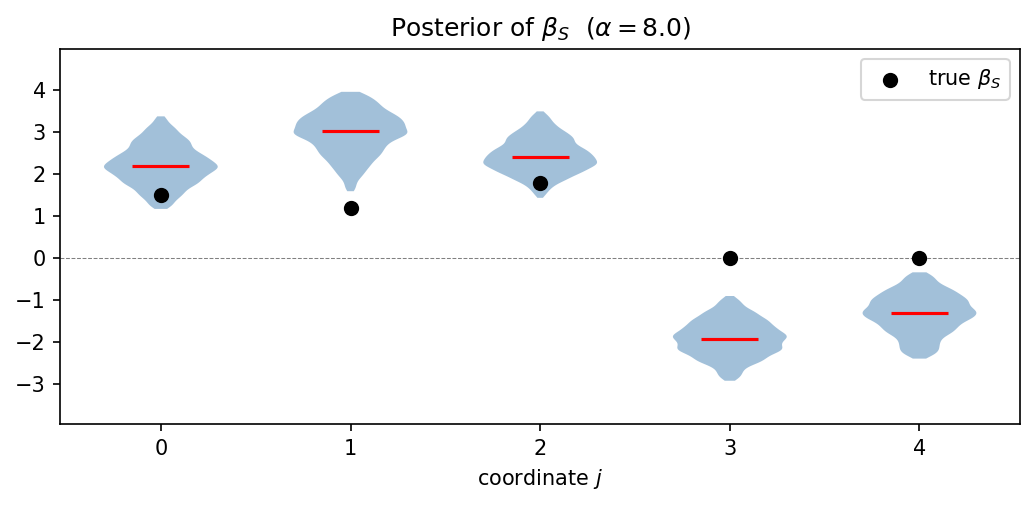}
    \caption{Posterior of $\beta_S$ under the naive joint model.}
  \end{subfigure}
  \hfill
  \begin{subfigure}[t]{0.48\textwidth}
    \includegraphics[width=\textwidth]{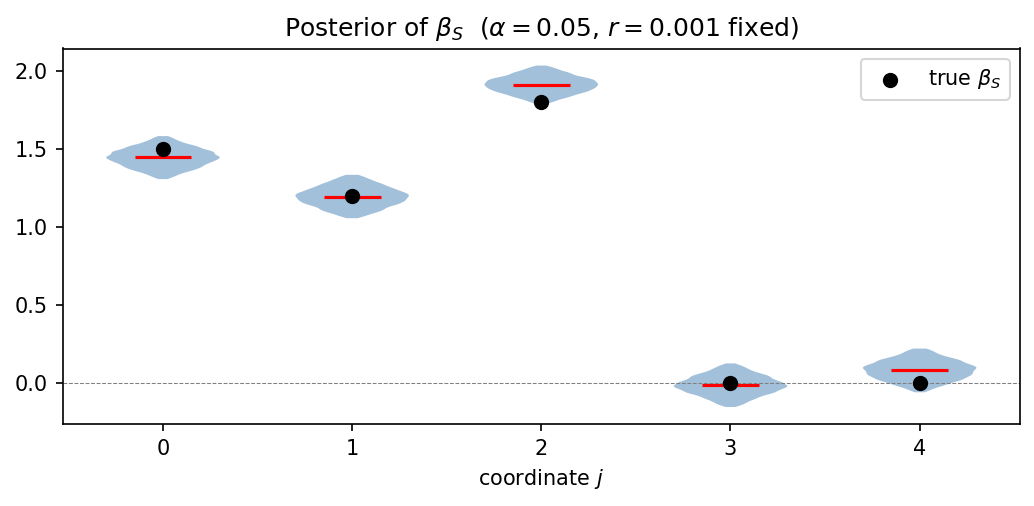}\\
    \includegraphics[width=\textwidth]{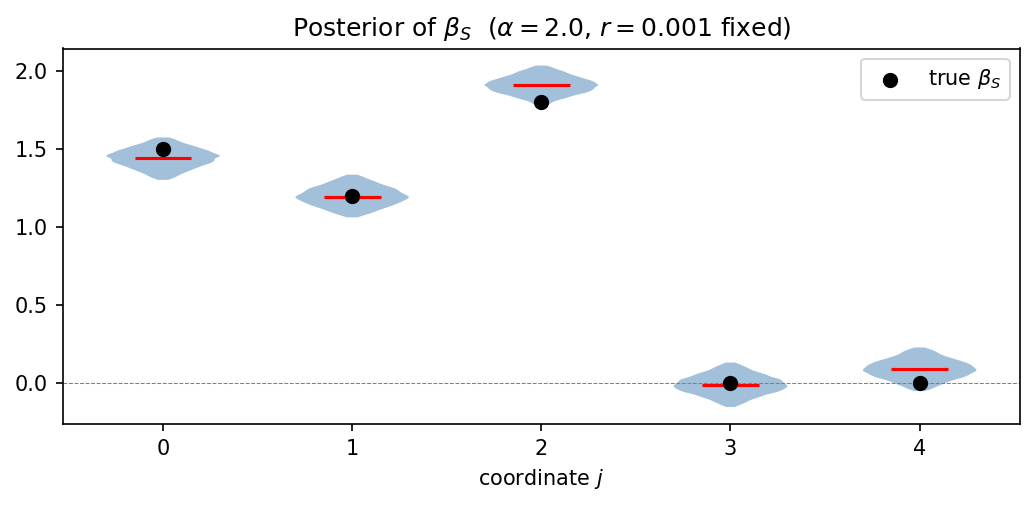}\\
    \includegraphics[width=\textwidth]{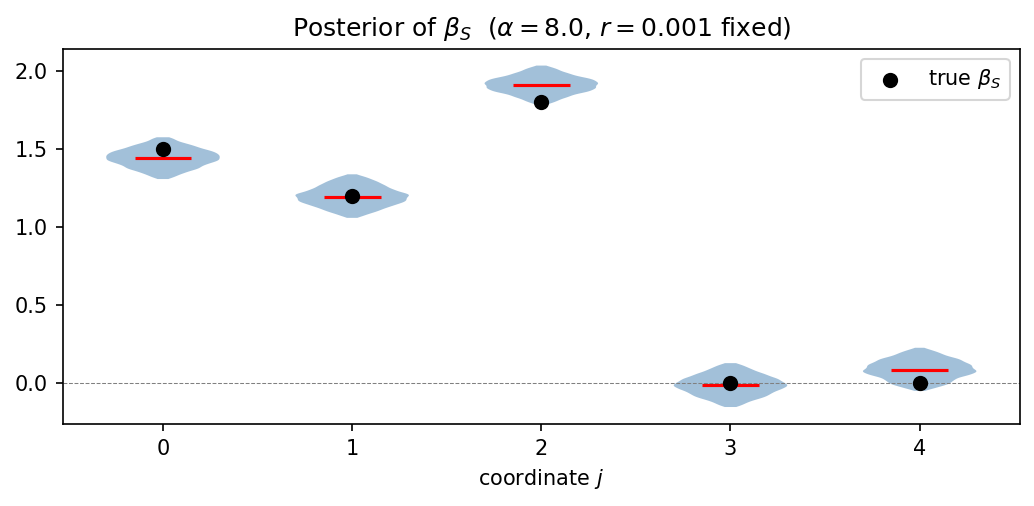}
    \caption{Posterior of $\beta_S$ under the distance-to-set model.}
  \end{subfigure}
  \caption{Posterior of $\beta_S$ under the naive joint model and the distance-to-set model for small ($\alpha = 0.05$), moderate ($\alpha = 2.0$), and large ($\alpha = 8.0$) deviation, with fixed coupling strength. In the naive joint model, the joint coupling distorts $\beta_S$ posterior away from the source-only estimate as deviation increases. In the distance-to-set model, $\beta_S$ is almost  unaffected by the target data.}
  \label{fig:sim_reverse_transfer}
\end{figure}

Figure~\ref{fig:sim_betaT_comparison}(b) displays the corresponding posteriors of $\beta_T$ under the naive joint model, in which $\beta_T \sim \mathrm{N}(\beta_S, \lambda^{-1} I_p)$ with $\lambda$ assigned a half-Cauchy prior $\text{C}^+(0.01)$. In contrast to the distance-to-set model, the joint model produces more diffuse posteriors despite large source sample size and it closely tracks the target-only OLS estimates.
Furthermore, under moderate deviation, the posterior is still unimodal --- the prior of $\beta_T$ fails to capture the uncertainty about whether transfer is warranted for each coordinate.

A more concerning limitation of the naive joint models is the effect of \emph{reverse transfer}. Due to the coupling between $\beta_S$ and $\beta_T$, information from the small, noisy target dataset can contaminate the posterior of $\beta_S$, which should be well-determined by the large source sample alone. To illustrate this, Figure~\ref{fig:sim_reverse_transfer} shows the posteriors of $\beta_T$ and $\beta_S$ under a joint model with a fixed, large coupling strength ($\lambda$ fixed at $1000$). When the deviation is small, the coupling is harmless. However, as $\alpha$ increases and $\beta^0_T$ departs substantially from $\beta^0_S$, the joint model pulls the posterior of $\beta_S$ toward the target estimate, distorting inference on a parameter that the source data alone should determine accurately. In contrast, the distance-to-set model is essentially immune to this issue by design: even when the coupling strength is large (with $r$ fixed at $0.001$), the projection satisfies $\hat\beta_T \approx \beta_S$, but because the distance is large (as the set shrinks with $r \downarrow 0$), $\omega$ becomes small. As a result, the posterior of $\beta_S$ remains effectively unaffected by the target data.

\subsection{Application in click-through rate prediction}
Finally, we consider an application of transfer learning to click-through rate (CTR) prediction in online advertising. Accurate user-level CTR estimates are central to bidding, budget allocation and campaign optimization, but data availability often varies substantially across campaigns.
While a large, established campaign may collect thousands of impressions, a newly launched or niche campaign often has far fewer observations, making reliable estimation challenging. Transfer learning provides a principled framework to leverage information from a data-rich source campaign to improve inference in a data-scarce target campaign, provided both share a common feature space and have reasonably similar regression coefficients.

We use the Criteo Private Ad dataset \citep{criteo2025privateAd}, which contains detailed impression-level logs from a real-world advertising platform. Our analysis considers data spanning three consecutive days and focuses on two campaigns: the source campaign with $n_S = 2{,}000$ users and 3,513 impressions, and the target campaign with $n_T = 200$ users and 762 impressions. To facilitate modeling, impression-level data are aggregated at the user level. The response for each user $i$ is defined as $y_i = \operatorname{logit}(\widehat{\mathrm{CTR}}_i)$, where $\widehat{\mathrm{CTR}}_i$ denotes the empirical click-through rate for that user. Each user's covariate vector is given by the mean of their numeric impression-level features across all impressions. The feature set contains $p = 21$ variables, grouped into four categories: (i) key-value features, (ii) features computed within the browser bits, (iii) features present in Criteo’s production system but intentionally made hidden to the bidding system due to privacy constraints, and (iv) contextual features. All features are standardized using the mean and standard deviation computed from the source data.

Our main goal is to obtain accurate predictions for the regression problem in the target campaign. We apply the distance-to-set transfer learning model \eqref{eq:transfer_learning_distance_to_set} to estimate the target regression coefficient $\beta_T \in \mathbb{R}^p$, incorporating information from the source campaign. The source coefficient $\beta_S$ is estimated jointly with $\beta_T$, and its posterior uncertainty is fully propagated through the inference procedure.

\begin{figure}[H]
  \centering
  \includegraphics[width=1\textwidth]{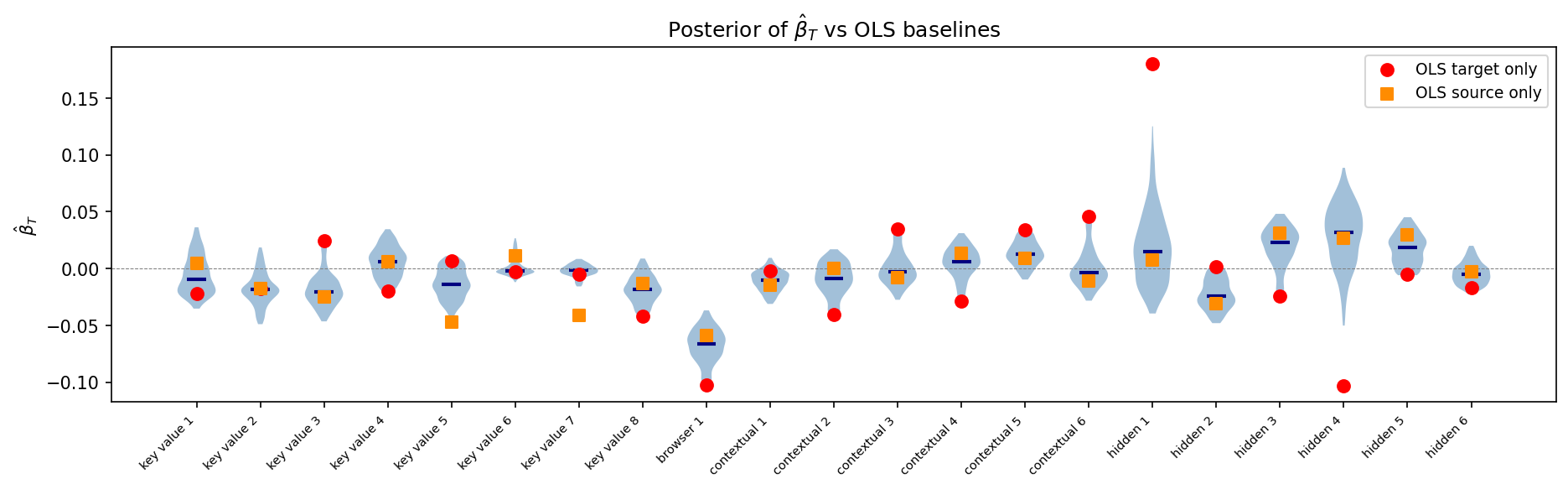}
  \caption{Posterior of target regression coefficients $\hat\beta_T$ under the distance-to-set model compared to OLS estimates fit on the target or source data only.}
  \label{fig:beta_T_hat_vs_ols}
\end{figure}

Figure~\ref{fig:beta_T_hat_vs_ols} shows the posterior $\hat\beta_T$ under the distance-to-set model, compared to the target-only and source-only OLS estimators. Although the source campaign has a reasonably large sample size ($n_S$), there remains appreciable posterior uncertainty in the transferred coefficients. The results demonstrate that the posterior of $\beta_T$ does not merely transfer all the information from $\beta_S$, instead, it is selective in adapting the information: several feature coefficients are completely based on the source data (e.g., hidden 1 and 4), while some features coefficients are based on the target data (e.g., key value 7).

For comparison, we assess the predictive performance of the distance-to-set model by using the posterior mean of $\hat\beta_T$ alongside three alternative approaches: (i) the ordinary least squares estimator fit solely on the target data; (ii) naive full-transfer, which applies the source OLS directly to the target cohort; and (iii) the transfer-LASSO estimator \citep{li2022transfer}.

\begin{figure}[H]
  \centering
  \includegraphics[width=0.5\textwidth]{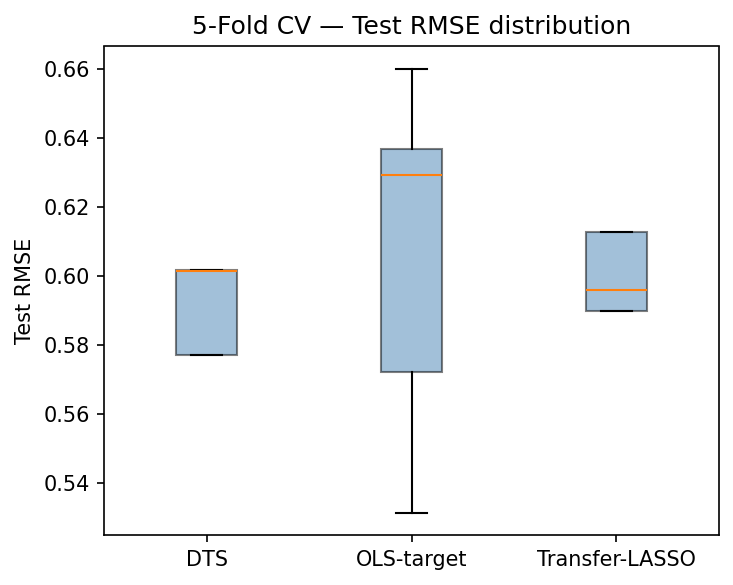}
  \caption{Prediction performance: 5-fold cross-validated root mean squared error (RMSE) for different estimators in the target campaign: OLS on target data (OLS-target), distance-to-set Bayesian transfer (DTS), and Transfer-LASSO \citep{li2022transfer}. DTS achieves the lowest mean RMSE and lowest variance across folds; full transfer (using source OLS) performs very poorly due to distributional mismatch (RMSE $2.03$), hence not shown in the figure.}
  \label{fig:cv_mse_boxplot}
\end{figure}

We use 5-fold cross-validation to evaluate the predictive performance of different estimators (with 160 subjects for training and 40 subjects for testing). Figure~\ref{fig:cv_mse_boxplot} shows the boxplot of prediction root mean squared error (RMSE). Notably, direct application of the source coefficients leads to poor performance on the target data, with a mean RMSE of $2.03$ (more than three times worse than the other three methods), so we do not show the results in the figure. This highlights substantial differences between campaigns and the inadequacy of full transfer without adaptation. In contrast, the distance-to-set transfer model achieves the lowest mean RMSE ($0.580$) among the methods, as well as the lowest standard deviation across folds ($0.063$), indicating the most stable out-of-sample prediction. The OLS-target estimator attains a mean RMSE of $0.606$ with standard deviation $0.053$. The transfer-LASSO achieves a mean RMSE of $0.585$ with standard deviation $0.066$, and it transfers 19 features from source, and adjusts only key value 6 and key value 7 in all folds. These results demonstrate that the distance-to-set model not only excels in mean predictive accuracy, but also offers more stable performance across folds.

\section{Discussion}

The distance-to-set framework enables the incorporation of an optimization component within the likelihood, blending canonical Bayesian statistical inference with projection-based computation. A central advantage over conventional latent variable models is that the marginal likelihood remains tractable. The latent projection $\hat z_i$ is obtained as a byproduct of likelihood evaluation, rather than as an additional random variable to be sampled, thereby avoiding the high-dimensional augmented space that often leads to slow mixing in data-augmented MCMC. At the same time, the projection retains an interpretable role as a latent representation of the observation relative to the underlying set, with uncertainty inherited from the posterior uncertainty over $\theta$.

Several interesting directions remain for future work. First, as shown in the transfer learning example, the invariance of the projection allows us to separate source and target information in a principled way.
Extending this framework to handle multiple sources or nonlinear feature maps would be a natural next step. Second, the distance-to-set approach connects closely to other ideas based on geometric discrepancy. Recent work on epigraph projections and proximal MCMC \citep{zhou2024proximal,shukla2025proximal} extend the class of kernels that can be considered beyond the squared distance we focus on here, and recent work on distributionally robust optimization \citep{blanchet2025distributionally}
and predictively oriented posteriors \citep{mclatchie2025predictively} share connections to these contributions as well.
Third, in some applications, exact projection onto $\mathcal Z_\theta$ may be unavailable or computationally expensive, motivating the development of numerical solutions and a library of high quality approximations. Developing practical algorithms and corresponding theoretical guarantees for such settings remains an important open problem.

\medskip
\noindent\textbf{Disclosure of the use of generative AI:} In preparing this paper, the authors used Claude Code to assist with code development and figure plotting. The authors take full responsibility for all content presented.

\spacingset{0.5}

\bibliographystyle{chicago}
\bibliography{ref_fixed}

\spacingset{1.8}

\section{Supplementary Materials}

\subsection{Proofs}

\textbf{Proof of Theorem 1}

\begin{proof} If translation--equivariance holds for the projection rule $\mathcal P_{\mathcal M_\theta}$, then
  $$
  \begin{aligned}
  v\in V_{\mathcal M_\theta +b}(z+b)
  &\Rightarrow z+b\in \mathcal P_{\mathcal M_\theta +b}(z+b+v) \Rightarrow z+b\in \mathcal P_{\mathcal M_\theta}(z+v)+b \Rightarrow z\in \mathcal P_{\mathcal M_\theta}(z+v) \\
  &\Rightarrow v\in V_{\mathcal M_\theta}(z).
  \end{aligned}
  $$
  Conversely, take any $y$ and $z\in \mathcal P_{\mathcal M_\theta}(y)$, so $y-z\in V_{\mathcal M_\theta}(z)=V_{\mathcal M_\theta+b}(z+b)$ yields $z+b\in \mathcal P_{\mathcal M_\theta +b}(y+b)$ for all $b$, which is exactly translation--equivariance.
  \end{proof}
  \smallskip
\textbf{Proof of Theorem 2}

\begin{proof}
  For convex $\mathcal Z_{\theta}$, we can apply the Steiner formula (see \cite{colesanti2022hadwiger} for details):
  \begin{equation}
    V(\mathcal Z_{\theta}^t)=\sum_{k=0}^{d} \left(\kappa_{d-k} V_k(\mathcal Z_{\theta}) \right )\,\,t^{d-k},\qquad t> 0,
  \end{equation}
  where $\kappa_j=\pi^{j/2}/\Gamma(j/2+1)$ is the $j$-dimensional volume of the unit Euclidean ball, and $V_k(\mathcal Z_{\theta})>0$ is the $k$-th dimensional volume. Differentiating $V(\mathcal Z_{\theta}^t)$ and integrating \eqref{eq:normalizing_constant} yields
  \(
  m_{\theta,\sigma}
  & =\sum_{k=0}^{d-1} (d-k) \kappa_{d-k}\,\,V_k(\mathcal Z_{\theta})
  \int_0^\infty e^{-t^2/\sigma}\,t^{d-k-1}\,dt.\\
  & =\sum_{k=0}^{d-1}(d-k)\kappa_{d-k}\,\,V_k(\mathcal Z_{\theta})\,
  \frac12\,\sigma^{(d-k)/2}\,\Gamma\!\Big(\frac{d-k}{2}\Big)\\
  & =\sum_{k=0}^{d-1}\kappa_{d-k}\,V_k(\mathcal Z_{\theta})\,
  \,\sigma^{(d-k)/2}\,\Gamma\!\Big(\frac{d-k}{2}+1\Big)\\
  & =\sum_{k=0}^{d-1} \pi^{(d-k)/2} \,V_k(\mathcal Z_{\theta})\,
  \,\sigma^{(d-k)/2},
  \)
  where we used the identity $\int_0^\infty e^{-t^2/\sigma}\,t^{m-1}\,dt=\frac12\sigma^{m/2}\Gamma(m/2)$ for $m>0$.

\end{proof}

  \textbf{Proof of Theorem 3}

  \begin{proof}
    The proof is a direct application of the Danskin's Theorem. Here
    $
      \partial_\theta \mathrm{dist}^2(y, \mathcal M_\theta) = \text{conv}_{ u\in U_*(\theta)} \{ \nabla_\theta \Psi(u, \theta)  \},
    $
    which requires $U_*(\theta)$ to be compact. If $\arg\max_u \Phi(u, \theta)$ is not compact, we can make $U_*(\theta)$ a compact subset hence satisfies the condition.
    \end{proof}

  \textbf{Proof of Theorem 4}

  \begin{proof}

    The log of the unnormalized likelihood ratio between $\theta$ and $\theta_0$ is
    \(
      \ell_n(\theta) - \ell_n(\theta_0) = - \frac{n}{\sigma} \left[ R_n(\theta) - R_n(\theta_0) \right].
    \)
    By the uniform law of large numbers,
    $
    R_n(\theta)-R_n(\theta_0)=\{R(\theta)-R(\theta_0)\}+o_{\mathbb P}(1)    $
    uniformly in $\theta$.

    For
    $
      R_n(\theta) = \frac{1}{n} \, \mathrm{dist}^2(y^{(n)}, \mathcal{Z}^{(n)}_\theta),
    $
    by (A1), there exists $\kappa > 0$ such that
    $
      R(\theta) - R(\theta_0) \geq \kappa \, \rho^2(\theta, \theta_0)
    $
    for all $\theta \in \Theta$. In particular, for $\theta \in \Theta_{n}(M,M_1)$,
    \(
      \inf_{\theta \in \Theta_{n, M}} \left\{ R(\theta) - R(\theta_0) \right\} \geq \kappa M^2 \epsilon_n^2.
    \)
Therefore, we can define the gap process
    \(
      \Delta_n(\theta) := R_n(\theta) - R_n(\theta_0), \qquad T_n := \inf_{\theta \in \Theta_{n, M}} \Delta_n(\theta).
    \)

    We pick a constant $0 < \tau < 1$ and define the test
    \(
      \varphi_n :=\mathbf 1\left\{ T_n \leq \tau \kappa M^2 \epsilon_n^2 \right\},
    \)
    in the below, we will use $\tau=1/2$.

    \noindent\textbf{(a). Vanishing type-I error.}
    Under $H_0$, we decompose
    \(
      T_n\ & = {\inf_{\theta\in\Theta_{n}(M,M_1)}\{R(\theta)-R(\theta_0)\}}
      \;+\; \inf_{\theta\in\Theta_{n}(M,M_1)}\big\{[R_n(\theta)-R(\theta)]-[R_n(\theta_0)-R(\theta_0)]\big\}
      \\
      & \ge  \kappa M^2\varepsilon_n^2  -\,2\ \sup_{\theta\in\Theta}\,|R_n(\theta)-R(\theta)|.
    \)
    Hence by (A2)
    \(
      \mathbb{E}_{\theta_0} \varphi_n = \mathbb P_{\theta_0}\!\left(T_n\le \tfrac12\,\kappa M^2\varepsilon_n^2\right)
      \ \le\
      \mathbb P_{\theta_0}\!\left(\sup_{\theta\in\Theta}|R_n(\theta)-R(\theta)|\ \ge\ \tfrac14\,\kappa M^2\varepsilon_n^2\right) \to 0.
    \)

    \noindent\textbf{(b). Vanishing type-II error.}
    Fix $\theta^\star\in\Theta_{n}(M,M_1)$ as the true parameter under $H_1$.  Then
    \(
    \{\varphi_n=0\}\ \Rightarrow\ T_n>\tfrac12\,\kappa M^2\epsilon_n^2
    \ \Rightarrow\ \Delta_n(\theta^\star)=R_n(\theta^\star)-R_n(\theta_0)\ \ge\ \tfrac12\,\kappa M^2\varepsilon_n^2.
    \)
    Let $\tilde\ell_{n,\theta}(y)$ denote the log of the normalized likelihood and write
    \(
\tilde\ell_n(\theta_0)-      \tilde\ell_n(\theta)
    =
    \frac{n}{\sigma}\,\Delta_n(\theta)\ +\ \Delta C_n(\theta),
    \)
    where $\Delta C_n(\theta):=\log m^{(n)}_{\theta,\sigma}-\log m^{(n)}_{\theta_0,\sigma} $.

    Then on the event $\{\varphi_n=0\}$,
    \(
    \tilde\ell_n(\theta_0)-\tilde\ell_n(\theta^\star)
    \ \ge\ \frac{n}{\sigma}\cdot \frac{1}{2}\,\kappa M^2\epsilon_n^2\ +\ \Delta C_n(\theta^\star).
    \)
    Using the fact that the expectation of the likelihood ratio is 1, and upon applying Markov's inequality, we get
    \(
    \mathbb P_{\theta^\star}(\varphi_n=0)
    & \le \mathbb P_{\theta^\star}\!\left(\tilde\ell_n(\theta_0)-\tilde\ell_n(\theta^\star)\ \ge\ \tfrac{1}{2}\,\tfrac{n}{\sigma}\kappa M^2\epsilon_n^2+\Delta C_n(\theta^\star)\right) \\
    & \ \le\ \exp\!\left\{-\,\tfrac{1}{2}\,\tfrac{n}{\sigma}\kappa M^2\epsilon_n^2-\Delta C_n(\theta^\star)\right\}.
    \)

    Since by (A4) we have $\inf_{\theta\in\Theta_{n}(M,M_1)}\Delta C_n(\theta)\ge -c_0\,n\epsilon_n^2$, then
    \(
    \mathbb E_{\theta^\star}(1-\varphi_n)
    \ \le\ \exp\!\left\{-\Big(\tfrac{1}{2}\,\tfrac{\kappa}{\sigma}M^2-c_0\Big)\,n\epsilon_n^2\right\}
    \ =\ e^{-c' n\epsilon_n^2},
    \)
    for $M$ large enough so that $c':=\tfrac{1}{2}(\kappa/\sigma)M^2-c_0>0$. Taking the supremum  gives
    $\sup_{\theta\in\Theta_{n}(M,M_1)}\mathbb E_\theta(1-\varphi_n)\le e^{-c' n\epsilon_n^2}$.

    Since $\mathbb E_{\theta_0}\big[\Pi(\Theta_n(M,M_1)\mid y^{(n)})\big]
\le
\mathbb E_{\theta_0}[\varphi_n]
+\sup_{\theta\in\Theta_n(M,M_1)}\mathbb E_{\theta}[1-\varphi_n]\to 0$, we have $\Pi(\Theta_n(M,M_1)\mid y^{(n)})\xrightarrow{\mathbb P_{\theta_0}} 0.$

    \noindent\textbf{(c). Vanishing mass outside the ball of radius $M_1\varepsilon_n$.}
    Let
    $
    \Theta^{\mathrm{far}}_{n}(M_1)
    :=\{\theta:\rho(\theta,\theta_0)>M_1\varepsilon_n\}.
    $
    Write the posterior as
    \(
    \Pi(B\mid y^{(n)})
    =
    \frac{\int_{B}\exp\{\tilde\ell_n(\theta)\}\,d\Pi_0(\theta)}
    {\int_{\Theta}\exp\{\tilde\ell_n(\theta)\}\,d\Pi_0(\theta)}
    \le
    \frac{\int_{B}\exp\{\tilde\ell_n(\theta)\}\,d\Pi_0(\theta)}
    {\int_{   B_\delta}\exp\{\tilde\ell_n(\theta)\}\,d\Pi_0(\theta)}.
    \)
    with
    $
    B_\delta:=\{\theta: R(\theta)-R(\theta_0)<\delta\}.
    $
    By (A3), $\Pi_0(B_\delta)>0$ hence the denominator is positive.

    Recall $\tilde\ell_n(\theta)-\tilde\ell_n(\theta_0)=-(n/\sigma)\Delta_n(\theta)-\Delta C_n(\theta)$, where
    $\Delta_n(\theta)=R_n(\theta)-R_n(\theta_0)$ and
    $\Delta C_n(\theta)=\log m^{(n)}_{\theta,\sigma}-\log m^{(n)}_{\theta_0,\sigma}$.
    Let
    $
    E_n(\delta)
    :=
    \left\{
    \sup_{\theta\in\Theta}\big|R_n(\theta)-R(\theta)\big|
    \le \delta
    \right\}.
    $
    By (A2), $\mathbb P_{\theta_0}(E_n(\delta))\to 1$. On $E_n(\delta)$, for all $\theta\in B_\delta$,
    \(
    \Delta_n(\theta)
    =R_n(\theta)-R_n(\theta_0)
    \le (R(\theta)+\delta)-(R(\theta_0)-\delta)
    \le 3\delta,
    \)
    hence
    \(
    \tilde\ell_n(\theta)
    \ge
    \tilde\ell_n(\theta_0)-\frac{n}{\sigma}\,3\delta-\Delta C_n(\theta),
    \qquad \forall \theta\in B_\delta.
    \)
For $\delta=\sigma\varepsilon_n^2$, by (A5), we have
    $
    \sup_{\theta\in B_\delta}\big|\Delta C_n(\theta)\big|
    \le c_1\,n\varepsilon_n^2.
    $
On $E_n(\delta)$, we have
    \(
    \int_{B_\delta}\exp\{\tilde\ell_n(\theta)\}\,d\Pi_0(\theta)
    \ge
    \exp\!\left\{\tilde\ell_n(\theta_0)-\frac{n}{\sigma}\,3\delta-c_1 n\varepsilon_n^2\right\}\Pi_0(B_\delta).
    \)
    Next, by Fubini and the likelihood-ratio identity,
    \(
    \mathbb E_{\theta_0}\!\left[\int_{\Theta^{\mathrm{far}}_{n}(M_1)}
    \exp\{\tilde\ell_n(\theta)-\tilde\ell_n(\theta_0)\}\,d\Pi_0(\theta)\right]
    =
    \int_{\Theta^{\mathrm{far}}_{n}(M_1)}\mathbb E_{\theta_0}\!\left[\frac{L_n(\theta)}{L_n(\theta_0)}\right]d\Pi_0(\theta),
    \)
    which equals $
    \Pi_0(\Theta^{\mathrm{far}}_{n}(M_1))$.
    Therefore, by Markov's inequality, for any $\eta>0$,
    \(
    \mathbb P_{\theta_0}\!\left(
    \int_{\Theta^{\mathrm{far}}_{n}(M_1)}
    \exp\{\tilde\ell_n(\theta)-\tilde\ell_n(\theta_0)\}\,d\Pi_0(\theta)
    >\eta
    \right)
    \le \frac{\Pi_0(\Theta^{\mathrm{far}}_{n}(M_1))}{\eta}.
    \)
    By (A6), $\Pi_0(\Theta^{\mathrm{far}}_{n}(M_1))\le e^{-b n\varepsilon_n^2}$, so taking $\eta=e^{-(b/2)n\varepsilon_n^2}$, we get
    \(
    \int_{\Theta^{\mathrm{far}}_{n}(M_1)}
    \exp\{\tilde\ell_n(\theta)-\tilde\ell_n(\theta_0)\}\,d\Pi_0(\theta)
    \le e^{-(b/2)n\varepsilon_n^2}
    \quad\text{with }\mathbb P_{\theta_0}\text{-probability }\to 1.
    \)
Since $b/2 > 3 + c_1$, we have
    $
    \Pi\!\left(\Theta^{\mathrm{far}}_{n}(M_1)\mid y^{(n)}\right)
    \ \xrightarrow{\mathbb P_{\theta_0}}\ 0.
    $

Combining (a)--(c) yields the result.

  \end{proof}

  \subsection{Intractable normalizing constant and Poisson point process data augmentation}
We now discuss the case when $m_{\theta,\sigma}$ is intractable, and the likelihood involves independent $y_i$ hence the likelihood contains an $m^{-n}_{\theta,\sigma}$ term.
In order to develop a useful algorithm, we impose the following constraints in the prior support for $\theta$ and $\sigma$:
\begin{itemize}
\item (Boundedness on $\mathcal Z_\theta$) there exists a known $R_{\max}$ and a center $c$ such that $\mathcal Z_\theta \subseteq B(c,R_{\max})=\left\{z: \|z-c\| \le R_{\max}\right\}$, with $\|\cdot\|$ some norm;
\item (Boundedness on $\sigma$) there exists a known $\sigma_{\max}$ such that $\sigma \le \sigma_{\max}$ for all $\sigma$.
\end{itemize}

With the above, we now develop a data augmentation strategy that effectively cancels out the $m^{-n}_{\theta,\sigma}$ term.
We first introduce a latent variable $(\psi \mid \theta, \sigma,y)\sim \text{Gamma}(n, m_{\theta,\sigma})$, leading to joint posterior
\(
\Pi(\theta,\sigma,\psi\mid y) \propto \; &  \Pi_0(\theta, \sigma) \;{\exp(- \sum_{i=1}^n \text{dist}^2(y_i, \mathcal Z_\theta)/\sigma)} \;\psi^{n-1} \\
& \cdot \exp \big [ -\psi
\underbrace{ \int_0^\infty e^{-u} \left \{ \int_{\mathbb R^d}\mathbf 1 (x \in  \mathcal Z_{\theta}^{\sqrt{u\sigma}}\setminus \mathcal Z_{\theta})\,dx\right \} \ du }_{=m_{\theta,\sigma}} \big ].
\)
Note that the expression of $m_{\theta,\sigma}$ only requires $\mathcal Z_\theta$ to be bounded, not necessarily convex.

For all $(\theta,\sigma)$ in the prior support, and all $u>0$, we have the $\sqrt{u\sigma}$-neighborhood contained in a known ball,
$
\mathcal Z_\theta^{\sqrt{u\sigma}}
\subseteq B\!\left(c,\, R_{\max}+\sqrt{u\sigma_{\max}}\right)
=: B_u,
$
within which we can define a Poisson point process $\Phi$ on $\mathbb R^d\times(0,\infty)$ with intensity measure
\(
\mu_0(dx,du)
:= \psi\, e^{-u}\,\mathbf 1\{x\in B_u\}\,dx\,du,
\)
which is free of $\theta$ and $\sigma$. The total mass is finite and can be computed using numerical integration:
\(
\Lambda_0
= \iint \mu_0(dx,du)
= \psi C_0, \quad C_0 = \tilde\kappa_d \int_0^\infty e^{-u}\,\bigl(R_{\max}+\sqrt{u\sigma_{\max}}\bigr)^d\,du,
\)
with $\tilde\kappa_d$ being the volume of the unit ball in $\mathbb R^d$. Letting
$
A_{\theta,\sigma}
:= \bigl\{(x,u): u>0,\ x\in \mathcal Z_\theta^{\sqrt{u\sigma}}\setminus \mathcal Z_\theta\bigr\},
$
the point process falling completely out of $A_{\theta,\sigma}$ has probability
\(
\Pr_{\mu_0}\bigl(\Phi\cap A_{\theta,\sigma}=\varnothing\bigr)
= \exp\bigl\{-\mu_0(A_{\theta,\sigma})\bigr\}, \quad \mu_0(A_{\theta,\sigma}) = \psi\, m_{\theta,\sigma}.
\)
Therefore, we have the joint posterior augmented with the latent gamma $\psi$ and the Poisson point process $\Phi$:
\[\label{eq:joint_posterior_augmented}
\Pi(\theta,\sigma,\psi,\Phi\mid y)
\propto
\Pi_0(\theta,\sigma)\,\psi^{n-1}
\exp\!\left(
-\sum_{i=1}^n \frac{\operatorname{dist}^2(y_i,\mathcal Z_\theta)}{\sigma}
\right)\,
p(\Phi)\,
\mathbf 1\{\Phi\cap A_{\theta,\sigma}=\varnothing\},
\]
where $p(\Phi)$ denotes the law of $\Phi$ under $\mu_0$. Although the dimension of $|\Phi|$ is random, making the joint posterior unsuitable for standard Metropolis-Hastings algorithm (without using reversible jump step), we can still use a Metropolis-Hastings within Gibbs sampler to update $\Phi$, $(\theta,\sigma)$, and $\psi$ separately. Specifically, we alternate between the following steps.

\begin{itemize}
  \item Update $\Phi$ given $\theta,\sigma,\psi$:
\begin{enumerate}
    \item[a] Draw $K \sim \mathrm{Poisson}(\Lambda_0)$.
    \item[b] For $j = 1, \ldots, K$:
        Draw $u_j$ from the density $f(u) \propto e^{-u}|B_u|$.
         Draw $x_j \mid u_j \sim \mathrm{Uniform}(B_{u_j})$.
         If $(x_j, u_j) \in A_{\theta,\sigma}$, discard this point; otherwise, keep it in $\Phi$.
\end{enumerate}
\item Update $\psi$ by drawing from $\mathrm{Gamma}(n + |\Phi|, C_0)$.
\item
Propose $(\theta',\sigma')$ from $q( \cdot \mid \theta, \sigma)$,
    and accept $(\theta',\sigma')$ with probability
    \(
    1 \wedge
    \frac{
    \Pi_0(\theta',\sigma')\,
    \exp\!\left(
    -\sum_{i=1}^n \frac{\operatorname{dist}^2(y_i,\mathcal Z_{\theta'})}{\sigma'}
    \right)  q(\theta,\sigma \mid \theta',\sigma')
    }{
    \Pi_0(\theta,\sigma)\,
    \exp\!\left(
    -\sum_{i=1}^n \frac{\operatorname{dist}^2(y_i,\mathcal Z_\theta)}{\sigma}
    \right) q(\theta',\sigma' \mid \theta,\sigma) }
    \mathbf 1\{\Phi\cap A_{\theta',\sigma'}=\varnothing\}.
    \)
\end{itemize}
The second step is a thinning procedure for sampling from the Poisson point process with intensity $\mu_0(dx,du) 1\{ (x,u) \notin A_{\theta,\sigma}\}$. The third step $p(\Phi)$ cancels out since $\Phi$ is considered as fixed.

\subsection{Simplifying distance calculations via duality}
A key computational advantage of the distance-to-set model is that it avoids sampling the latent coordinates $z_i$. Such an advantage is diminished if the distance calculation itself is too costly. We therefore seek efficient methods for evaluating $\text{dist}(y_i, \mathcal Z_\theta)$ when the projection lacks a closed-form solution.

We focus on the case in which $\mathcal Z_\theta$ is convex with non-empty relative interior. There are several ways to compute the projection, such as alternating projection method when $\mathcal Z_\theta$ is the intersection of several convex sets \citep{bauschke1996projection}. In some cases even when directly solving for the projection, referred to as primal problem, can be computed, it may be more efficient to work with the \textit{dual formulation}
\[\label{eq:strong_duality}
  \min_{z\in \mathcal Z_\theta} \frac{1}{2}\|y_i-z\|^2  = \max_{u\in \mathbb{R}^n} \langle y_i,u\rangle - \frac{1}{2}\|u\|^2_* - S_{\mathcal Z_\theta}(u),
\]
where $\|\cdot\|_*$ is the dual norm $\|u\|_* = \sup_{\|v\|\le 1} u\top v$, and $S_{\mathcal Z_\theta}(u)=\sup_{z\in \mathcal Z_\theta} u^{\top} z$ is the support function of $\mathcal Z_\theta$ \citep{rockafellar1997convex}. Indeed, Dykstra's algorithm can be understood as a dual form of the aforementioned alternating projection method \citep{han1988successive}, and will return the projection of the starting point even when there is more than one point in the intersection of the sets.

The representation is useful when the support function and dual norm (right hand side) are easier to evaluate than the projection itself (primal, left hand side).
For example, for the dual norm, we have $(\|\cdot\|_p)_* = \| \cdot \|_{q}$, where $1/p + 1/q = 1$. This means that $\|\cdot\|_2^* = \|\cdot\|_2$ and $\|\cdot\|_1^* = \|\cdot\|_\infty$. An example use case is via Moreau's decomposition, the projection onto the supremum norm ball (of radius 1) can be expressed as $P_{\mathcal{B}_\infty}(y) = y - S(y)$ where $S(\cdot)$ is the soft-thresholding operator that we recall corresponds to the $\| \cdot \|_1$ norm.

Likewise, support functions often admit closed forms for basic convex sets.  The support function for ellipsoid $\{z\in \mathbb{R}^n: z^\top Q^{-1} z \le r\}$ is $S_{\mathcal M_\theta}(u)= \sqrt{r u^{\top }Q u }$, and for a polytope (polyhedron with bounded vertices) is $S_{\mathcal M_\theta}(u)=\max_{i\in \mathcal I} u^{\top} v_i$, where $v_i$ are the vertices of the polytope. Duality can also be useful in computing projections onto sums of these sets as well as nonconvex sets \citep{won2019projection,won2023unified}.

Once a dual maximizer $\hat u$ is obtained, the projection of $y$ onto $\mathcal Z_\theta$ is recovered as $y_i- \hat z_i \in \partial (1/2) \|\hat u\|^2_*$ , where $\partial$ denotes the subdifferential. For differentiable $\|\hat u\|_*$, such as $\|\cdot\|_p$ with $1<p<\infty$, we have $\hat z_i = y_i- \nabla (1/2) \|\hat u\|^2_*$ , which is simply $\hat z_i=y_i -\hat u$ when $p=2$. Therefore, by solving the dual problem, we can obtain both the distance and the latent projection $z_i=\hat z$.

We now use an example to illustrate the duality technique.
\begin{example}[\bf Smoothing via an ellipsoid subject to affine constraints]\label{ex:ellipsoid}
  We consider the case of Euclidean projection to a convex set $\mathcal Z_\theta$, the intersection of an ellipsoid and an affine set:
  \[
    \mathcal Z_\theta = \{z\in \mathbb{R}^d: z^\top Q^{-1} z \le r, A z\le b\},
  \]
  where $Q\in \mathbb{R}^{d\times d}$ is a positive definite matrix, $A\in \mathbb{R}^{c\times d}$ is a matrix, and $b\in \mathbb{R}^c$ is a vector. The projection to ellipsoid (without affine constraints) is useful as it takes the form $\tilde z_i = Q(Q+ \lambda_r I)^{-1} y_i$, where $\lambda_r$ is chosen so that $\tilde z_i^\top Q^{-1} \tilde z_i=r$. The form of $\tilde z_i$ is the same as the Kriging smoother for the expected value of mean for a Gaussian process with covariance matrix $Q$ and nugget effect $\lambda_r$ \citep{cressie2015statistics}. However, there is also interest to include additional constraints, such as the order constraints $z_1 \le z_2 \le \cdots \le z_d$ considered by \cite{lin2014bayesian}. The projection onto the intersection of an ellipsoid and an affine set does not admit a closed-form expression. Although one could resort to alternating projection methods, each iteration can be expensive due to multiple inversions of $(Q+\lambda_r I)$ at different value of $\lambda_r$.

  Using duality, we can alternatively solve for the projection to $\mathcal Z_\theta$ as follows. The support function for the ellipsoid $\mathcal M^1_\theta=\{z: z^\top Q^{-1} z \le r\}$ and polyhedron $\mathcal M^2_\theta=\{z: A z\le b\}$ is
  \(
  & S_{\mathcal M^1_\theta}(u)= \sqrt{r u^{\top }Q u },
  \quad S_{\mathcal M^2_\theta}(u)=
  \begin{cases}
    \min\limits_{v \ge 0:A^{\top} v = u} b^\top v, & \text{if } u \in \text{cone}(\{A_1,\ldots,A_c\}) \\
    \infty, & \text{otherwise}.
  \end{cases}
  \)
  Since the support function for an intersection of sets can be calculated via the infimal convolution
$
S_{\mathcal M_\theta}(u)= \inf_{u_1+\cdots + u_K=u} \sum_{k=1}^K S_{\mathcal M_\theta^k}(u_k),
$ we have
  $
  S_{\mathcal M_\theta}(u)= \min_{v \ge 0}  b^{\top} v + \sqrt{r (u-A^{\top} v)^{\top} Q (u-A^{\top} v) }.
  $
   The dual problem becomes
  \(
  \max_{u\in \mathbb{R}^d, v \ge 0} y^\top u - \frac{1}{2}\|u\|^2_2
  - b^{\top} v - \sqrt{r (u-A^{\top} v)^{\top} Q (u-A^{\top} v) },
  \)
  which can be solved efficiently via L-BFGS with smoothing on the gradient. Once the optimizer $(\hat u, \hat v)$ is obtained, the projection follows as $\hat z = y - \hat u$. The resulting computation avoids repeated matrix inversions that would arise in direct projection onto the constrained ellipsoid.
  \hfill $\blacksquare$
  \end{example}

\subsection{Gradient-based proposal algorithm and distance calculation}
 For simplicity, we focus on the case when $\Pi(\theta, \sigma \mid y)$ is tractable.
Both MH-Barker and NUTS require evaluation of the log density $\log \Pi(\theta, \sigma \mid y)$ and its gradient (or subgradient) with respect to $\theta$.
The main difficulty is that the distance term depends on $\theta$ through the projection $\hat z_i(\theta)$. The chain rule
$
\nabla_{\theta} \text{dist}^2(y_i, \mathcal Z_\theta)
= 2\,\langle\, \hat z_i - y_i,\ \nabla_\theta \hat z_i(\theta)\rangle,
$
is not directly useful when $\hat z_i(\theta)$ lacks a closed form and is only available through numerical optimization \citep{scieur2022curse}.
Fortunately, we can again appeal to duality to bypass the calculation of $\nabla_\theta \hat z_i(\theta)$. We first state the following result.

\begin{theorem}
 Denote the dual function $\Psi(u,\theta) := 2\langle y,u\rangle - \|u\|^2_* - 2S_{\mathcal Z_\theta}(u)$. If there is a maximizer subset $U_*(\theta) \subseteq \arg\max_u \Psi(u, \theta)$, such that: (i) $U_*(\theta)$ is compact, (ii) for each $\hat u \in U_*(\theta)$, the map $\theta \mapsto S_{\mathcal{Z}_\theta}(\hat u)$ is  differentiable at $\theta$, and (iii) the gradient $\nabla_\theta S_{\mathcal{Z}_\theta}(u)$ is continuous in $u$ for all $u$ in a neighborhood of $U_*(\theta)$, then for any dual maximizer $\hat u\in U_*(\theta)$,
 \(
-2\nabla_\theta S_{\mathcal Z_\theta}(u) \big\vert_{u=\hat u}\in \partial_\theta \mathrm{dist}^2(y, \mathcal Z_\theta).
 \)
Moreover,  when the dual maximizer is unique,
  $$\nabla_\theta \mathrm{dist}^2(y, \mathcal Z_\theta)=-2\nabla_\theta S_{\mathcal M_\theta}(\hat u).$$

\end{theorem}
Thus, one can obtain $\nabla_\theta \mathrm{dist}^2(y, \mathcal Z_\theta)$ by simply differentiating the support function over $\theta$, and plugging the dual optimal $\hat u$ back into the gradient. To illustrate, the gradient in Example 3 is simply $\nabla_r \mathrm{dist}^2(y, \mathcal Z_\theta)=-2 \sqrt{(\hat u-A^{\top} \hat v)^{\top} Q (\hat u-A^{\top} \hat v) } r^{-1/2}$.

Consequently, when projections $\hat z_i(\theta)$ lack a closed form, this result suggests that one may solve the dual problem, recover both the distance and its gradient, and use these quantities within popular samplers such as MH-Barker or NUTS.
 Many computing software packages now support the customization of the log-posterior function and its gradient.

\subsection{Distance majorization and divergence-to-set model}\label{sec:divergence}
The distance between a point and a set is induced by the minimum metric between two points, for which the metric satisfies properties such as non-negativity, the triangle inequality, and symmetry. However, in practice, it is often useful to consider more general functions that may not fulfill all these criteria \citep{mmsf}. Divergences beyond Euclidean distance \citep{breg} can serve as substitutes for the distance in our modeling framework as discussed below.

We define the \textit{divergence-to-set} as:
\(
  \text{div}(y_i, \mathcal Z_\theta) = \|y_i- \mathcal T_{\gamma}(y_i) \|,\qquad \mathcal T_{\gamma}:\mathcal{Y} \mapsto \mathcal Z_\theta,
\)
where $\mathcal T_{\gamma}$ is a single-valued function parameterized by $\gamma$, and $\gamma$ depends on $\theta$. Clearly, when $\mathcal T_{\gamma}(y_i)$ is specified as the projection function, then divergence-to-set reduces to the distance-to-set.  It is not hard to see that the above is a majorization of the distance: since $\mathcal T_{\gamma}(y_i)  \in  \mathcal Z_\theta$,
\(
\text{dist}(y_i, \mathcal Z_\theta)= \min_{z\in \mathcal Z_\theta} \|y_i- z\| \le \|y_i- \mathcal T_{\gamma}(y_i)\| \quad \forall y_i \in \mathcal{Y}.
\)

Let the divergence-to-set model follow
\[\label{eq:distance_majorization_divergence}
 & \mathcal L(y_i ; \theta, \sigma) =   \phi [ \text{div}(y_i, \mathcal Z_\theta); \sigma],
\]
 we see how this construction generalizes the distance-to-set likelihood \eqref{eq:gaussian_kernel}. It takes the same form  upon replacing this exact projection by a more tractable map $T_\gamma$ when practical. Moreover,  though $\text{div}(y_i, \mathcal Z_\theta)$ is only an approximation to $\text{dist}(y_i, \mathcal Z_\theta)$, the scale parameter $\sigma$ remains free and can adapt in the posterior so that
the two ratios $\text{div}(y, \mathcal Z_\theta)^2/\sigma$ and $\text{dist}(y, \mathcal Z_\theta)^2/\tilde\sigma$ roughly match  in the distribution.

We now discuss two important applications in which this divergence-to-set approximation (which we will refer to as the \textit{divergence model} for short) is particularly useful. The first is when the exact projection $\mathcal P_{\mathcal Z_\theta}(y_i)$ lacks a closed-form expression, but can be estimated iteratively via
\(
  \hat z_i = \mathcal T_{\gamma}(y_i) = G^{K}(z_i^0),
\)
where $G:\mathbb{R}^d \to \mathbb{R}^d$ denotes a one-step update operator, $z_i^0$ is the initial value, and $K$ is the number of iterations. Under appropriate convergence assumptions, we  have $\lim_{K\to \infty}G^{K} (z^0) = \arg\min_{z\in \mathcal Z_\theta} \|y-z\|$, but in practice one virtually always uses some stopping criterion so that the number of steps $K$ is a function of the last output of $G$ (for example, the first time that the gradient norm falls below a tolerance) or even fixed. Formally, the resulting estimate $\hat z_i$ is therefore an approximation to the exact projection $z^*=\arg\min_{z\in \mathcal Z_\theta} \|y-z\|$ and $\text{div}(y, \mathcal Z_\theta)=\|y- \hat z\| \ge \|y- z^*\|$ as long as $\hat z\in \mathcal Z_\theta$.

The divergence model resolves any algorithmic randomness resulting from this scenario: we can  represent a set of rules associated with the iteration via $\gamma$, such as initialization  of $z_i^0$, step-size schedule, and stopping criterion, which are specified as deterministic functions of $y_i$ and $\theta$. In this case, the induced map $T_\gamma$ defines a valid divergence-majorizing surrogate for the exact distance.

The second application arises in unsupervised machine learning problems based on reconstruction error. Consider estimators as
\(
\min_{\xi\in  X_\theta}  \frac{1}{n}\sum_{i=1}^n \|y_i- \mathcal R(y_i,\xi,\theta)\|^2_2,
\)
where $\mathcal R(y_i,\xi,\theta)$ is a reconstruction map parameterized by $\xi\in X_\theta$ and $\theta\in \Theta$.  For example, in an autoencoder \citep{bank2023autoencoders,kingma2013auto}, $\mathcal R(y_i,\xi,\theta)= D_{\xi,\theta} \circ E_{\xi,\theta}(y_i)$, where $E_{\xi,\theta}$  is an encoder that transform $y_i$ to a latent code $w_i\in \mathbb{R}^m$ with \textit{bottleneck} dimension $m\ll d$, and $D_{\xi,\theta}$ is the decoder that transforms $w_i$ back to the data space $\mathbb{R}^d$.  Although these approaches typically involve optimizing over the nuisance parameter $\xi$, there often remains a parameter of interest $\theta$ whose uncertainty is scientifically meaningful. For instance, in autoencoders, $\theta$ may encode structural choices such as the sparsity pattern of the latent representation $w_i$, which in turn determines the intrinsic dimensionality of the latent space \citep{moran2022identifiable}.

In such settings, one may define $\mathcal T_\gamma(y_i)=\mathcal R (y_i,\hat\xi,\theta)$ where $\hat \xi$ minimizes the reconstruction loss for a given $\theta$. The resulting reconstruction error  majorizes the squared distance  to the set
\(
\mathcal Z_\theta = \{\mathcal R(\tilde y,\xi,\theta): \xi\in  X_\theta, \tilde y\in \mathcal{Y}\}.
\)
For an autoencoder, $\mathcal Z_\theta$ contains all possible reconstructions with intrinsic dimension at most $m$, with sparsity pattern determined by $\theta$. We see the distance is dominated by the reconstruction error via the relation
\[ \sum_{i=1}^n \text{dist}^2(y_i, \mathcal Z_\theta) = \sum_{i=1}^n
\min_{\xi_i \in  X_\theta, \tilde y_i\in \mathcal{Y}} \sum_{i=1}^n \|y_i- \mathcal R(\tilde y_i,\xi_i,\theta)\|^2_2. \]  Thus, reconstruction-based learning procedures can be interpreted as giving rise to divergence-to-set models or vice versa, in which the reconstruction map plays the role of a tractable surrogate projection. This perspective suggests that a range of popular machine learning procedures may be embedded into the divergence-to-set framework.

\subsection{Sampling from the predictive distribution}
In this framework for posterior inference, the distance-to-set model naturally also allows sampling from the posterior predictive distribution: given $(\theta, \sigma)$, we can independently draw samples
$
\mathcal L(y^*_{i} \mid \theta, \sigma) \propto \exp(- {\operatorname{dist}^2(y^*_{i},\mathcal Z_{\theta})}/{\sigma})
$
for $i=1,2,\ldots,n^*$. For low-dimensional settings,  rejection sampling may be feasible with a suitable proposal distribution. In higher dimensions, however, rejection sampling quickly becomes inefficient due to the curse of dimensionality. Instead, a simple MCMC sampler that targets the above predictive distribution can be constructed: starting with the identity
\(
  \exp\{-{\operatorname{dist}^2(y^*_{i},\mathcal Z_{\theta})}/{\sigma}\}= \int_0^{\infty} 1( \text{dist}^2(y^*_{i},\mathcal Z_{\theta}) \le u_i \sigma) \exp(-u_i) d u_i,
\)
we obtain the following MCMC sampler that alternates between:
\begin{itemize}
  \item Update $\tilde u_i \sim \text{Exp}(1)$, set $u_i = \tilde u_i + \text{dist}^2(y^*_{i},\mathcal Z_{\theta})/\sigma$.
  \item Update $y^*_{i} \sim \text{Unif}(\mathcal Z_{\theta}^{\sqrt{u_i\sigma}} \setminus \mathcal Z_{\theta})$.
\end{itemize}
To provide an illustration, Figure \ref{fig:l1_ball_y_scatter} shows samples drawn from the distance-to-$\ell_1$-ball model with different parameters. We note that in general, the second step can become difficult  due to the geometry of the shell $Z_{\theta}^{\sqrt{u_i}\sigma} \setminus \mathcal Z_{\theta}$. In such cases, one may replace with a ball-walk proposal as long as $Z_{\theta}^{\sqrt{u_i}\sigma} \setminus \mathcal Z_{\theta}$ is continuous: propose a new point $y^\prime_{i} \sim \text{Unif}[B(y^*_{i}, \alpha\sqrt{u_i\sigma})]$, where $\alpha>0$ controls the step size (we tune $\alpha$ so that the acceptance rate is around 0.4),  and accept the proposal if $0<\text{dist}(y^\prime_{i}, \mathcal Z_{\theta}) < \sqrt{u_i\sigma}$.

\hfill $\blacksquare$
\begin{figure}[H]
  \centering
  \begin{subfigure}[t]{0.3\textwidth}
      \centering
      \includegraphics[width=\textwidth]{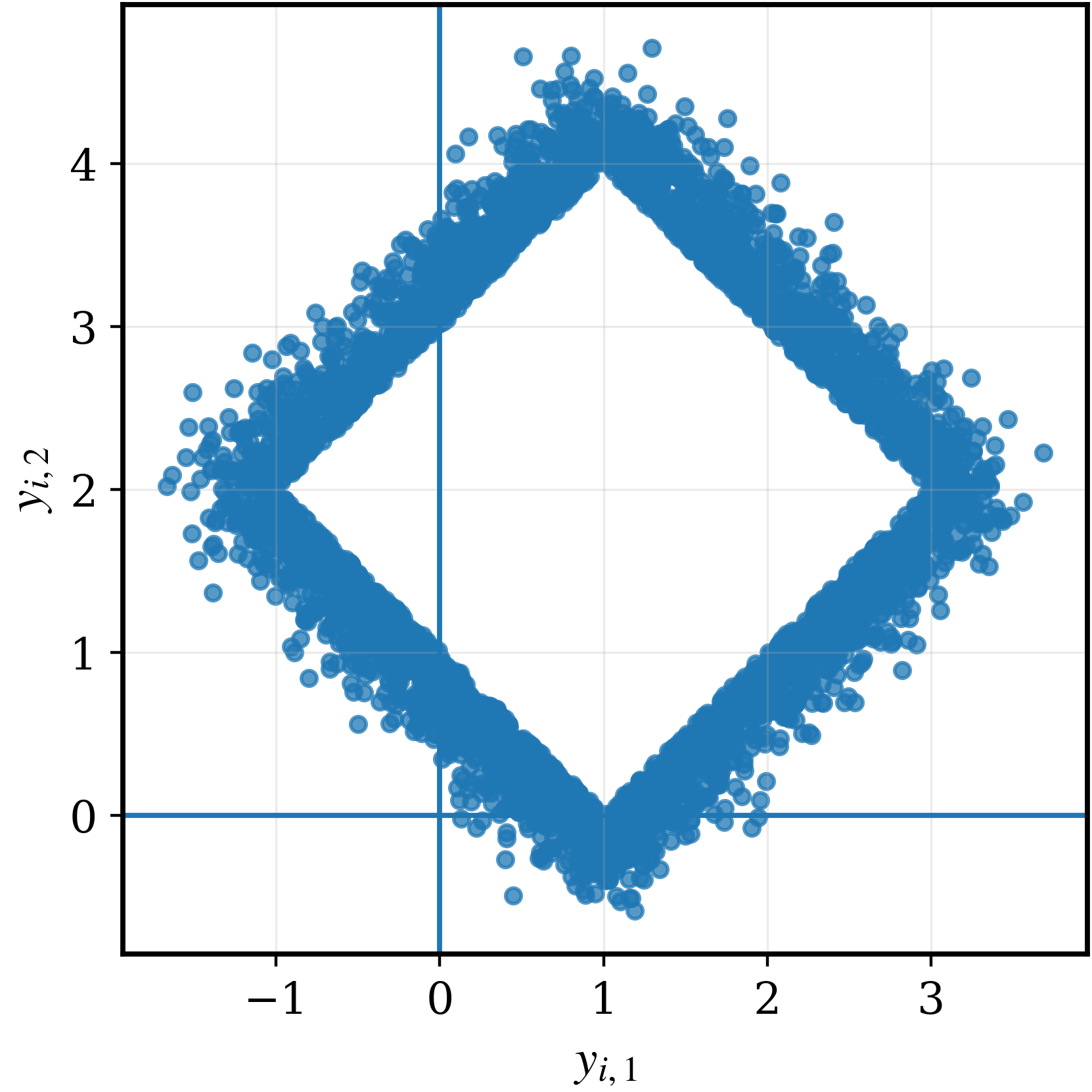}
      \caption{$\mu=(1,2), r=2,\sigma=0.1$.}
  \end{subfigure}
  \;\;
  \begin{subfigure}[t]{0.3\textwidth}
      \centering
      \includegraphics[width=\textwidth]{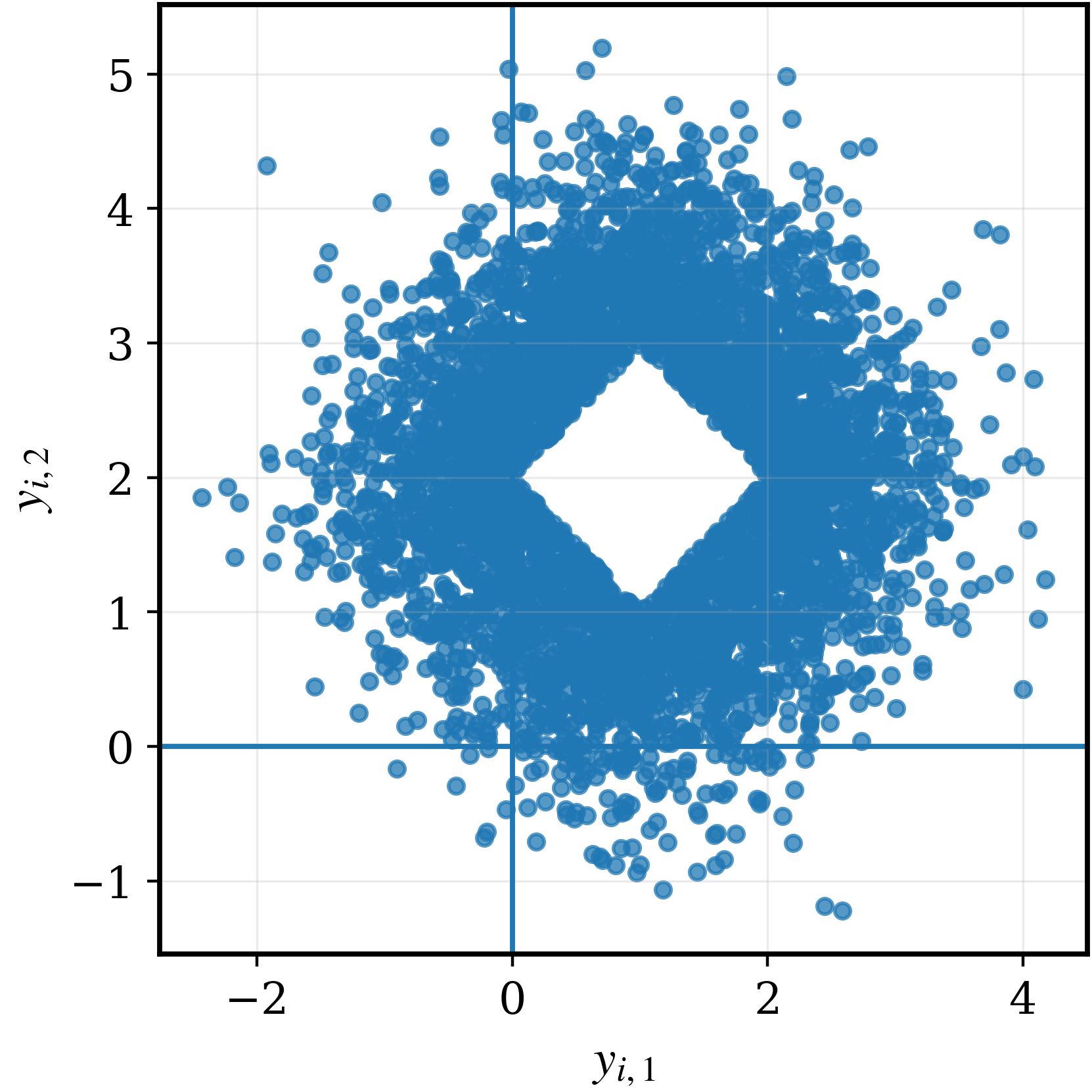}
      \caption{$\mu=(1,2), r=1,\sigma=1$.}
  \end{subfigure}
  \caption{Illustration of the distance-to-set model: samples drawn from two distance-to-$\ell_1$-ball models in $\mathbb{R}^2$, with different radii of the ball and different noise levels. \label{fig:l1_ball_y_scatter}}
\end{figure}

\subsection{Additional simulation studies}

\noindent\textbf{Multivariate mixed-effects model with sparse random effects}

We begin our empirical study by revisiting the multivariate mixed-effects model in \Cref{ex:multi_mix_effects}.  We simulate data from
$
y_i=\mu+\gamma_i+\epsilon_i,
$ for $i=1, \ldots, n$,
with $n=1000$ and $d=20$. We draw $\mu\sim \text{N}(0,100\, I_d)$,  noise terms $\epsilon_i\sim \text{N}(0, I_d)$, and sparse random effects $\gamma_i\in\mathbb R^d$. For each i, the first 5 coordinates are generated as $\gamma_{i,j}= s_{i,j}t_{i,j}$, where $s_{i,j}\in \{-1,1\}$ follows a Rademacher distribution, and $t_{i,j}\sim\text{Gamma}(2,1/5)$,  so that they are away from zero with high probability, and the remaining coordinates are set to zero. We then fit the distance-to-set model using
$
\mathcal Z_\theta=\big \{ \mu+\gamma_i:\ \sum_{j=1}^d |\gamma_{i,j}|= r \big\},\theta=(\mu,r).
$

We run Barker-MH for 100{,}000 iterations, discarding the first 50{,}000 iterations as burn-in, and thinning the chain by 10. The algorithm takes about 5 minutes to run on a standard laptop, and the average Metropolis acceptance rate was $0.34$. Figure~\ref{fig:multivariate_linear_regression} shows the trace plots of the posterior samples for $(\mu_1,\sigma,r)$, and the boxplots for the posterior of $\gamma_{\cdot,j}$ over $j=1,\ldots,20$ observations. We observe excellent mixing performance for all parameters, and the model accurately detects those non-zero random effects and shrinks the others to zero or close to zero.

\begin{figure}[H]
  \centering
  \begin{subfigure}[t]{0.48\textwidth}
      \centering
      \includegraphics[width=\textwidth]{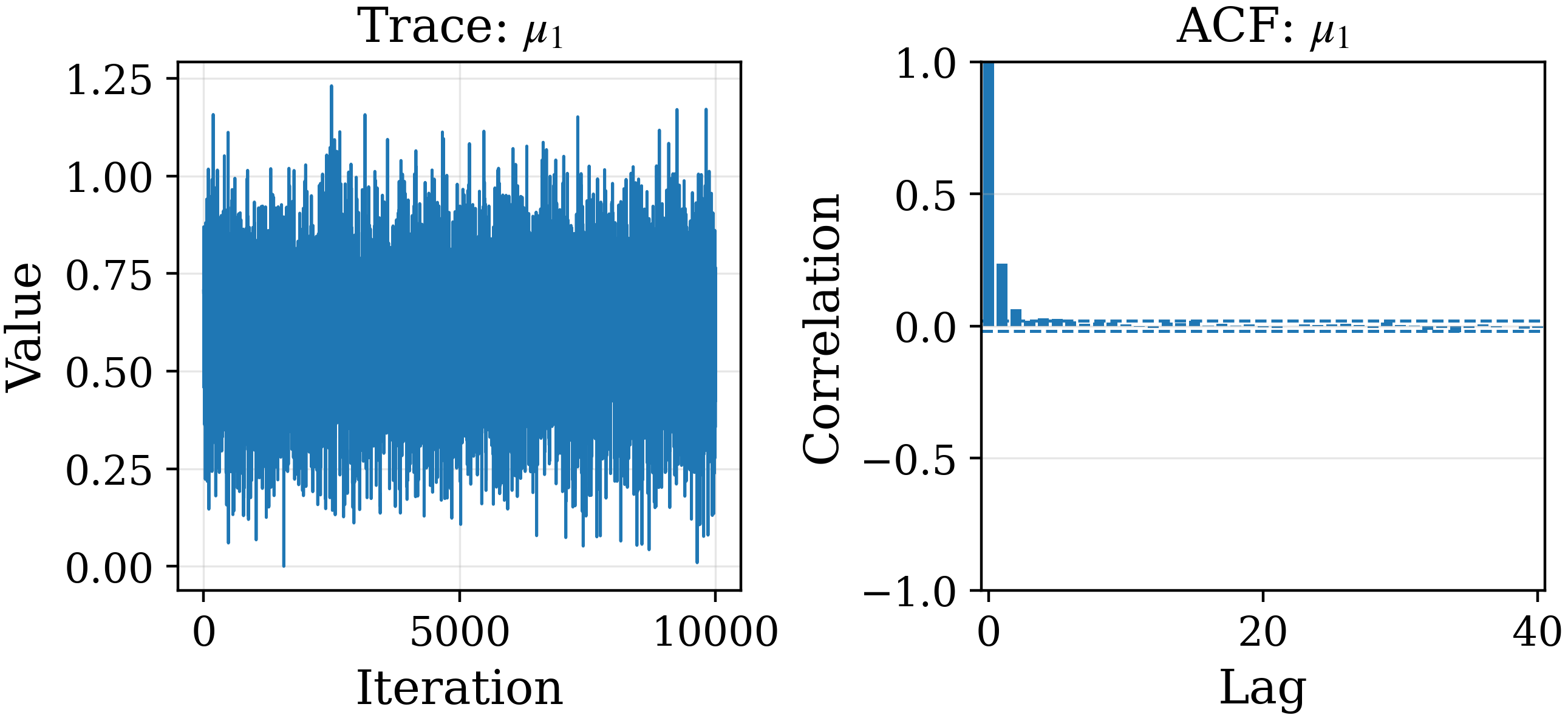}
      \caption{Fixed effect $\mu_1$.}
  \end{subfigure}
  \hfill
  \begin{subfigure}[t]{0.48\textwidth}
      \centering
      \includegraphics[width=\textwidth]{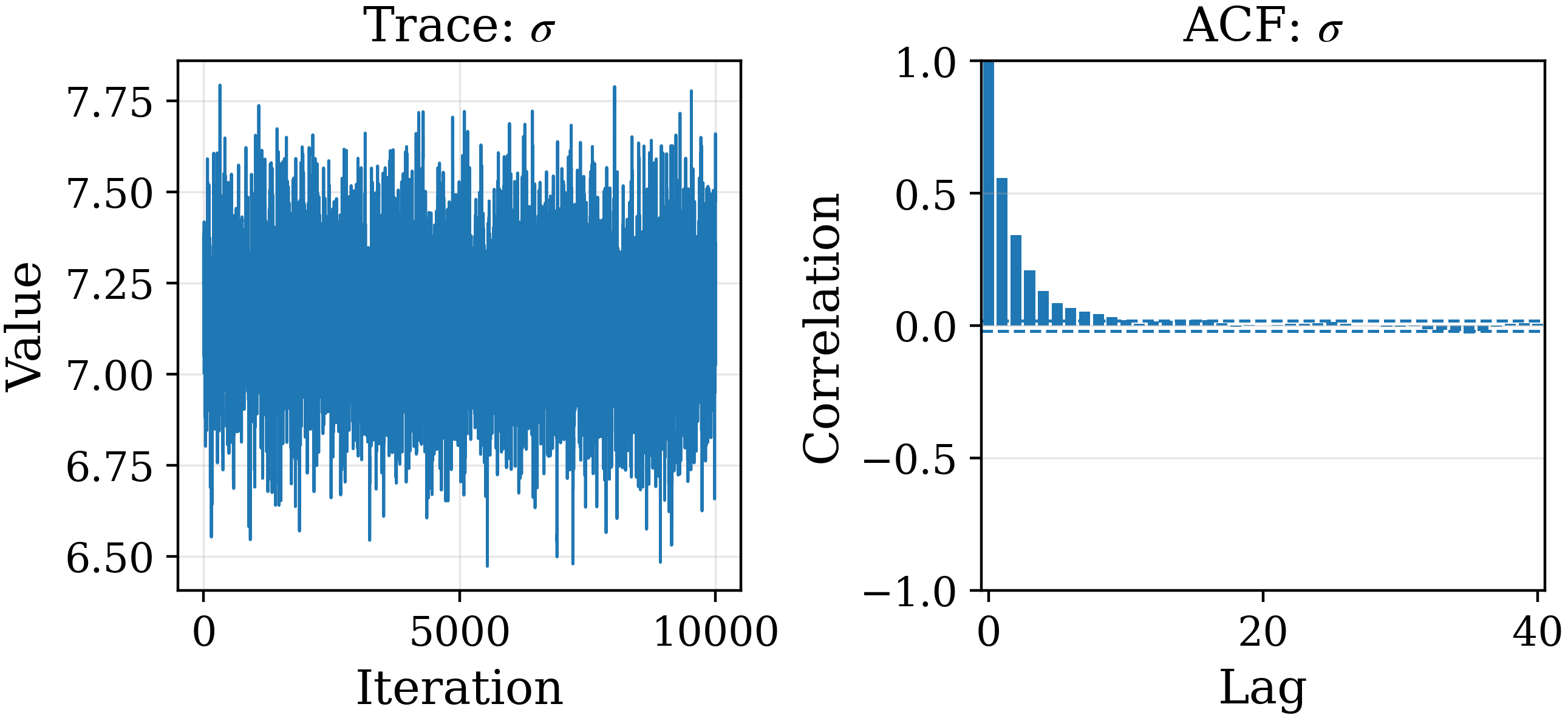}
      \caption{Distance scale $\sigma$.}
  \end{subfigure}
  \vskip\baselineskip
  \begin{subfigure}[t]{0.48\textwidth}
      \centering
      \includegraphics[width=\textwidth]{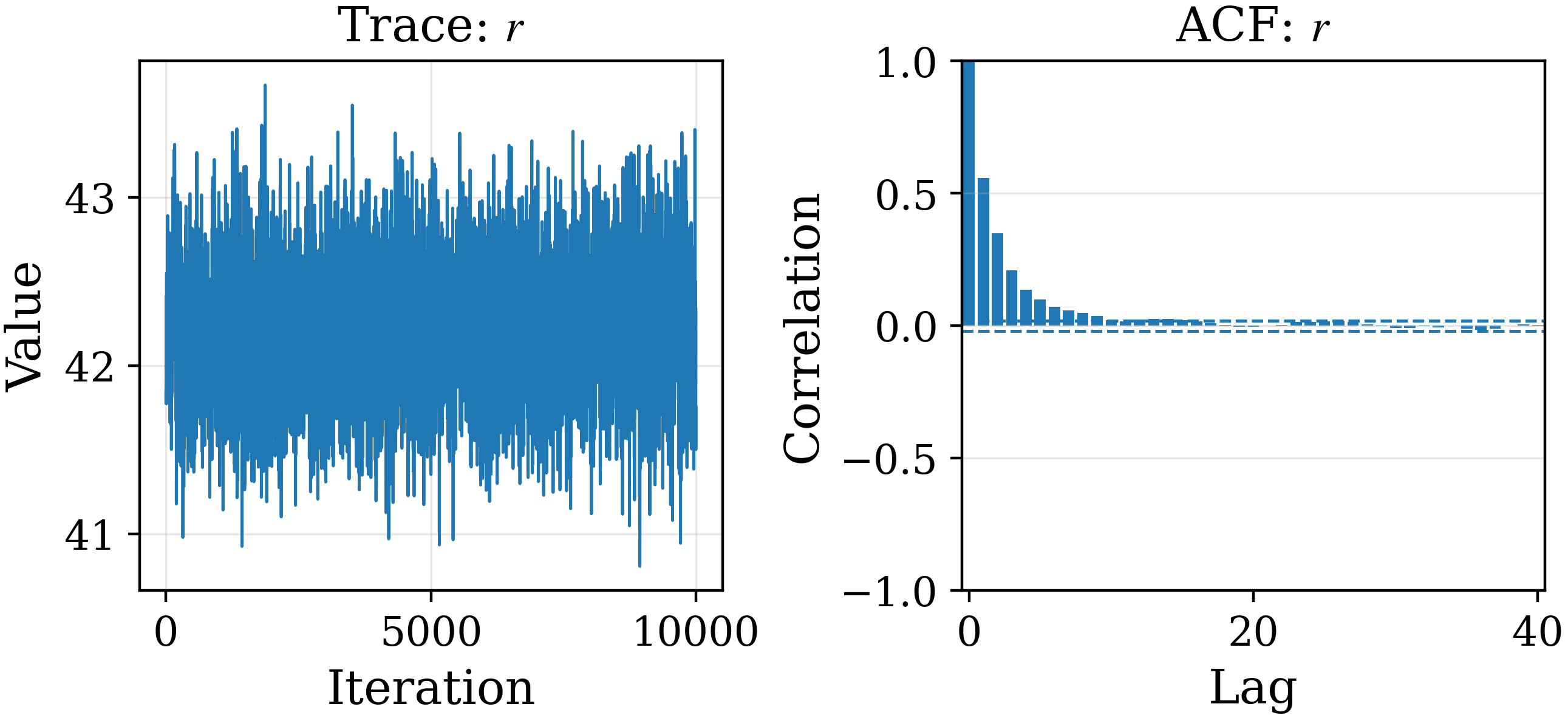}
      \caption{Set radius $r$.}
  \end{subfigure}
  \hfill
  \begin{subfigure}[t]{0.48\textwidth}
      \centering
      \includegraphics[width=\textwidth]{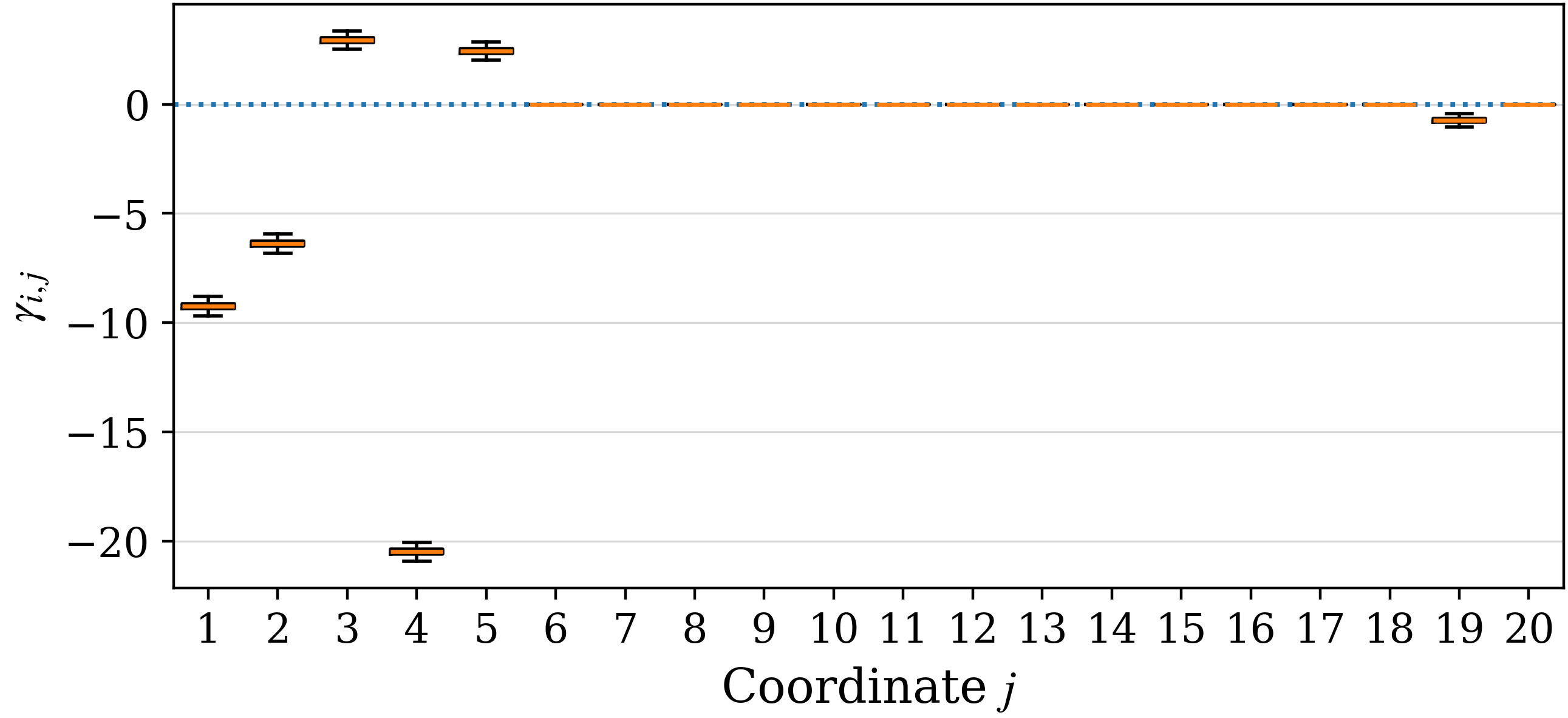}
      \caption{Posterior distribution of $\gamma_{i}$ for one subject.}
  \end{subfigure}
  \caption{Multivariate mixed-effects model with sparse random effect. \label{fig:multivariate_linear_regression}}
\end{figure}

For comparison, we also fit the latent variable model
\(
y_i=\mu+\gamma_i+\epsilon_i,\qquad \epsilon_i\sim \text{N}(0,\sigma I_d), \quad \gamma_{ij}\sim \mathrm{GDP}(b,3/b),
\)
where GDP denotes the generalized Double Pareto distribution \citep{armagan2013generalized} with scale parameter $b$ and shape parameter $3/b$. GDP is a continuous shrinkage distribution often used for Bayesian variable selection, and is suitable for modeling sparse random effects. We assign prior $b\sim \text{InvGaussian}(0.1 \sqrt{\sigma}, 0.1 \sqrt{\sigma})$ as the radius parameter in the distance-to-set model.

We find that in the latent variable model, the same Barker-MH algorithm  mixes much more slowly
even after 100,000 iterations. For a conservative comparison, we then reinitialize the chain with $(\mu,\sigma,\gamma)$ starting at their ground-truth values and $b=0.1$, and run the algorithm for 100,000 iterations. Figure~\ref{fig:latent_variable_linear_regression_gdp} shows that mixing performance is still poor. Such slow convergence and mixing performance is likely due to the high dimensionality for the parameter space of $(\mu,\sigma,b,\gamma)$, among which $\gamma\in \mathbb R^{1,000\times 20}$. We further empirically support this explanation and observe much improved mixing performance of the latent variable model in low dimensions with $n=50$ and $d=3$ (Figure~\ref{fig:latent_variable_linear_regression_gdp_low_d}). In contrast, the distance-to-set model succeeds as $(\mu,\sigma,r,\gamma)$ has low \textit{intrinsic} dimension since $\gamma$ is  determined once given $(\mu,\sigma,r)$.

\begin{remark}
For the latent variable model with a simple shrinkage prior on $\gamma_i$, it is possible to achieve better mixing with carefully designed samplers. For instance, when $\gamma_i \sim \text{Laplace}(0, r)$, the combination of Bayesian lasso data augmentation \citep{park2008bayesian} and two-block collapsed Gibbs sampling \citep{rajaratnam2019uncertainty} can yield excellent mixing performance. This study emphasizes that the distance-to-set model enables more general algorithms and prior choices, such as the Barker-MH algorithm and generalized Double Pareto prior, to work effectively without the need for customized algorithms.
\end{remark}

\begin{figure}[H]
  \centering
  \begin{subfigure}[t]{0.3\textwidth}
      \centering
      \includegraphics[width=\textwidth]{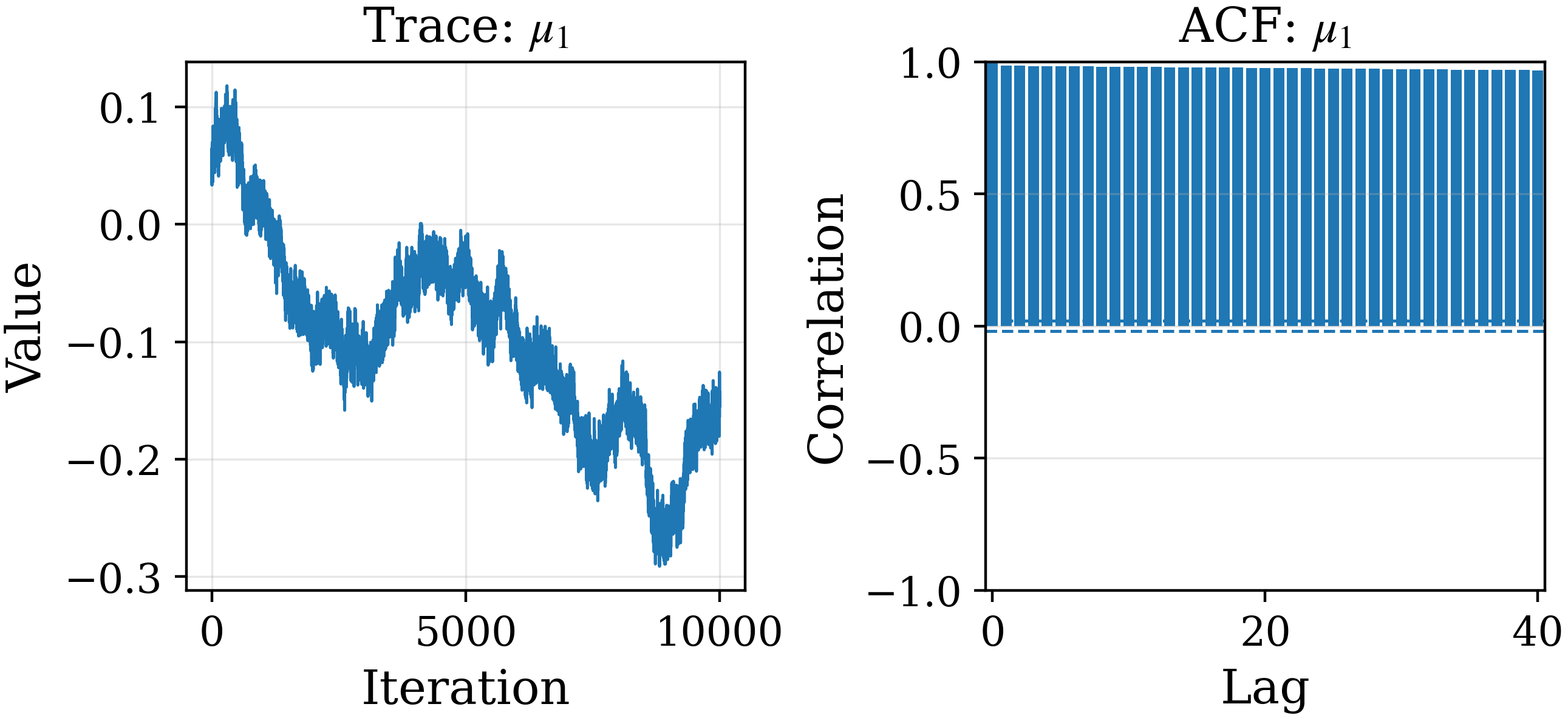}
      \caption{Fixed effect $\mu_1$.}
  \end{subfigure}
  \hfill
  \begin{subfigure}[t]{0.3\textwidth}
      \centering
      \includegraphics[width=\textwidth]{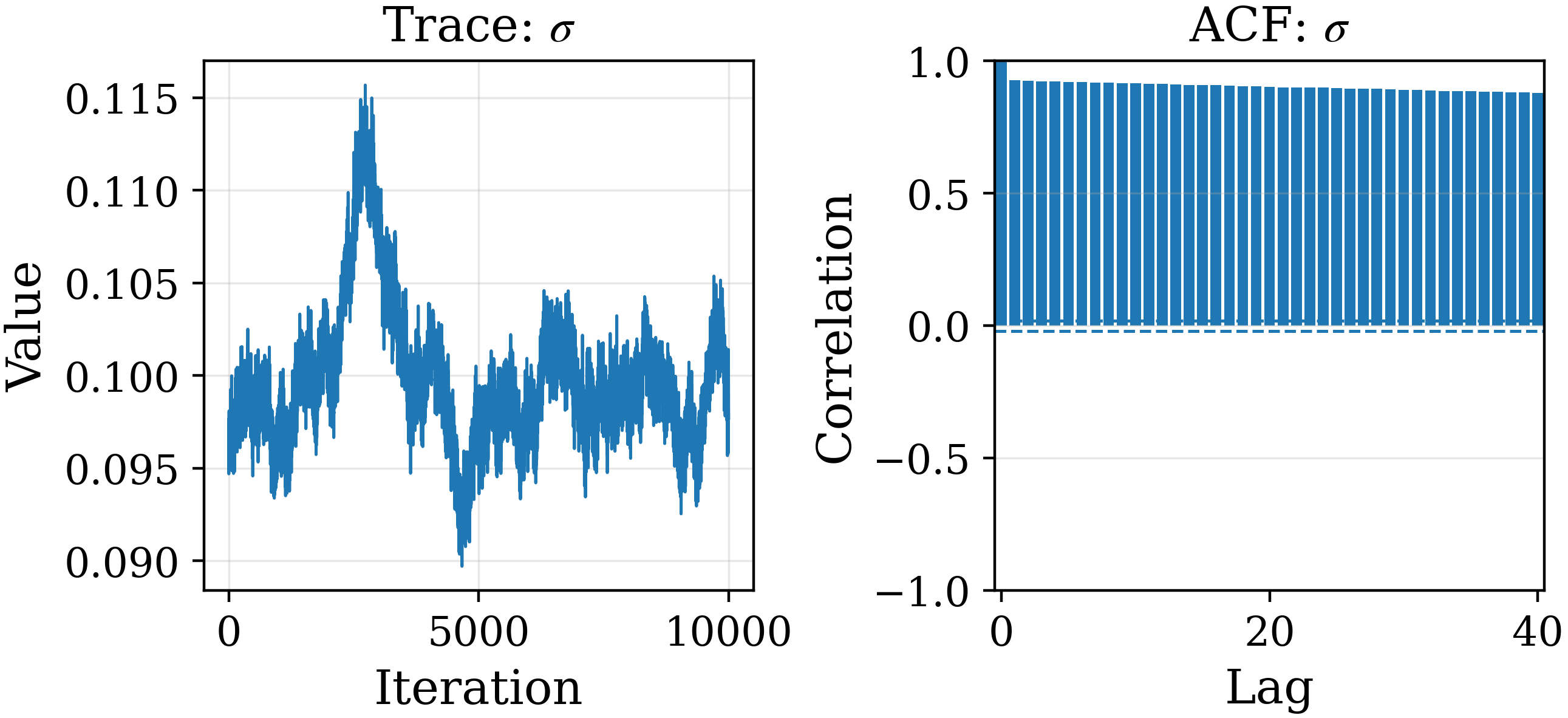}
      \caption{Distance scale $\sigma$.}
  \end{subfigure}
  \hfill
  \begin{subfigure}[t]{0.3\textwidth}
      \centering
      \includegraphics[width=\textwidth]{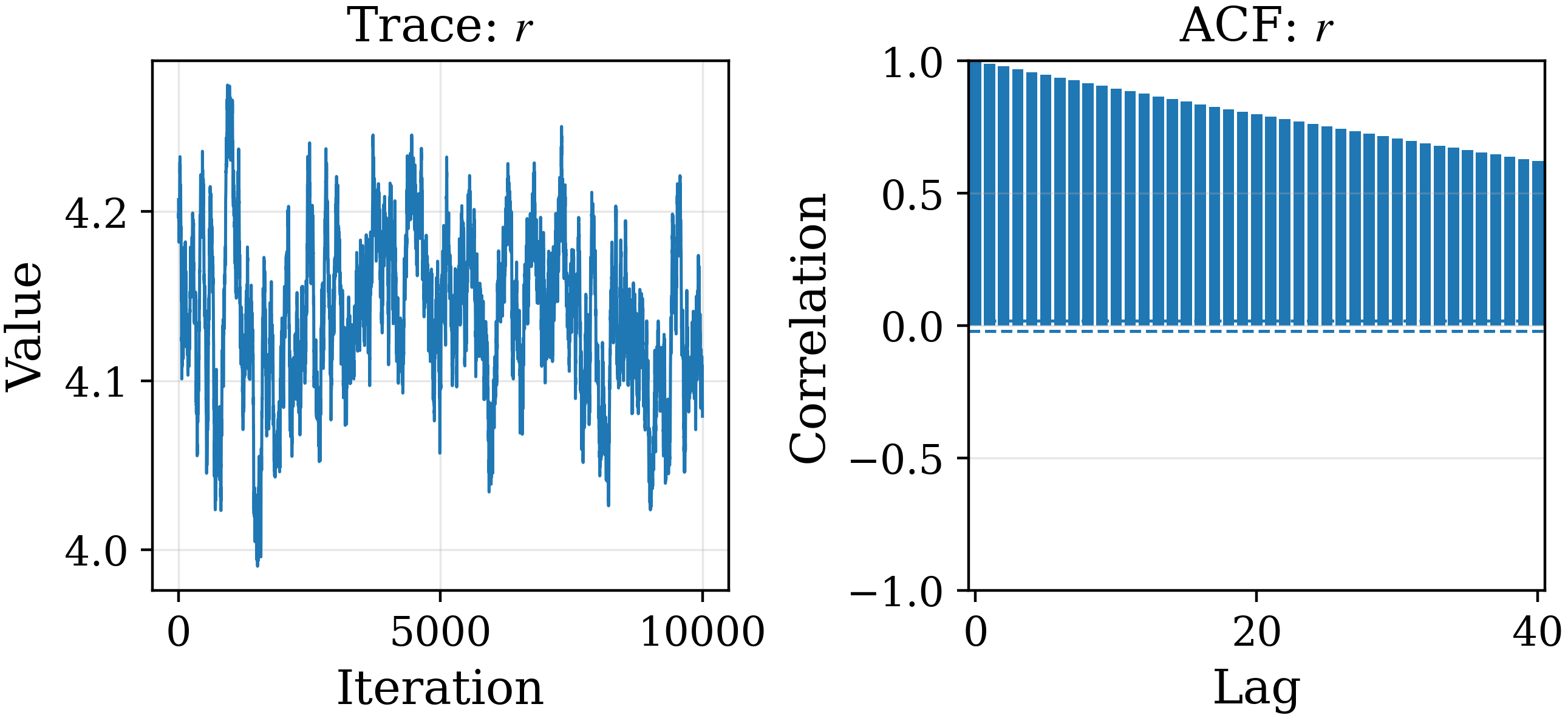}
      \caption{Set radius $r$.}
  \end{subfigure}
  \hfill
  \caption{Latent variable model with a continuous shrinkage prior on the random effects shows much worse mixing of Markov chains in the posterior estimation than the distance-to-set model. \label{fig:latent_variable_linear_regression_gdp}}
\end{figure}

  \begin{figure}[H]
    \centering
    \begin{subfigure}[t]{0.48\textwidth}
        \centering
        \includegraphics[width=\textwidth]{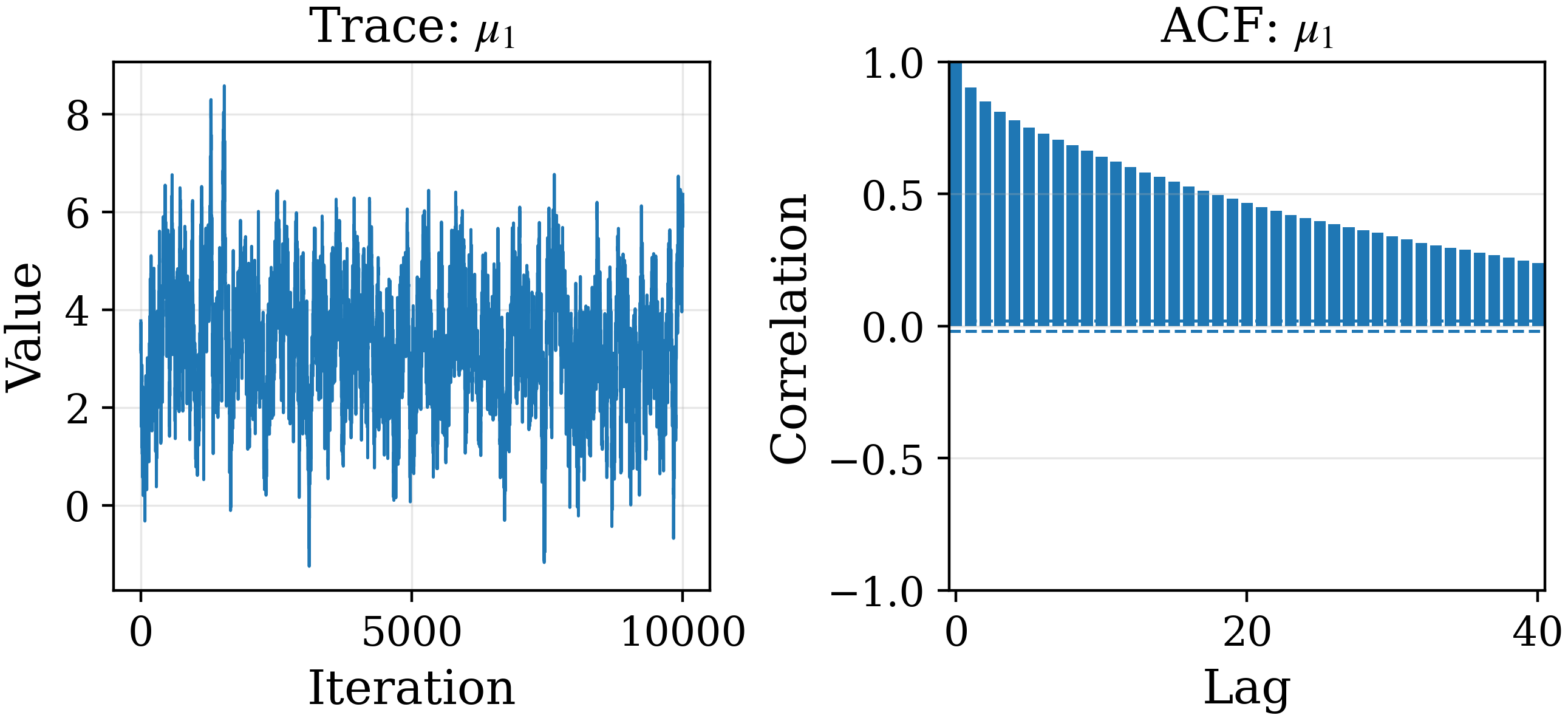}
        \caption{Fixed effect $\mu_1$.}
    \end{subfigure}
    \hfill
    \begin{subfigure}[t]{0.48\textwidth}
        \centering
        \includegraphics[width=\textwidth]{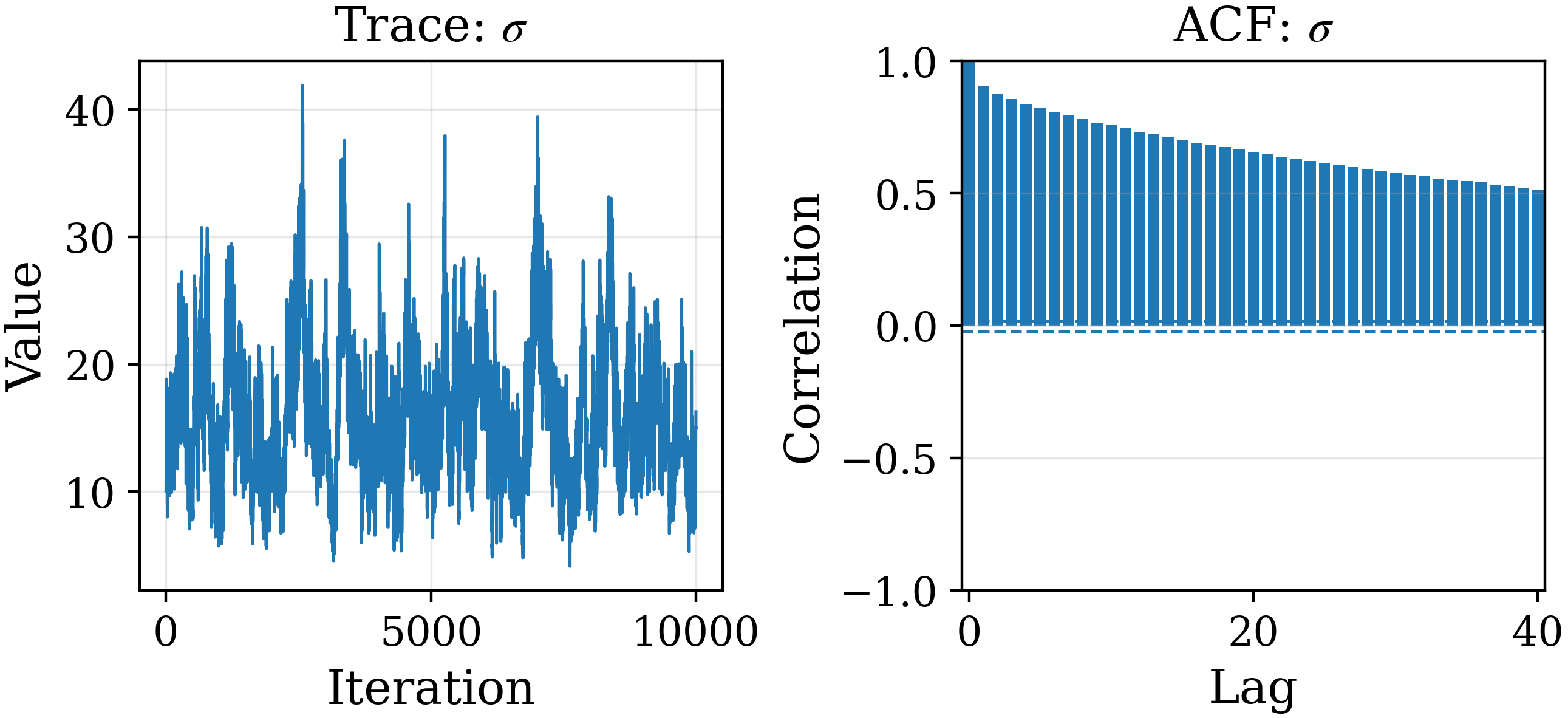}
        \caption{Distance scale $\sigma$.}
    \end{subfigure}
    \vskip\baselineskip
    \begin{subfigure}[t]{0.48\textwidth}
        \centering
        \includegraphics[width=\textwidth]{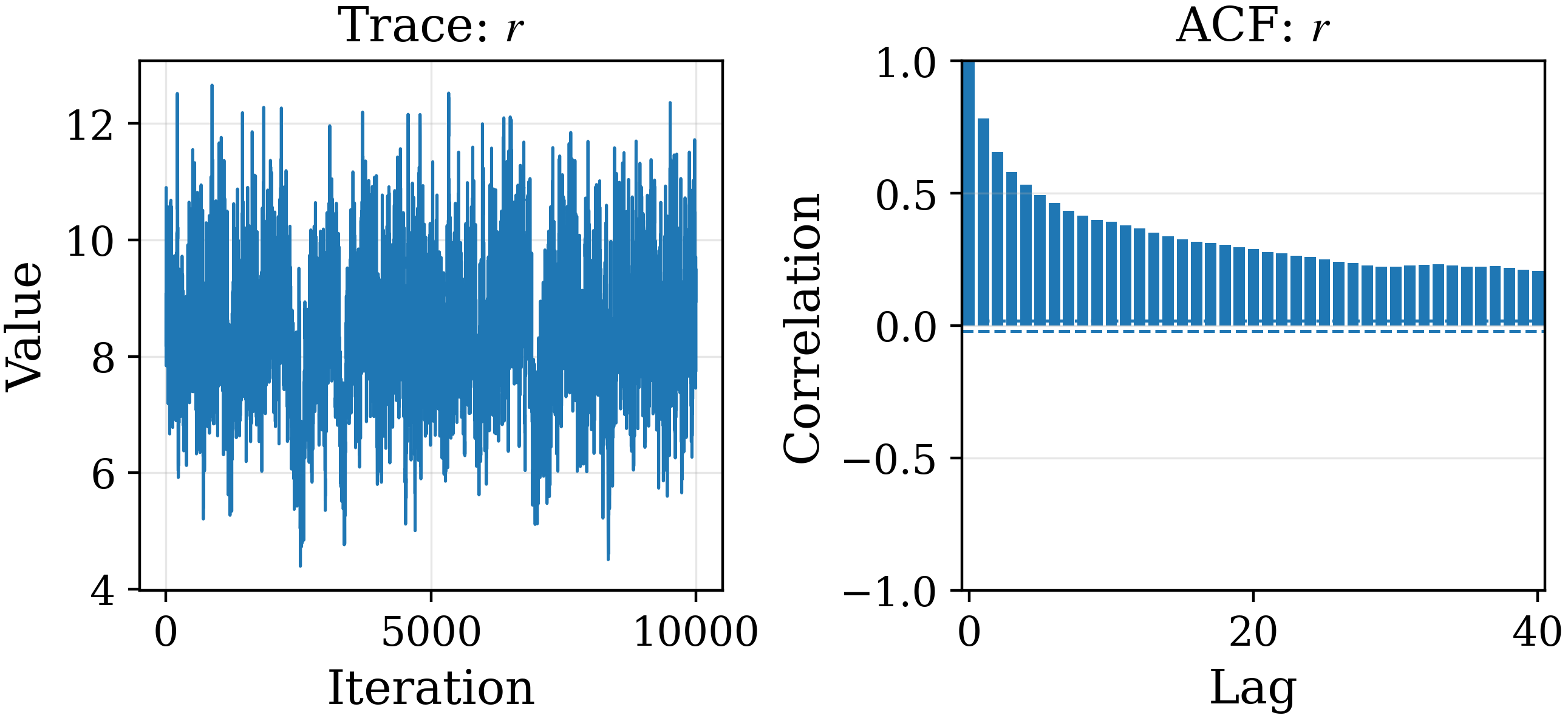}
        \caption{Set radius $r$.}
    \end{subfigure}
    \hfill
    \begin{subfigure}[t]{0.48\textwidth}
        \centering
        \includegraphics[width=\textwidth]{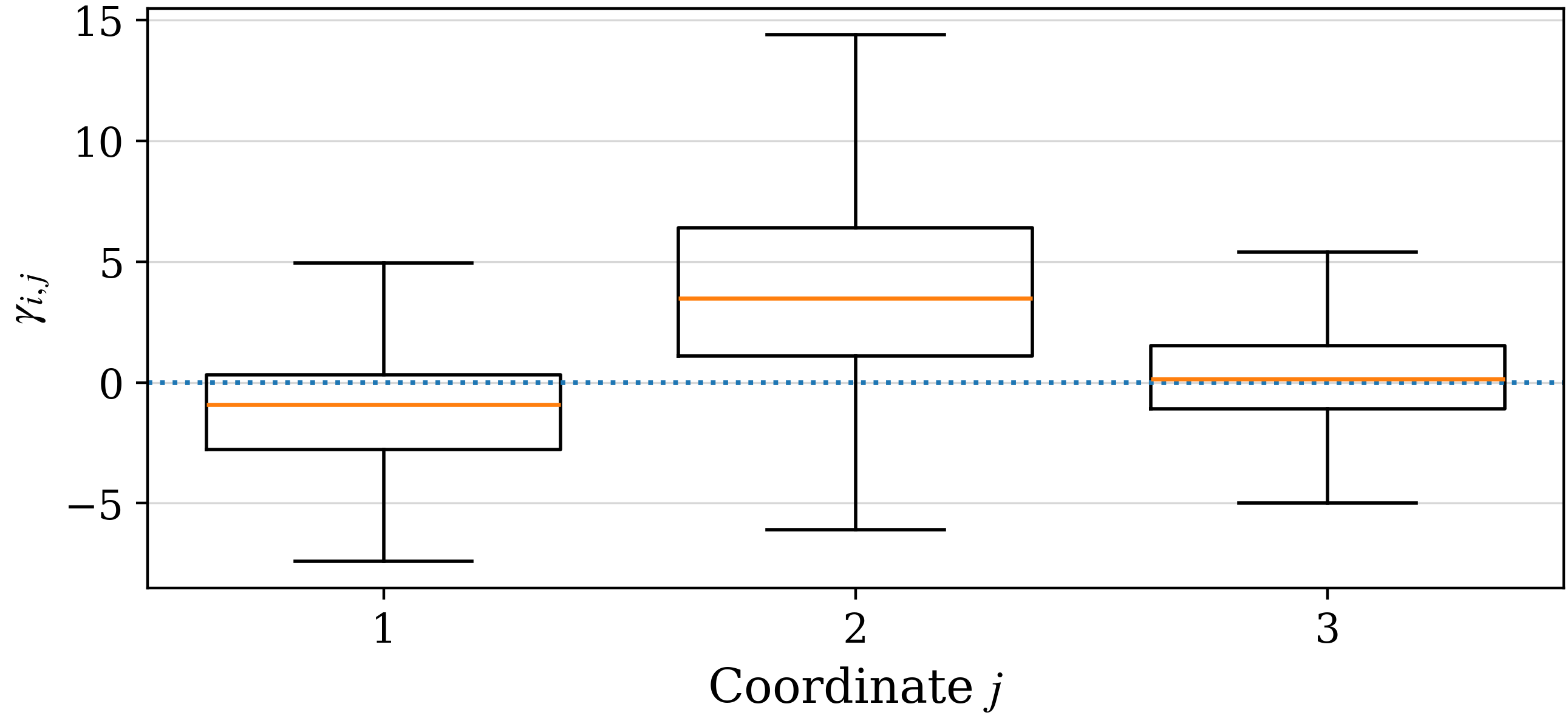}
        \caption{Posterior distribution of $\gamma_{i}$ for one subject.}
    \end{subfigure}
    \caption{Latent variable model with a continuous shrinkage prior on the random effects shows improved mixing of Markov chains in the posterior estimation when the dimension is low. \label{fig:latent_variable_linear_regression_gdp_low_d}}
  \end{figure}

\medskip
\noindent\textbf{Ellipsoid smoothing for monotone curve fitting}

We next consider monotone curve fitting using the ellipsoid smoothing model in \Cref{ex:ellipsoid}, with the affine constraint specified by $A$ as the $d \times (d-1)$ difference matrix and $b = 0$.  We randomly generate a monotone response curve from $d$ regularly spaced points $x_j$ in $[0,100]$, standardize the response, and add independent noise from $\text{N}(0,0.25)$. For simplicity, we use only one curve, and omit the index $i$.

\begin{figure}[H]
  \centering
  \begin{subfigure}[b]{0.5\textwidth}
    \centering
    \includegraphics[clip,width=\textwidth]{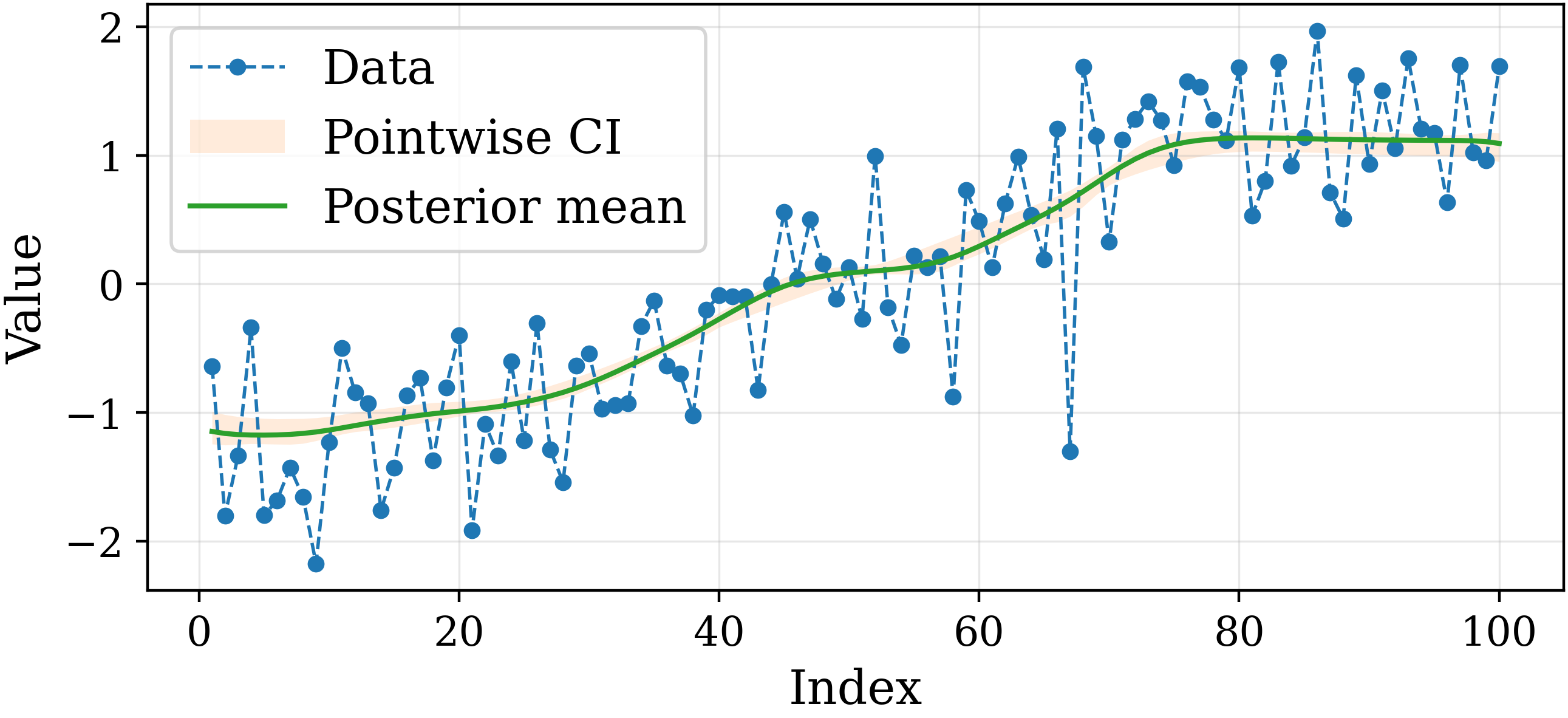}
    \caption{Estimated monotone curve using ellipsoid smoothing.}
\end{subfigure}\\
  \begin{subfigure}[b]{0.32\textwidth}
      \centering
      \includegraphics[trim=0 0 0 15, clip,width=\textwidth]{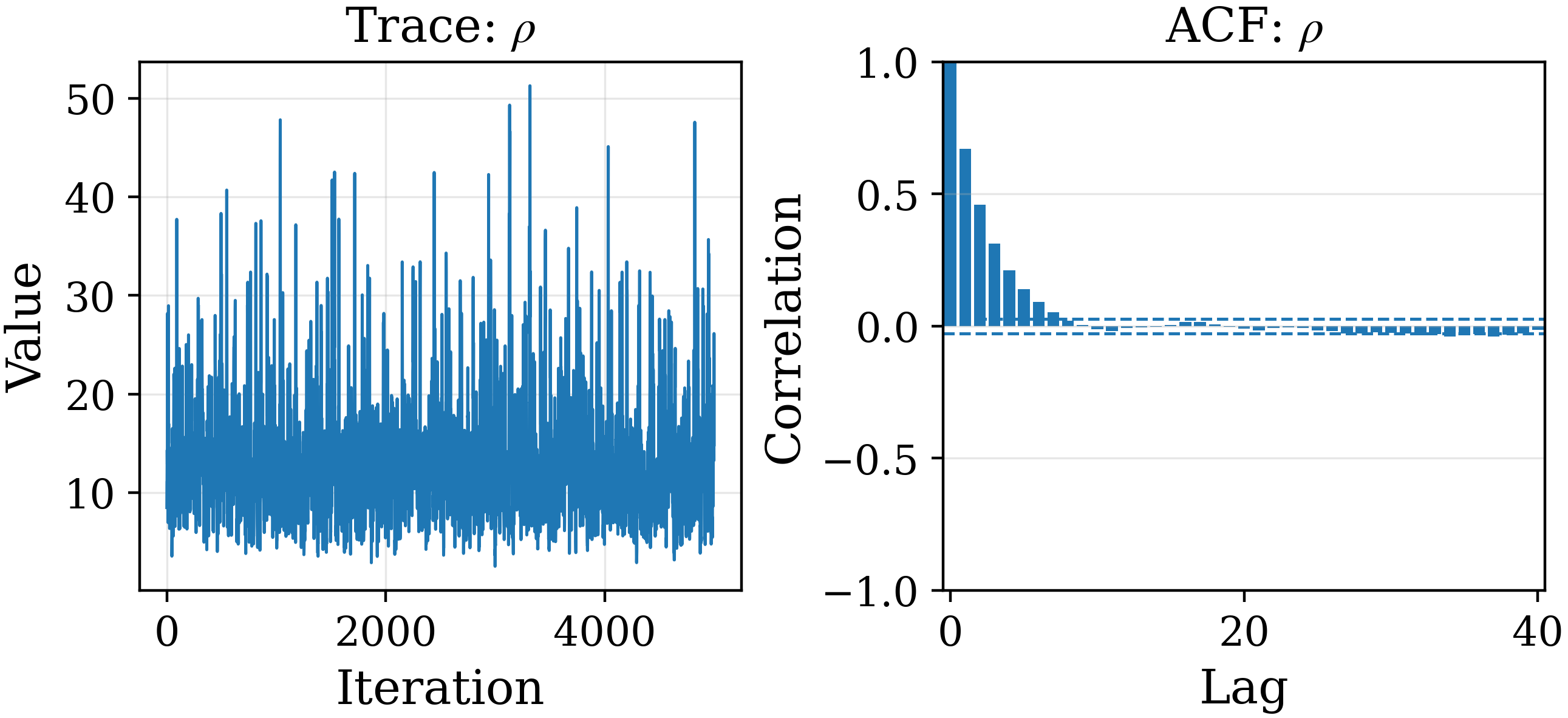}
      \caption{Bandwidth $\rho$ from ellipsoid smoothing.}
  \end{subfigure}
  \hfill
  \begin{subfigure}[b]{0.32\textwidth}
      \centering
      \includegraphics[trim=0 0 0 15, clip,width=\textwidth]{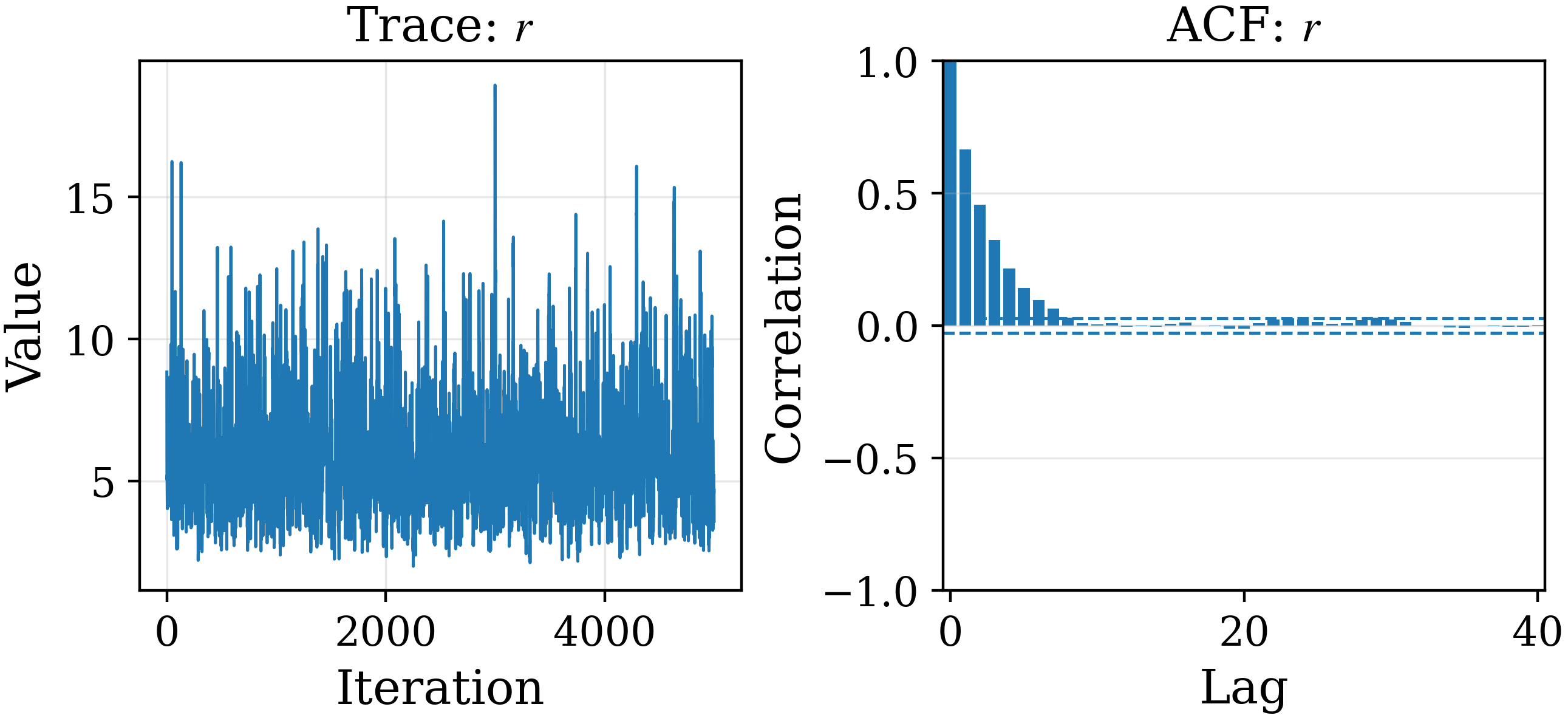}
      \caption{Ellipsoid radius $r$ from ellipsoid smoothing.}
  \end{subfigure}
  \hfill
  \begin{subfigure}[b]{0.32\textwidth}
      \centering
      \includegraphics[trim=0 0 0 15, clip,width=\textwidth]{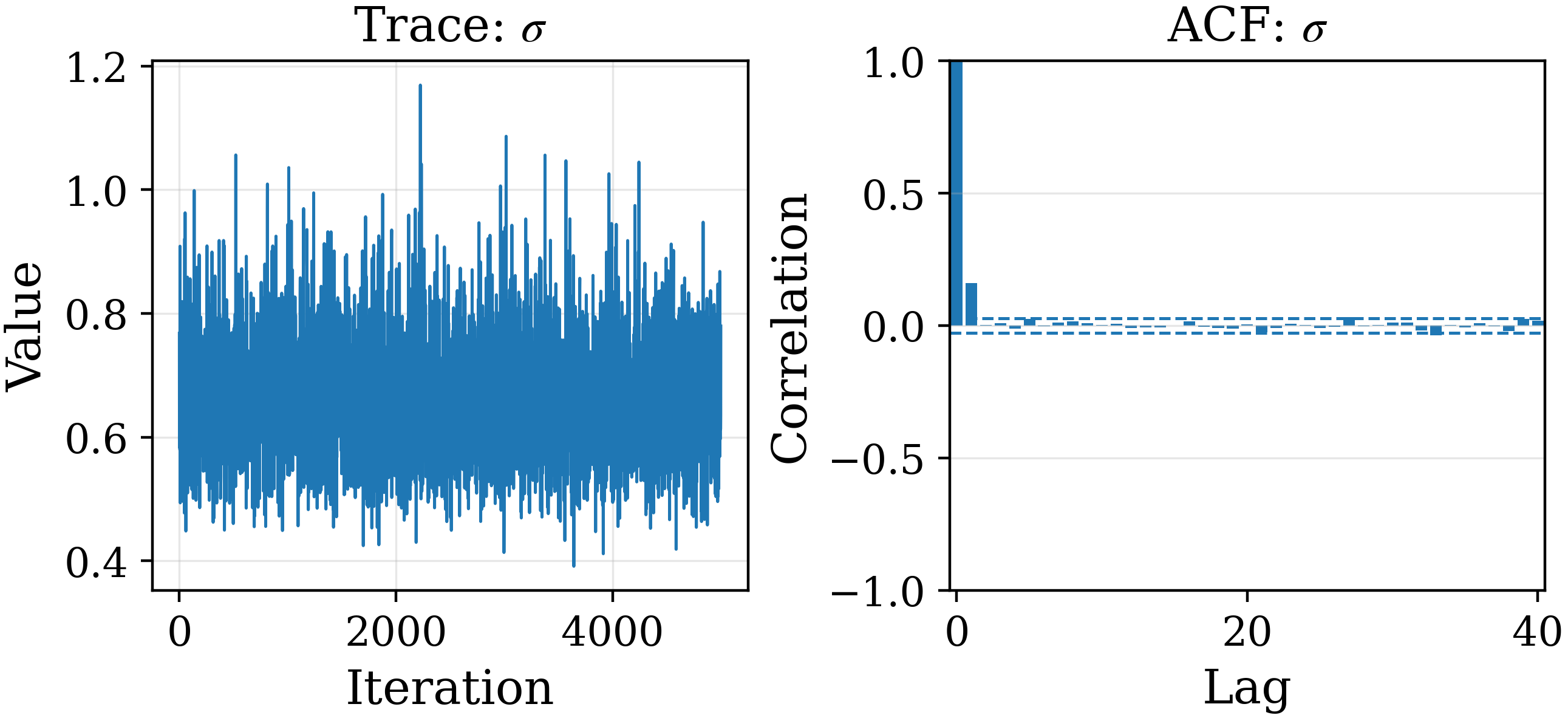}
      \caption{Distance scale $\sigma$ from ellipsoid smoothing.}
  \end{subfigure}\\
  \begin{subfigure}[b]{0.32\textwidth}
    \centering
    \includegraphics[trim=0 0 0 15, clip,width=\textwidth]{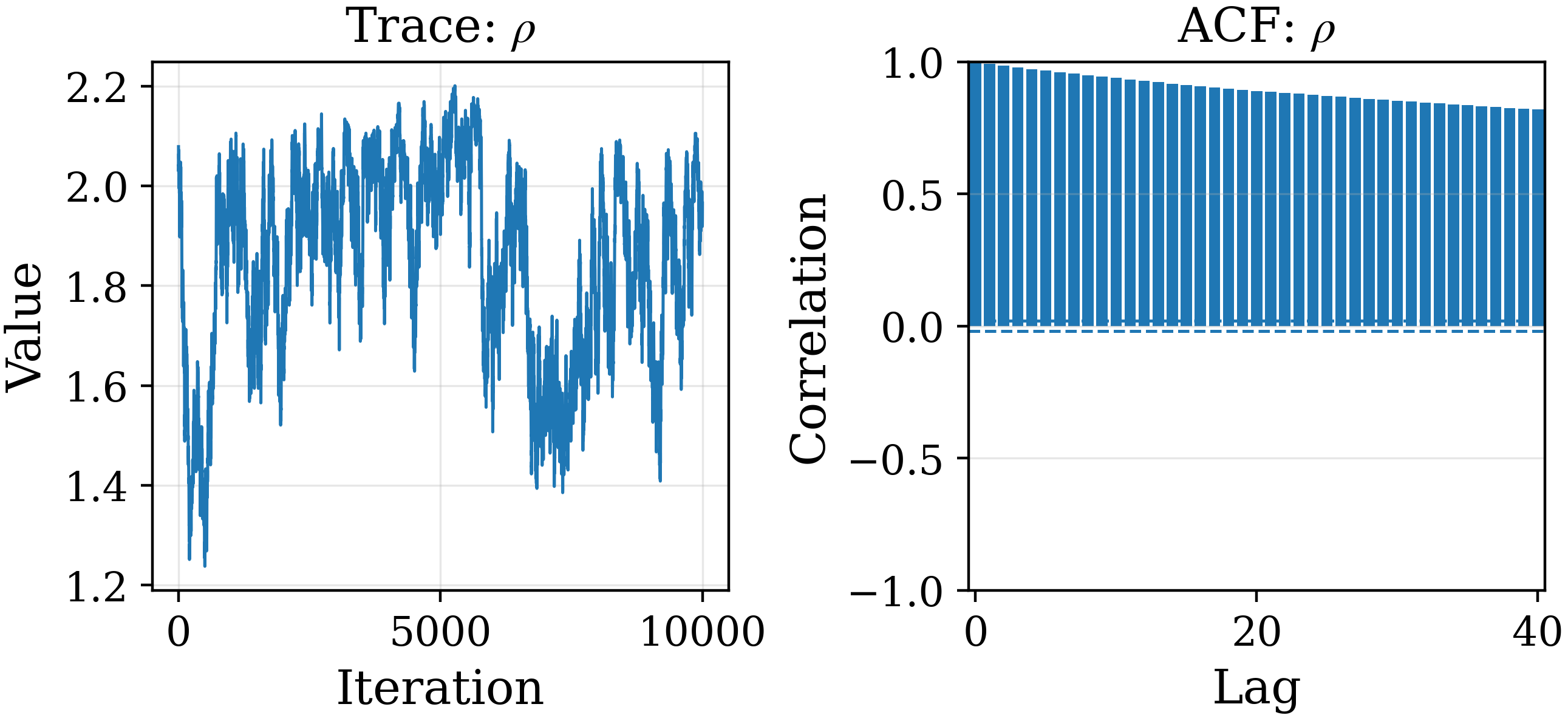}
    \caption{Bandwidth $\rho$ from monotone Gaussian process.}
\end{subfigure}
\hfill
\begin{subfigure}[b]{0.32\textwidth}
    \centering
    \includegraphics[trim=0 0 0 15, clip,width=\textwidth]{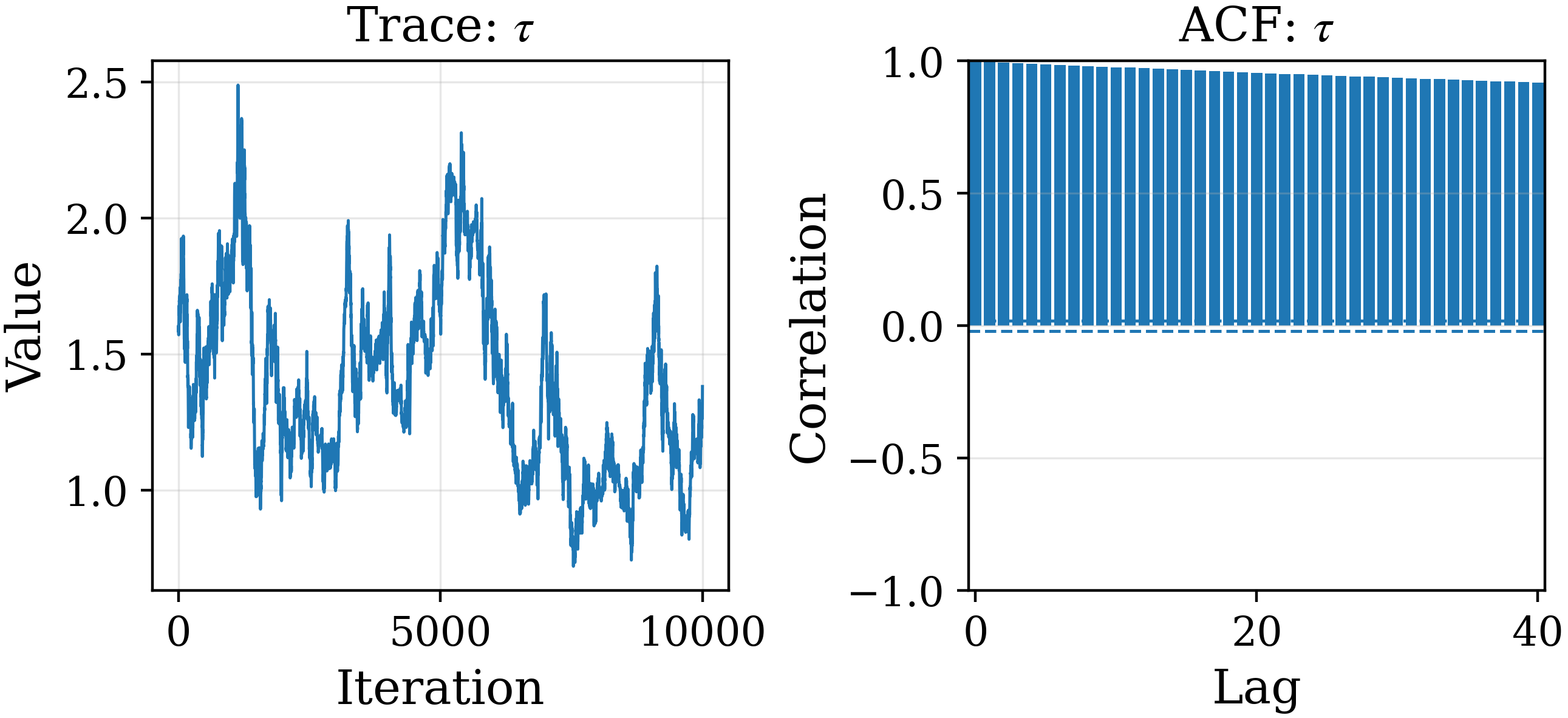}
    \caption{Covariance scale $\tau$ from monotone Gaussian process.}
\end{subfigure}
\hfill
\begin{subfigure}[b]{0.32\textwidth}
    \centering
    \includegraphics[trim=0 0 0 15, clip,width=\textwidth]{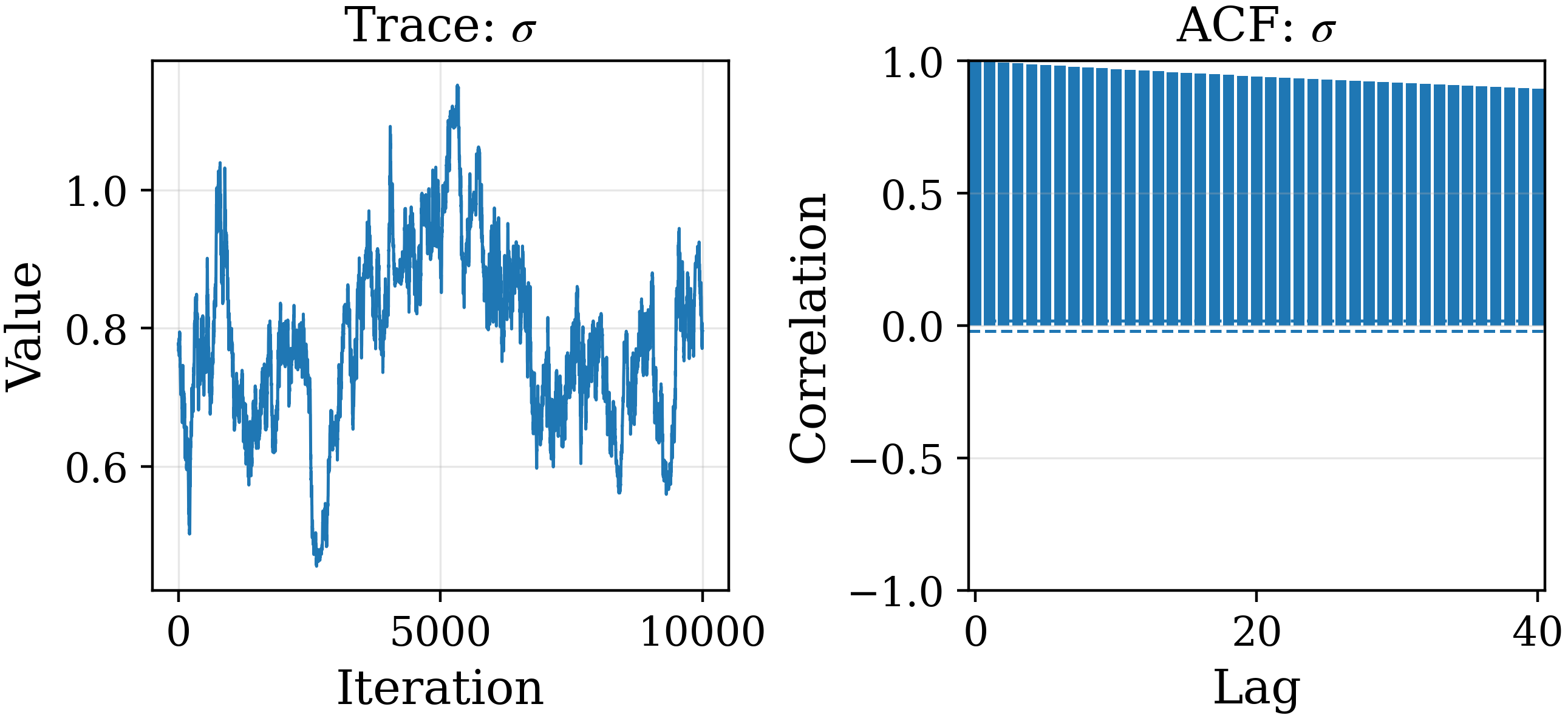}
    \caption{Noise scale $\sigma$ from monotone Gaussian process.}
\end{subfigure}
  \caption{Comparison on the computational performance of the distance-to-set model and the monotone Gaussian process model. \label{fig:ellipsoid_smoothing_comparison}}
\end{figure}

We fit the ellipsoid smoothing model  parameterized by a Gaussian kernel function $Q_{j,k}= \exp\{-(x_j-x_k)^2/\rho^2\}$, with prior $\rho\sim \text{InvGaussian}(1,1)$, with the same prior for $\sigma$ and $r$ as in \Cref{sec:prior_specification}. For comparison, we also fit a monotone Gaussian process model \citep{agrell2019gaussian}, which uses a projective Gaussian process prior underlying $y_i = f(x_i) + \epsilon_i$. Specifically, the monotonicity of $f(x_i)$ is enforced by  applying a monotone transformation to a zero-mean Gaussian process prior with covariance $\tau Q_{j,k}$, and $\epsilon_j \sim \text{N}(0, \sigma)$.
We parameterize $Q$ and assign priors for $\sigma$ and $\rho$ in the same way as in the ellipsoid smoothing model, and assign  $\tau\sim \text{InvGaussian}(1,1)$ for the covariance scale. Note that this approach differs from
the {\em projected posterior} method  \citep{lin2014bayesian}, which involves first sampling from the unconstrained Gaussian process posterior and then projecting the samples onto the monotone space. The projective prior model is more relevant
because it also requires constrained latent-function update during posterior sampling, hence a direct comparison with the distance-to-set model.

We estimate both posteriors using Barker-MH algorithm with 100,000 iterations, discard the first 50,000 iterations as burn-in, and thin the chain by 10. Both models took about 5 minutes to run on a laptop. As shown in Figure \ref{fig:ellipsoid_smoothing_comparison}, the distance-to-set model shows excellent mixing performance for all three parameters, while the monotone Gaussian process model
exhibits substantially slower mixing
due to the need for updating the relatively high-dimensional latent variable.